\documentclass[11pt,a4paper]{article}
\usepackage{amsmath, amssymb, amsthm, amsfonts, latexsym, graphicx, color}
\usepackage{authblk,url,multirow,mathtools,ulem}

\newlength{\actualtopmargin}
\newlength{\actualsidemargin}
  \theoremstyle{plain}
  \newtheorem{theorem}{Theorem}
  
  \newtheorem{lemma}[theorem]{Lemma}

  \theoremstyle{definition}
  \newtheorem{definition}[theorem]{Definition}

  \theoremstyle{remark}
  
  \theoremstyle{plain}
  \newtheorem*{theorem*}{Theorem}
  \newtheorem*{lemma*}{Lemma}
  \newtheorem*{corollary*}{Corollary}
  \newtheorem*{proposition*}{Proposition}
  \newtheorem*{claim*}{Claim}

\newcommand{\red}[1]{{\color{red}#1}}

\newcommand{\ii}{\mathbb{I}}

\newcommand{\poly}{\mathrm{poly}}

\newcommand{\oo}[1]{\Theta\left(#1\right)} 
\newcommand{\bigO}{O}
\newcommand{\bra}[1]{\langle #1 \vert}

\newcommand{\ket}[1]{\vert #1 \rangle}

\newcommand{\ketL}[1]{\left\vert #1 \right\rangle}
\newcommand{\ketbra}[1]{\vert #1 \rangle \langle #1 \vert}
\newcommand{\braket}[2]{\langle #1 \vert #2 \rangle}

\newcommand{\half}{\frac{1}{2}}

\newcommand{\Var}{\textrm{Var}}

\newcommand{\R}{\mathbb{R}}

\newcommand{\D}{\textrm{D}}
\newcommand{\M}{\textrm{M}}
\newcommand{\PF}{\textrm{PF}}
\newcommand{\PD}{\textrm{PD}}

\newcommand{\DirichletForm}{\mathcal{E}}
\newcommand{\II}{\mathbb{I}}
\newcommand{\Hir}{H_\textrm{ir}}
\newcommand{\SEQ}{\text{SEQ}}

\begin{document}
\title{\Large \textbf{
The pair-flip model: a very entangled translationally invariant spin chain
}}
\author[1]{Libor Caha\footnote{libor.caha@savba.sk}}
\author[1]{Daniel Nagaj\footnote{daniel.nagaj@savba.sk}}
\affil[1]{Research Center for Quantum Information, Institute of Physics, Slovak Academy of Sciences, D\'ubravsk\'a cesta 9, 845 11 Bratislava, Slovakia}

\maketitle
\vspace{-5mm}

 
\begin{abstract}
Investigating translationally invariant qudit spin chains with a low local dimension,
we ask what is the best possible tradeoff between 
the scaling of the entanglement entropy of a large block
and the inverse-polynomial scaling of the spectral gap.
Restricting ourselves to Hamiltonians with a ``rewriting'' interaction,
we find the {\em pair-flip model}, a family of spin chains
with nearest neighbor, translationally invariant, 
frustration-free interactions, with a very entangled ground state and an inverse-polynomial spectral gap. 
For a ground state in a particular invariant subspace, the entanglement entropy 
across a middle cut 
scales as log $n$ for qubits (it is equivalent to the XXX model), 
while for qutrits and higher, it scales as $\sqrt{n}$. 
Moreover, we conjecture that this particular ground state 
can be made unique by adding a small translationally-invariant 
perturbation that favors neighboring letter pairs, 
adding a small amount of frustration, while retaining the entropy scaling.

\end{abstract}

\section{Introduction}

Quantum many body systems described by simple, geometrically local quantum interactions in 1D, can display a wide array of behavior. One could use them to run universal quantum computation \cite{AGIK, Terhal}, or for simpler communication and transport tasks \cite{DanielBurgarth, Kay}, depending on the amount of control and design choices we have available. Sometimes, even though their description is simple and utility questionable, their classical simulation is very likely computationally difficult \cite{supremacyMontanaro, SchwarzSpeedup}.
On the other hand, the eigenstates or thermal states of quantum many body systems can exhibit interesting behavior -- error-correcting properties \cite{2017arXiv171004631B} or scaling of correlations \cite{scalingEisert} with respect to parameter changes.
Meanwhile, the difficulty of predicting and calculating their properties can range from classically easy \cite{2006quant.ph..2108B, 2016arXiv160208828A} through hard even with the use of quantum computers \cite{gscon, Quantum3SAT}, to uncomputable \cite{CPGWundecidable}, see for example the reviews \cite{QMAReviewBookatz, QHCReview}.

In this work, we focus on a simple subclass of this rich family of models -- translationally invariant spin chains with low local dimension and short-range interactions. Our goal is to discover and elucidate their properties, in particular the relationship of the scaling of the spectral gap with the size of the system and the possible behavior of ground state correlations (entanglement). 

The correlations of gapped systems in any dimension fall off exponentially \cite{expfalloff}. Moreover, gapped 1D systems always have ground states with only a little entanglement, obeying the {\em area law} \cite{1742-5468-2007-08-P08024}, which constrains the entropy of entanglement between two subsystems to a constant. This property makes them tractable on classical computers with heuristics like DMRG \cite{White92, White93}, or with the provably polynomial time algorithm of Landau et al. \cite{2013arXiv1307.5143L} or the recent rigorous renormalization group algorithms \cite{2016arXiv160208828A,2017arXiv170301994R}.
However, what happens when the area law does not hold, because the gap is not constant and closes as the system size grows? We are interested in the possible trade-off for the possibility of finding large quantum correlations in the ground state vs. having a large gap in the spectrum.
If the entanglement entropy of a block grows with the system size,
it might indicate the ground state is complex, interesting to analyze, and possibly hard to prepare.

On one hand, we know that systems with a very small (inverse exponential) gap can have lots of entanglement in their ground states. Simple examples (e.g. Latorre et al. \cite{1367-2630-12-11-113049}) have ground states that manifestly obey an entanglement {\em volume law}. However, we decide to investigate the more interesting, intermediate case, when the gap closes with growing system size only as an inverse polynomial. 

A system with a gap that closes as an inverse polynomial in the length of the chain reminds us of critical, frustrated systems. There, in 1D, we expect log-scaling of the entanglement entropy, as we have seen in the Ising and Heisenberg models \cite{LatorreRiera09, CalabreseCardy09}. Interestingly, we have found this type of behavior even for frustration-free systems \cite{CriticalityWithoutFrustration, MovaShor}. Even more surprisingly, when considering a larger local particle dimension, there can be exponentially more entanglement in the ground state: $\sqrt{N}$ vs. $\log{N}$. 
This motivates us to search for the best possible trade-off in the scaling of the entanglement entropy as a function of the inverse gap (which is itself inverse polynomial in $N$). In the process, we want to identify interesting qudit spin chains with surprisingly entangled ground states and no simple states with similarly low energy, reminiscent of the NLTS conjecture \cite{FreedmanHastingsNLTS,2013arXiv1309.7495A}.

Systems in 1D with an inverse polynomial spectral gap and highly entangled ground states naturally arise in the constructions of QMA-complete problems such as \cite{AGIK} and \cite{2013arXiv1312.1469H}. 
Aiming to encode universal quantum computation, 
these models have highly-tuned, position-dependent interactions. On the other hand, two models were designed specifically in order to exhibit high entanglement in the ground state. The construction of Irani \cite{2010JMP....51b2101I}
involves a 1D local Hamiltonian for particles with dimensions $d=21$ \cite{2010JMP....51b2101I} whose ground states correspond to sequentially creating EPR pairs and distributing them along a block of chain, with linear entanglement entropy. This entropy is related to the gap of the chain as $S\propto \Delta^{-1/12}$. The second model 
by Gottesman and Hastings 
\cite{1367-2630-12-2-025002} uses qudits with local dimension $d=9$, and ``distributes'' EPR pairs throughout blocks of qudits using a ``synchronized wheels'' construction. This construction reaches the scaling $S\propto \Delta^{-1/4}$. The drawbacks for these two models are mainly a large particle (qudit) dimension and a block-like structure that is not naturally translationally invariant -- fixing this requires a large increase in the local dimension. We thus ask: can one demonstrate a similar behavior of correlations in models that involve low-dimensional qudits? 

We restrict ourselves to models described by Hamiltonians $H=\sum_j H_{j,j+1}$, made from local terms $H_{j,j+1}$ acting nontrivially only on nearest neighbor particles on a line. We decide to focus on systems with a small local dimension $(d = 2,3,\dots)$. How much entanglement could the ground states of such chains possess, and how do their spectral gaps scale? We have seen that random models for $d=4$ can have very entangled ground states, but we don't know much about their gaps \cite{MovassaghRandom, gilyen_preparing_2017}. Our approach is then to construct particular models that we can understand.

For qubits ($d=2$), frustration-free Local Hamiltonians always have a ground state in the form of a product state of one qubit and two qubit states \cite{2006quant.ph..2108B, 2011PhRvA..84d2338J, 2011PhRvA..83e0301C}.
For qutrits ($d=3$), the AKLT model with a 2-dimensional MPS description of the ground state is well known. However, much more entanglement can hide in such frustration-free chains. In Bravyi et al. \cite{CriticalityWithoutFrustration}, we invented the translationally invariant, qutrit {\em Motzkin spin chain}, discussed in more detail in Section~\ref{sec:motzkinreview}. Its ground state is a uniform superposition over all Motzkin paths corresponding to well bracketed words with spaces (e.g. $\ket{\psi}=\frac{1}{\sqrt{9}}(\ket{0000}+\ket{0012}+\ket{0102}+\ket{1002}+\ket{1020}+ \ket{1200}+\ket{0120}+\ket{1212}+\ket{1122})$ for $N=4$, interpreting each ``$1$'' as a left bracket, ``$2$'' as a right bracket, and ``$0$'' as an empty space).
This model has Schmidt rank $\chi=\frac{N}{2}+1$, entanglement entropy $\Theta(\log N)$, and an inverse polynomial spectral gap $\Delta= N^{-\Theta(1)}$ (numerically, we find it is actually close to $N^{-3}$). 
Note that this ground state becomes unique only after we add boundary terms which raise the energy of badly-bracketed states. 

\subsection{Review of related work}
Since the Motzkin spin chain, a rich collection of new, related results about translationally invariant spin chains with very entangled ground states has appeared. First, Movassagh and Shor \cite{MovaShor} have extended and analyzed the Motzkin chain (bracket model) for particles with dimensions $d=2s+1$, viewing it as having $s$ different species/colors of brackets. For $d\geq 5$, this model has a ground state with an exponential Schmidt rank and entanglement entropy which grows as $S=\Theta(\sqrt{N})$. This model has an inverse polynomial spectral gap $\Delta= N^{-\Theta(1)}$. Furthermore, they found a way to get rid of the boundary terms and utilize a translationally invariant external field instead. However, this model is then no longer frustration-free.

Second, Salberger and Korepin \cite{FredkinSpinChain} 
discovered the {\em Fredkin spin chain} family of half-integer spin models, which we discuss in more detail in Section~\ref{sec:fredkinreview}.
Each of them is a next-nearest neighbor model (3-local interaction) whose ground state is a uniform superposition of all possible Dyck paths -- well bracketed words without spaces (e.g. the ground state for $N=6$ is $\ket{\psi}=\frac{1}{\sqrt{5}}(\ket{121212}+\ket{112212}+\ket{121122}+\ket{112122}+\ket{111222})$, interpreting the state ``$1$'' as a left bracket and the state ``$2$'' as a right bracket). This model has $\Theta(\log N)$ entanglement entropy for uncolored case (spin-$\half$) and $\Theta(\sqrt{N})$ for the case with several bracket colors (spin-$\frac{3}{2}$ and higher). Movassagh \cite{MovFSCGap} has then shown that the spectral gap scales as $\Delta=N^{-\Theta(1)}$.
Note that this gap scaling rules out the description of these models by relativistic conformal field theory
\cite{CFT}.

Another question about these models is the ground state magnetization and two point correlation functions, as well as von Neumann and R\'enyi entropies for any partition. 
These calculations were done for the Motzkin chain by 
Movassagh \cite{MovassaghCorrelations} and later extended for the colored case and the Fredkin chain in \cite{ViolationOfCD}.
Next, Dell'Anna et al. \cite{ViolationOfCD} showed violation of the cluster decomposition property (CDP) for colored Motzkin and Fredkin spin chains, in contrast with the uncolored case with a light-cone-like propagation.
Meanwhile, Brand\~{a}o et al. \cite{2017arXiv171004631B} observed that the eigenstates of translationally invariant spin chains can have approximate quantum error correcting code properties -- with the Motzkin spin chain as one of their examples.

Finally, there is a collection of recent results about deformed variants of the Motzkin and Fredkin chains. Zhang, Ahmadain and Klich \cite{2016arXiv160607795Z} provided a parametrized version of the Motzkin spin chain Hamiltonian with an area weighted deformation. The introduced parameter $t$ leads to a ground state which is a superposition of Motzkin paths weighted by $t$ to the area under the particular path. There is a phase transition in entanglement entropy scaling going from bounded to logarithmic and back to bounded, for spin-1 chains, and from bounded through square root to extensive scaling for chains with spin $\geq 2$. A similar result was found for Fredkin spin chains by Salberger et al. \cite{DeformedFredkinSpinChain}. Levin and Movassagh \cite{LevineMovassagh} showed an inverse exponential upper bound on the energy gap of the area weighted Motzkin spin chain with extensive entropy (similar result also appeared as an example in \cite{2017arXiv170310133C}). The same result was achieved also for weighted Fredkin spin chain by Udagawa and Katasura \cite{UdagawaKatasuraGapMagnDefFSC}. Herein, the authors also studied magnetization and von Neumann and R\'enyi entropy and entanglement spectrum for any bipartition. Zhang and Klich \cite{2017arXiv170203581Z} provided a multi-parameter deformation of Fredkin spin chain with the deformation depending on the position in the chain (non translationally invariant) and a condition under which its ground state remains frustration free. When one takes equal parameters along the chain, this deformation collapses to the weighted Fredkin spin chain \cite{DeformedFredkinSpinChain}. They also provided a different calculation of the entanglement entropy and an inverse exponential upper-bound on the energy gap for $t>1$. Sugino and Padmanabhan then \cite{2017arXiv171010426S} invented a different parametrization of the uncolored Motzkin chain based on decorating the Motzkin paths with elements of symmetric inverse semigroups and demonstrated systems with phase transitions of the entanglement entropy scaling from logarithmic to bounded, as well as from logarithmic to square root. Later Sugino, Padmanabhan, and Korepin \cite{2018arXiv180400978P} used an analogous parametrization and demonstrated a similar behavior for spin-$\half$ Fredkin chain, and promised the parametrization also for the colored case \cite{PadmanabhanSuginoKorepinInPreparation}.
Barbiero, et al. \cite{PhysRevB.96.180404} found that the area weighted Motzkin spin chains with deformation parameter $t<1$ exhibit gapped Haldane topological order ($SU(2)$ for the uncolored and $SO(5)$ for the colored case, occurring in the $s=1$ and $s=2$ AKLT model, repsectively). Their numerical calculations signal a Berezinskii-Kosterlitz-Thouless phase transition at $t=1$.
Chen, Fradkin and Witczak-Krempa  \cite{1751-8121-50-46-464002, PhysRevB.96.180402}, used DMRG to find multiple dynamical exponents (the $1/\poly$ degree of the energy gaps) of low lying excitations for the uncolored Motzkin and Fredkin spin chains, indicating that these models have multiple dynamics. They express the ground state in the continuum limit and show that the mutual information between two disjoint intervals deep inside the bulk tends to zero, which is in contrast with $1+1$ dimensional CFT systems and found an emerging $1+0$ dimensional conformal-type symmetry, in the sense of conformal quantum mechanics. However, these models are not described by relativistic CFT as the dynamical exponent is lower-bounded by 2 \cite{MovaShor,MovFSCGap}. Very recently, Adhikari \& Beach \cite{AdhikariBeachDeformingAwayFromFF} proposed and studied a generalized spin chain model that interpolates between the ferro and antiferromagnetic quantum Heisenberg models and includes the uncolored Fredkin spin chain as a special tuning point. They investigated its phase diagram numerically and semi-analytically and found multiple phase transitions, also observing that the uncolored Fredkin spin chain turns out to be unstable with respect to antiferromagnetic frustration.

\subsection{Our contribution and organization of the paper}
In this paper, we present a family of translationally invariant {\em pair-flip spin chains} that goes beyond the Motzkin and Fredkin spin chains in terms of the entanglement entropy vs. gap trade-off, with only nearest-neighbor interactions and smaller local dimension. There is one drawback -- to break the inherent degeneracy of the ground state, we need to add a small term that effectively counts the average number of adjacent identical particle pairs on the chain, which means the model is no longer frustration free.
However, as an upshot of this, we conjecture that our model retains its properties also for periodic boundary conditions. 

First, we present an introduction to rewriting Hamiltonians in Section~\ref{sec:results},
and review our inspiration -- the Motzkin chain (Section~\ref{sec:motzkinreview}) and the Fredkin chain (Section~\ref{sec:fredkinreview}), providing our own closed-form combinatorial results and precise asymptotic scaling, with proofs in Appendix~\ref{sec:technical}.
Then in Section~\ref{sec:pf} we present our new pair-flip (PF) model family and its entanglement and gap properties, with proofs for the $d=2$ qubit model in Section~\ref{sec:pf2} and the $d\geq 3$ family in Section~\ref{sec:pf3}.

\section{Our results: very entangled spin chains from rewriting Hamiltonians}
\label{sec:results}

In this paper, we propose and analyze the pair-flip (PF) model: a family of nearest-neighbor translationally invariant 1D models with particles of dimension $d\geq 2$. Each model has a unique, highly entangled ground state. For $d=2$, the entanglement entropy scales as $\log N$, reminiscent of critical systems in 1D, and similar to the $d=3$ Motzkin spin chain.
Moreover, this model has a mapping to the XXZ model and provides another viewpoint for its analysis. However, things get much more interesting for local dimension $d\geq 3$. There, we prove the ground states have an exponential Schmidt rank and square root entanglement entropy, while keeping an inverse polynomial spectral gap. 

Our main result is achieving this $\sqrt{N}$ entanglement entropy scaling in the unique ground state of a qutrit ($d=3$) translationally invariant Hamiltonian, with a better (numerical) gap scaling and local dimension smaller than the Motzkin spin chain ($d\geq 5$), and the next-nearest-neighbor Fredkin spin chain ($d\geq 4$). 
We prove the basic properties of the PF model using a combination of analytic combinatorics, random walks, and graph theory. We also investigate some of the model's properties numerically. 

The interactions in our models come from the class of {\em rewriting Hamiltonians}.
When we label the local basis states of a spin chain by letters $A,B, \dots$, we call a projector term acting nontrivially on two neighboring particles, with the form 
\begin{align}
	\ii\otimes \frac{1}{2} \left(\ket{AB}-\ket{CD}\right)\left(\bra{AB}-\bra{CD}\right)_{i,i+1} \otimes \ii,
\end{align}
a {\em rewriting interaction}.
We call a Hamiltonian built from such interactions a {\em rewriting Hamiltonian}. Its terms connect the computational basis states $\ket{\dots AB \dots}$ and $\ket{\dots CD \dots}$ by a transition.
We can also view this as a {\em rewriting rule} $AB \leftrightarrow CD$ connecting the strings $\dots AB \dots$ and $\dots CD \dots$.
These interactions divide the Hilbert space of the whole chain into easily identifiable invariant subspaces with a beautiful structure. 
Each subspace is spanned by computational basis states labeled by words connected by the rewriting rules. 
Because the rewriting projectors energetically prefer uniform superpositions of the states $\ket{\dots AB \dots}$ and $\ket{\dots CD \dots}$,
it turns out there is a unique, zero-energy ground state in each invariant subspace: the uniform-superposition of computational basis states (words) from that subspace. For example, for the rewriting rule $01 \leftrightarrow 10$, one of the invariant subspaces is spanned by strings with a single 1, and the uniform superposition of such states (the W-state) is a zero-energy ground state. Moreover, the set of words that span a subspace can often be easily identifiable -- in our cases by a pushdown automaton, as we will see in Section~\ref{sec:pf}.
However, this is not easy in general, since the {\em word problem} of Thue systems \cite{ThomasThue}, with rules that do not change word length, is PSPACE-complete and undecidable for arbitrary rules. 

Rewriting Hamiltonians are also a tool for quantum complexity, used widely in QMA-hardness constructions. To see how they are related to the so-called quantum Thue systems, tiling models and translationally invariant Hamiltonians whose ground states encode answers to QMA$_{\textrm{EXP}}$-complete problems, see e.g. \cite{2016arXiv160501718B}.

Our pair-flip (PF) model Hamiltonian starts as a translationally invariant rewriting Hamiltonian with the structure described above.
As such, it has a degenerate ground state, with one for each invariant subspace. Previous work (Motzkin, Fredkin spin chains) has dealt with this degeneracy by adding non-translationaly invariant projector terms at the chain boundaries.
In contrast, our models do not have/need boundary conditions. 
However, as we also desire a model with a unique ground state, we have to ensure only one of the ground states -- the uniform superposition of words from one subspace -- is energetically preferred. For this, we add a local perturbation term, energetically preferring the states with a highest average number of neighboring identical letter pairs. Furthermore, we can do this with only a negligible disturbance to the original eigenstate, keeping its high entanglement. The price we pay for this is that our systems are no longer frustration free, and that the gap becomes smaller, governed by the perturbation-induced energy splitting between the subspaces.

\subsection{The Motzkin spin chain: a review}
\label{sec:motzkinreview}

Let us now summarize the previous results on translationally invariant, low-local dimension (qudit) chains with rewriting Hamiltonians,
so that we can compare them to our new pair-flip (PF) model in Section~\ref{sec:pf}.
We also sketch the techniques that will be useful for the analysis of the PF model.
The reader can find a detailed derivation of some the results discussed below in Appendix~\ref{sec:technical}, as well as in the papers \cite{CriticalityWithoutFrustration,MovaShor,FredkinSpinChain,ViolationOfCD}.

We start with the basic properties of the Motzkin spin chain (also called the bracket model), 
an example of a low local dimension ($d=5$) translationally invariant, frustration-free spin chain with a unique, very entangled ground state and an inverse polynomial gap.
We also present some of the combinatorial, graph-theoretic, and perturbation techniques used to find its properties. These will be useful for the analysis of our pair-flip model family in Section \ref{sec:pf}.

The original spin-1 ($d=3$) Motzkin chain was introduced in \cite{CriticalityWithoutFrustration}. 
It is a {\em rewriting Hamiltonian} with the 3-letter alphabet $L,R,0$ (left bracket, right bracket, empty space), and three rewriting rules:
\begin{align}
	00 \leftrightarrow LR, 
	\qquad 0L \leftrightarrow L0, 
	\qquad 0R \leftrightarrow R0.
\end{align}
Observe that applying these rewriting rules to the all-0 string, one can obtain all well-bracketed words (with empty spaces), e.g. $LLRR$ or $0L0R$. 
As described in Section~\ref{sec:results}, the rewriting rules translate to projector terms in the Hamiltonian
for a qutrit ($d=3$) spin chain, with basis states $\ket{0},\ket{L},\ket{R}$:
\begin{align}
	H_{\textrm{create}} &= \frac{1}{2} \sum_{i=1}^{N-1}\ket{00-LR}\bra{00-LR}_{i,i+1}, \label{Hcre}\\
	H_{\textrm{move}} &= \frac{1}{2} \sum_{i=1}^{N-1} \left( \ket{0L-L0}\bra{0L-L0} + \ket{0R-R0}\bra{0R-R0}\right)_{i,i+1}. \label{Hmov}
\end{align}
We can interpret $\ket{L}, \ket{R}$ as particles that hop on this chain and can appear and disapper in pairs.
Note that the ``empty chain'' state $\ket{0\cdots 0}$ is connected by this Hamiltonian only to other well-bracketed words (e.g. $\ket{LR0L0R}$). In fact, the Hilbert space splits into invariant subspaces classified by an {\em irreducible word} -- the string of $L$'s and $R$'s remaining after we erase all matching bracket pairs from a computational basis state.

The whole Motzkin chain Hamiltonian for a qutrit chain of length $N$ is
\begin{align}
	H^{\textrm{Motzkin}}_{\textrm{+ends}} &= H_{\textrm{create}} + H_{\textrm{move}} + H_{\textrm{ends}}, \label{HLRend}
	\end{align}
and includes an additional non-translationally invariant endpoint term
\begin{align}
	H_{\textrm{ends}} &= \ket{R}\bra{R}_1 + \ket{L}\bra{L}_N,
\label{Hend}
\end{align}
energetically penalizing states $\ket{R\dots}$ that begin with a closing bracket or $\ket{\dots L}$ that end with an open bracket. In combination with the terms $H_{\textrm{move}}$ and $H_{\textrm{create}}$, it also penalizes states that are connected by the transitions in the Hamiltonian to states that begin/end with a bad bracket. Altogether, there remains only a single, unique, frustration-free ground state: 
the uniform superposition of all computational basis states labeled by well-bracketed words.

This unique ground state exhibits a logarithmic scaling of half-chain entanglement entropy with system size, typical for critical systems in 1D.
Moreover, this Hamiltonian built from projectors is frustration-free (all its terms kill the ground state). Later, Movassagh \cite{MovassaghCorrelations} also presented a spin-matrix language formulation of this model and calculated its various correlation functions. 

The model's straightforward generalization introduces several bracket species/colors \cite{MovaShor} (colored bracket model/ colored Motzkin spin chain), and an alphabet $\{0,L_1,R_1,\dots,L_s,R_s\}$. For $s$ bracket colors, the spin chain now has local particle dimension $d=1+2s$. The possibility of different bracket species greatly increase the complexity of the ground state, so that its Schmidt rank (cutting the chain in half) grows exponentially, while the entropy grows as $\oo{\sqrt{N}}$ \eqref{Smotzkin}.

Let us look at the properties of the qutrit Motzkin spin chain \cite{CriticalityWithoutFrustration} in more detail, as the techniques we use in analyzing our PF model are quite related.

\vskip3pt\noindent \textbf{Schmidt decomposition and entanglement entropy scaling.}

The computational basis states for the Motzkin spin chain are Motzkin paths, illustrated in Figure~\ref{fig:dyckmotzkinpf} and discussed in more detail in Appendix~\ref{sec:technical}.

\begin{figure}[h]
	\begin{center}
		\includegraphics[width=10cm]{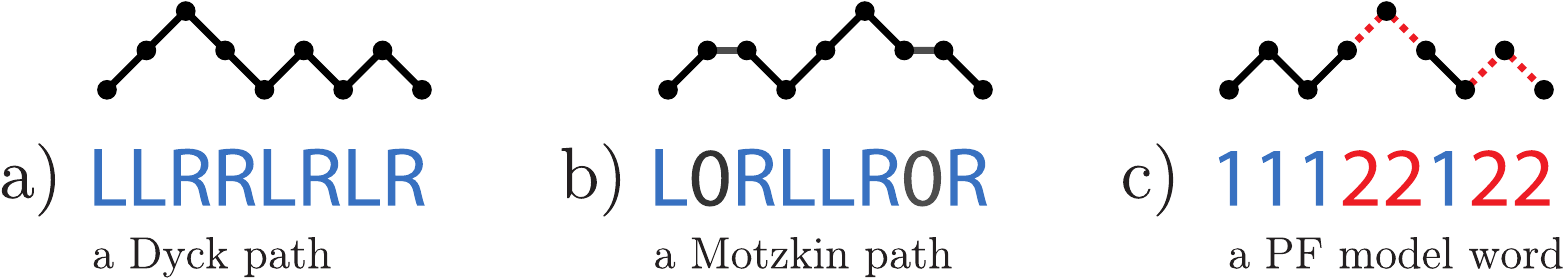}
		\caption{Examples of a) a Dyck path, also seen as a well-bracketed word ``[[]][][]'', b) a Motzkin path, also a well-bracketed word with spaces ``[$\cdot$][[]$\cdot$]'', and c) a pair-flip model fully reducible word ``11122122''.
		}\label{fig:dyckmotzkinpf}
	\end{center}
\end{figure}

The Motzkin chain ground state $\ket{\psi}$ is the uniform superposition of well bracketed words (with spaces). For a chain of length $N=2n$, any well bracketed word with spaces (Motzkin path) $w$ can be split in half as $uv$ with a word $u$ of length $n$ that has some number $r$ of extra left brackets, and a word $v$ of length $n$ that has $r$ extra right brackets. 
Thus, we can write the ground state in a Schmidt decomposition form as
\begin{align}
	\ket{\M_{n:n}} = \frac{1}{\sqrt{M_N}}\sum_{w \in \M_n} \ket{w}_{1,\dots,N}
	= 
	\sum_{r=0}^n
	\underbrace{\frac{M_{n}^{r}}{\sqrt{M_N}}}_{\lambda_r}
	\underbrace{\frac{1}{\sqrt{M_{n}^{r}}}\left(\sum_{u \,\rightarrow \,r \textrm{ ('s}} \ket{u}_{1,\dots,n}
	\right)}_{\ket{\psi_{r,0}}}
	\underbrace{\frac{1}{\sqrt{M_{n}^{r}}}\left(\sum_{v \,\rightarrow \,r\textrm{ )'s}} \ket{v}_{n+1,\dots,2n}
	\right)}_{\ket{\psi_{0,r}}}, \label{schmidtdecomp}
\end{align}
where $M_N$ is the number of Motzkin paths of length $N$ from \eqref{Mn},
$M_n^{r}$ the number of Motzkin paths of length $n$ ending at level $r$ from \eqref{Mnk},
the first sum is over all well-bracketed Motzkin paths $w$ of length $N=2n$,
and the second term involves sums over Motzkin paths $u$ of length $n$ with $r$ extra left brackets and Motzkin paths $v$ with $r$ extra right brackets.
The state $\ket{\psi_{r,0}}$ is the normalized uniform superposition of all Motzkin paths of length $n$ with $r$ extra left brackets and 0 extra right brackets.
The Schmidt coefficients across the half-chain decomposition for this state are thus 
\begin{align}
	\lambda_r = \frac{M_{n}^{r}}{\sqrt{M_N}},
\end{align}
for $r=0,\dots,n$. The Schmidt rank is polynomial in $N$, but what is more interesting is how the $\lambda_r$'s fall with growing $r$. 
It turns out $\oo{\sqrt{N}}$ of them are relevant, and that when we calculate the entanglement entropy for this decomposition, we obtain
\begin{align}
	S\left(\ket{\M_{n:n}}\right) = - \sum_r \lambda_r^2 \log_2 \lambda_r^2 = \frac{1}{2}\log_2 N + \bigO(1)
	\quad\textrm{bits}, 
\end{align}
a logarithmic scaling with respect to the chain length $N$,
similar to critical systems in 1D. Surprisingly, here it happens for a frustration-free system.

\vskip3pt\noindent \textbf{Colored brackets.}
This model is easily generalizable to multiple types/colors of brackets $L_1,R_1, \dots, L_s,R_s$. Movassagh and Shor \cite{MovaShor} have calculated the properties of this model, and showed that already for $d=5$ with two bracket types: $L_1,R_1$ and $L_2,R_2$, which we can interpret as (, ), and \red{[}, \red{]}, plus the empty symbol, the 
Schmidt rank grows exponentially, as there are many ways to have extra brackets in the left part of the chain -- e.g. ((, (\red{[}, \red{[}(, \red{[[}, etc. -- and these have to be matched by extra brackets in the right part of the chain.
This means each Schmidt coefficient $\lambda_r$ in \eqref{schmidtdecomp}
is repeated $s^r$ times, as we have that many possibilities of coloring the extra brackets.
An analysis of the Schmidt coefficients results in entanglement entropy scaling as
\begin{align}
	S\left(\ket{\M^{(s)}_{n:n}}\right) = 2\log_2(s)\,\sqrt{\frac{\sigma N}{\pi}}+\frac{1}{2}\log_2(\sigma N) + \bigO(1) 
	\quad\textrm{bits},
	\label{Smotzkin}
\end{align}
where $\sigma=\frac{\sqrt{s}}{2\sqrt{s}+1}$.
For the system to have a unique ground state, another set of projector terms is required, penalizing mismatched bracket pairs, e.g. $L_1 R_2$, i.e. $(\,\red{]}$. If such brackets appear in a word on a neighboring pair, it is obvious that this word is not well-bracketed. On the other hand, words that are not well bracketed are connected by the rewriting rules to other words with either bad brackets at the endpoints, or badly matched pairs of brackets somewhere inside the chain.
This way, we again find a unique ground state, the uniform superposition of well bracketed words.

\vskip3pt\noindent \textbf{The boundary and translational invariance.}
One might wish to remove the requirement for endpoint projectors like $H_{\textrm{ends}}$ in \eqref{Hend} which introduce an energy split between the degenerate ground states of the bracket Hamiltonian. Such terms break translational invariance. However, without them, there is a ground state with 0 energy in each subspace -- the uniform superposition of all well-bracketed words, the uniform superposition of all words with 1 extra left bracket, etc., up to the ``boring'' product state $\ket{LL\cdots L}$. 
However, there is a way to break this degeneracy in a translationally invariant way by introducing a little bit of frustration -- a small cost per bracket. 

Movassagh and Shor \cite{MovaShor} claim\footnote{
Relying on our exact calculation and extensive numerical investigation, we believe the resulting scaling in \cite{MovaShor} is correct. However, we believe one needs much more care in the approximations of binomials by exponentials as some of the disregarded error terms might become relevant and ruin the calculation.
In particular, even for small $m$, we find that the prefactors as well as constant terms in \cite{MovaShor} are off.
}
 that the average number of brackets in a uniform superposition with $m$ extra brackets on a chain of length $N$ grows with $m$ proportionally to $m^2/N$.
Thus, it is the lowest in the uniform superposition of states from the subspace connected to the ``empty'' state with no extra brackets. Thus, we can look at a modified Hamiltonian, where a bracket costs something:
\begin{align}
	H^{\textrm{Motzkin}}_{\delta} &= H_{\textrm{create}} + H_{\textrm{move}} + \delta \underbrace{\sum_{i=1}^{N} \left( \ket{L}\bra{L} + \ket{R}\bra{R}\right)_i}_{H_{\textrm{cost}}},
	\label{Hcost}
\end{align}
with $H_{\textrm{create}}$ and $H_{\textrm{move}}$ from \eqref{Hcre} and \eqref{Hmov}. The additional term $H_{\textrm{cost}}$ can be viewed as a global field that energetically penalizes the bracket states $\ket{L}$, $\ket{R}$, and prefers the ``empty'' state $\ket{0}$.
Choosing a small, well tuned cost $\delta = 1/\textrm{poly}(N)$ then allows one to select a unique ground state -- a state close to the uniform superposition of the well-bracketed Motzkin paths. It requires treating the term $\delta H_{\textrm{cost}}$ as a perturbation, for which the gap $\Delta_k$ in each invariant subspace must be at least an inverse polynomial in $N$, and $\delta \|H_{\textrm{cost}}\| \ll \Delta_{k}$ for all $k$.
A small $\delta$ naturally results in the new ground state very much resembling the uniform superposition, retaining its entropy properties. However, the degree of the inverse polynomial in the gap of the new Hamiltonian increases. 

Yet another way to remove the boundary terms exists, requiring an increase of the local dimension by 1, due to Gottesman and Irani \cite{v009a002} and mentioned e.g. in Bausch et al. \cite{2016arXiv160501718B}.
The trick is to introduce a special particle type $X$, and terms that are tuned so that exactly two $X$ particles appear, while the most favorable place for them is on the ends of the chain. We can then use an interaction of a ``bad'' bracket with the $X$as a substitute for endpoint terms. Of course, the extra price for this construction, besides the local dimension increase, is a decrease in the gap (when keeping the norm of the Hamiltonian), while the model also becomes frustrated.

We will now continue with the review of rewriting Hamiltonians and look at the Fredkin spin chain with next-nearest interactions, whose states can be viewed as Dyck paths.
At the end of the following Section, we present an exact calculation underlying another way to remove the boundary-terms for the Fredkin spin chain. We will break the energy degeneracy between ground states from different subspaces by counting the average number of peaks (see Fig.~\ref{fig:dyckmotzkinpf}a), i.e. 
$L_iR_i$ pairs on neighboring particles.


\subsection{The Fredkin spin chain and other related models}
\label{sec:fredkinreview}

Afer the Motzkin chain, the {\em Fredkin spin chain} family for all half integer spins was introduced by \cite{FredkinSpinChain,ViolationOfCD}. The authors remove the need for the extra local dimension corresponding to background empty spaces ($0$'s) by extending the range of the interactions to next nearest neighbors. They devise clever rewriting rules involving three neighboring sites:
\begin{align}
	LLR \leftrightarrow LRL, \qquad LRR \leftrightarrow RLR, \label{Fredkinrules}
\end{align}
reminiscent of the reversible Fredkin (controlled-swap) gate \cite[eq. 4]{FredkinSpinChain}: $110 \leftrightarrow 101$.

\begin{figure}[!h]
	\begin{center}
		\includegraphics{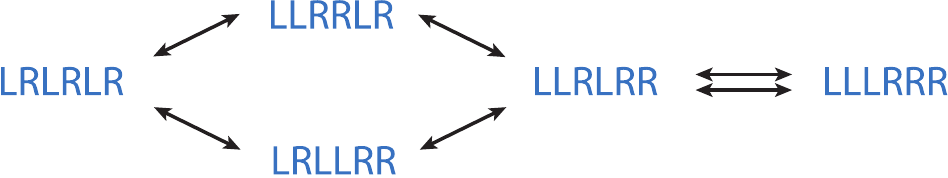}
		\caption{The Fredkin spin chain: the transition (multi)graph of well bracketed words induced by the rewriting rules  
		\eqref{Fredkinrules}, for $N=6$.
		}\label{fig:FredkinTransitionGraph}
	\end{center}
\end{figure}

It turns out that these rewriting rules connect all well bracketed words,
as illustrated in Figure~\ref{fig:FredkinTransitionGraph}.
The model also includes endpoint terms punishing wrongly open brackets, i.e. states $\ket{R\cdots}$ and $\ket{\cdots L}$.
The corresponding next-nearest-neighbor rewriting Hamiltonian with endpoint projectors thus has a unique ground state: the uniform superposition of all well bracketed words, also known as Dyck paths (see Fig.\ref{fig:dyckmotzkinpf}a). 
This model also has a colored version with several bracket types, with rewriting rules moving any bracket through a matching bracket pair (a peak): $L_jR_jL_k\leftrightarrow L_kL_jR_j$ and $L_jR_jR_k\leftrightarrow R_kL_jR_j$, changing the type of a matched bracket pair $L_jR_j\leftrightarrow L_kR_k$, together with boundary conditions \eqref{Hend} and a penalizing term for ``matching'' wrong bracket types as in Section~\ref{sec:motzkinreview}.
The transitions in the Hamiltonian keep the number of extra brackets intact, and thus the Hilbert space again splits into invariant subspaces.
The local dimension is $d=2s$ for $s$ types of brackets.
For qubits (one bracket type) the Fredkin chain Hamiltonian 
is $H^{\textrm{Fredkin}}+H_{\textrm{ends}}$, with
\begin{align}
	H^{\textrm{Fredkin}}=&\frac{1}{2}\sum_{i=1}^{n-2}\left(\ketbra{LRL-LLR}+\ketbra{LRR-RLR}\right)_{i,i+1,i+2},\\
	H_{\textrm{ends}} &= \ket{R}\bra{R}_1 + \ket{L}\bra{L}_n.
\end{align}

Similarly as for the $d=3$ Motzkin chain, the entanglement entropy of the uncolored Fredkin chain (a single type of brackets means $s=1$ and local dimension $d=2$) scales as $S\left(\ket{\D_{n:n}}\right)=\frac{1}{2}\log_2 N + \bigO(1)$. Adding more colors ($s\geq 2$), the entanglement entropy starts to grow as $S(\ket{\D^{(s)}_{n:n}})=S\left(\ket{\D_{n:n}}\right)+2\log_2(s)\, \sqrt{\frac{N}{2\pi}}+\bigO(1)$.
For this, we need local dimension $d=2s \geq 4$. Later, Movassagh \cite{MovFSCGap} proved that this model also has an inverse polynomial spectral gap.

The uniqueness of the ground state is again ensured by the non-translationally invariant boundary terms. Nevertheless, we can remove this requirement. We have mentioned at the end of the last Section that counting particles can break the ground state degeneracy. Here, we invent a different approach: preferring peaks, i.e. bracket pairs $LR$ on neighboring spins (the peaks in Figure~\ref{fig:dyckmotzkinpf}a).
This is one of the original contributions of this paper.

\begin{figure}[h]
	\begin{center}
		\includegraphics[width=14cm]{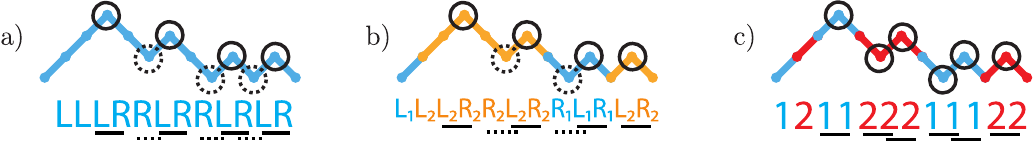}
		\caption{a) The peaks (full circles) and valleys (dotted circles) in Dyck paths/bracket words. a) The peaks and valleys in multicolored Dyck paths. To count as a valley, both its sides must have the same color. c) The pairs in PF model words.}\label{fig:peakcount}
	\end{center}
\end{figure}

We claim that the Hamiltonian
\begin{align}
	H^{\textrm{Fredkin}} - \delta \sum_{i=1}^{N-1} \sum_{j=1}^{s} \ket{L_j R_j}\bra{L_j R_j}
	\label{fredkinperturb}
\end{align}
for a carefully chosen $\delta = 1/\textrm{poly}(N)$ has a unique ground state well approximated by the uniform superposition of well-bracketed words.
The same holds for the colored Fredkin spin chain.
To prove this, we exactly calculate the average number of peaks 
(letters $LR$ on neighboring sites, see Figure~\ref{fig:peakcount}a)
in uniform superpositions of Fredkin spin chain words (Dyck paths) of length $N$, in each invariant subspace with $k$ extra brackets.
In Appendix~\ref{sec:peakcount} \eqref{pnkfullcolor}, we find that this number is  
\begin{align}
	p_{n,k} = \frac{N+2}{4} - \frac{k(k+2)}{4N},
\end{align}
for even $N$, decreasing with $k$, i.e. it is the largest for $k=0$.
Let us then treat the peak-counting term as a perturbation. Note that it does not affect the invariant subspace structure.
With a bit of work on top of \cite{MovFSCGap} we can show that the Fredkin Hamiltonian has an inverse-polynomial gap in each of its invariant subspaces. 
Let us choose a small $\delta= 1/\textrm{poly}(N)$ in \eqref{fredkinperturb}.
The first order in $\delta$ energy shift of 
each former ground state 
with $k$ extra brackets 
depends on the expectation value 
\begin{align}
 - \delta \bra{\psi_k} \left(\sum_{i=1}^{N-1}  \sum_{j=1}^{s} \ket{L_j R_j}\bra{L_jR_j}\right)\ket{\psi_k} = 
\frac{\textrm{total peaks}}{\textrm{all words}} = - \delta p_{n,k}, \label{pertE}
\end{align}
which is the lowest (negative) for $k=0$, and then increases with increasing $k$ \eqref{pnkfullcolor}.
Thus, we can use this perturbation to remove the need for boundary terms in the Fredkin spin chain.
Furthermore, the new ground state is the old uniform superposition of well-bracketed words, with a first-order correction in $\delta$. 
The difference from the uniform amplitudes must be small, otherwise the energy of the state would be far from \eqref{pertE}.
The amplitude of a particular word depends on the number of its peaks. However, the distribution of the number of peaks for words with a reasonable number of extra brackets is tightly centered about the average $p_{n,k}$. Therefore, for a small inverse polynomial $\delta$, the correction to the ground state can be so small, that the corrections to the Schmidt coefficients do not change the resulting entropy scaling. However, as of now, we do not yet have an exact statement of this and leave the robustness of the entropy scaling as an open direction for further research. 

Now we can finally introduce our the pair-flip model, which has very similar properties, but already for a lower local dimension $d=3$, and with only nearest neigbor interactions.

\section{The pair-flip (PF) model}
\label{sec:pf}

Let us finally present the pair-flip (PF) model.
It is a rewriting Hamiltonian, whose rules can be seen as pairs of neighboring particles changing type together,
i.e. $11\leftrightarrow 22$. The local dimension $d$ counts the number of particle types. The computational basis states are thus strings of letters from the alphabet $\{1,\dots,d\}$, and the local single-particle basis states are $\ket{1},\dots,\ket{d}$. 
Unlike the Motzkin spin chain, and similarly to the Fredkin chain, there is no ``empty space'' state on the chain.

The PF model Hamiltonian has the form $H_{\textrm{PF}} = H_{\textrm{flip}}^{\textrm{PF}}+ \delta H_{\textrm{cost}}^{\textrm{PF}}$. First, we have projector terms corresponding to pair-flipping rewriting rules:
\begin{align}
	H_{\textrm{flip}}^{\textrm{PF}} &= \frac{1}{2} \sum_{i=1}^{N-1} \sum_{t=1}^{d} \sum_{t'\neq t} \ket{t't'-tt}\bra{t't'-tt}_{i,i+1}, \label{HXXflip}
\end{align}
and a Potts-model like cost term favoring neighboring particles of the same type:
\begin{align}
	H_{\textrm{cost}}^{\textrm{PF}} &= - \sum_{i=1}^{N-1} \sum_{t=1}^{d} \ket{tt}\bra{tt}_{i,i+1}. \label{HXXcost}
\end{align}
This model differs from the Motzkin chain in three main ways. 
First, the ``particles'' are their own ``antiparticles'' --  instead of matching pairs of brackets, the transitions here involve pairs of letters without a left/right orientation. Second, there is no ``movement'' term, as there is no explicit ``empty space'' state of the chain 
Particle ``movement'' is facilitated indirectly. For example, the letter $1$ in the word $122$ can ``move'' two spaces to the right by the sequence of pair-flip transitions $122 \rightarrow 111 \rightarrow 221$.
In this way, the PF model is similar to the Fredkin chain, which also does not include empty spaces, yet it connects its basis states in nontrivial ways.
Third, instead of boundary terms or particle counting, we will use the pair-counting cost term to break the energy degeneracy between ground states from different invariant subspaces.

\vskip3pt\noindent\textbf{Irreducible strings.}
The transitions allowed by the PF Hamiltonian connect {\em words} (computational basis states) that are related by flipping the type of neighboring letter pairs. For example, the word $2222$ is connected to $1122$, $2112$ and $2211$, while two of those are connected to the word $1111$. However, none of these are connected to the words $1222$ or $1212$, as there is no way to obtain them from $1111$ by pair-flips.
The strings $1222$ or $1212$ can't be {\em reduced} to $1111$.
This intuition can be formalized by into a unique {\em reduction} procedure, classifying words of the PF-model by their {\em irreducible strings}.
Informally, we read the word from left to right, match letter pairs and erase them. Whatever is left after all matched pairs are erased is the irreducible string (see Figure~\ref{fig:reduction}).

\begin{figure}
\begin{center}
\includegraphics[width=16cm]{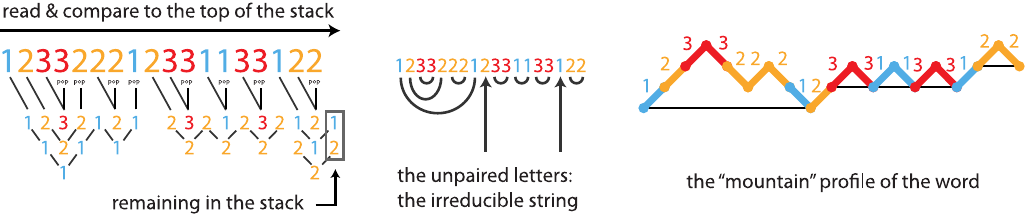}%
\caption{
The {\em reduction procedure} for the pair-flip (PF) model produces an irreducible string from a word: Read the word from left to right, push (store) the letters into a stack, or pop them out when equal to the currently read letter. The resulting letter pairings are depicted on the right. Here, the irreducible string is ``21''.
The last picture shows also the ``mountain'' profile of the word 12332221\,\underline{2}\,331133\,\underline{1}\,22.
}
\label{fig:reduction}%
\end{center}
\end{figure}

\begin{definition}[the irreducible string of a PF model word]
\label{def:reduction}
Consider a word $w=A_1 \cdots A_n$ with $N$ letters coming from an alphabet with $d$ different letters (colors).
Let a pushdown automaton read the word from left to right, starting with $k=0$, and an empty stack string $\sigma$. 
While the word is not fully read yet ($k<N$), keep moving to the right ($k \leftarrow k+1$), reading the letter $A_k$, and continue as
\begin{enumerate}
\item If $A_k\neq S_1$ (the first letter of the stack string $\sigma$), push $A_k$ into the stack: $\sigma \, \,  \leftarrow\, \,  A_k \sigma$, or
\item if $A_k=S_1$, pop the stack: $\sigma \, \, \leftarrow \, \, S_2S_3S_4\cdots$, pairing the new letter $A_k$ with the former top letter $S_1$ of the stack and erasing them. 
\end{enumerate}
After reaching the end of the word, read what's left in the stack from bottom to top ($\cdots S_3 S_2 S_1$). This is the {\em irreducible string} of the word $w$. If the irreducible string is empty, we call the word {\em fully reducible}. 
\end{definition}

The transitions in the Hamiltonian can not modify the {\em irreducible string} of a word. Thus, there is an invariant subspace of the Hamiltonian for each different irreducible string. Each irreducible string is simply a sequence of letters that can not immediately repeat. Observe that a chain of length $2n$ has $1+\sum_{i=1}^{n} d(d-1)^{2i-1}$ possible irreducible strings, which grows exponentially in $n$ for $d\geq 3$. In the simpler case $d=2$, the number of invariant subspaces is only $1+2n$.
We will discuss this Hilbert space structure in more detail below in Sections~\ref{sec:pf2} and \ref{sec:schmidt}. 

\vskip3pt\noindent\textbf{Schmidt decomposition and entropy.}
Just as we did for the uniform superposition of well-bracketed words (Motzkin paths) in \eqref{schmidtdecomp}, it is straightforward to write down a Schmidt decomposition of a uniform superposition of fully reducible words. For a chain of length $N=4m$, cut across the middle, we get

\begin{align}
	\ketL{\textrm{PF}_{2n:2n}^{(d)}}
	&= \frac{1}{\sqrt{W_{2m}^{(d)}}}\sum_{t\in \PF^{(d)}_n} \ket{t}_{1,\dots,N}
	\label{schmidtcut}\\
	&= 
	\sum_{r=0}^{m}
	\sum_{|u|=2r}
	\underbrace{\frac{W_{m-r,2r}^{(d)}}{\sqrt{W_{2m}^{(d)}}}}_{\lambda_r}
	\underbrace{\frac{1}{\sqrt{W_{m-r,2r}^{(d)}}}\sum_{v \rightarrow u} \ket{v}_{1,\dots,n}
	}_{\ket{\psi_{u}}}
	\underbrace{\frac{1}{\sqrt{W_{m-r,2r}^{(d)}}}\sum_{w \rightarrow u^R} \ket{w}_{n+1,\dots,2n}
	}_{\ket{\psi_{u^R}}}, 
	\nonumber
\end{align} 
where we are summing over all irreducible strings $u$ of length $|u|=2r$ for $r=0,\dots,m$, PF model words $v$ of length $2m$ that reduce to $u$ and PF model words $w$ of length 2m that reduce to $u^R$, the mirror image of $u$, so that $vw$ as well as $uu^R$ are fully reducible words.

We can look at the quantum correlations in this state by analyzing the Schmidt coefficients $\lambda_r$. There are $I_{2r} = d(d-1)^{2r-1}$ different irreducible words of length $2r$, so the Schmidt coefficient $\lambda_r$ appears $d(d-1)^{2r-1}$ times in the decomposition. Therefore, the Schmidt rank grows exponentially with the chain length $N$ for $d\geq 3$, as $r$ goes from $0$ to $m=\frac{N}{4}$.

To get a more detailed understanding of the entanglement entropy, we need to look at the asymptotic scaling of $W_{m-r,2r}^{(d)}$ and $W_{m}^{(d)}$. In Section~\ref{sec:pf2}, we perform the calculation exactly for the $d=2$ model. For the $d\geq 3$ models, in Section~\ref{sec:pf3} we obtain upper and lower bounds on the 
word counts, and the entanglement entropy. We show that it scales like that of the $d-1$ colored Fredkin spin chain,
which has local dimension $2(d-1)$.

\vskip3pt\noindent\textbf{An inverse-polynomial gap.}
A nice feature of the PF model is that it has an inverse-polynomial, and not exponentially small gap.
To prove this, we first look at how the PF Hamiltonian splits the Hilbert space into invariant subspaces. The {\em fully reducible} subspace is spanned by the fully reducible words. We call the other invariant subspaces {\em irreducible}, and lower bound their energy gaps separately.

{\em In the fully reducible subspace}, we relate the gap of the Hamiltonian to the gap of a Markov chain induced by the transition rules. This is easy to do, as the Hamiltonian is composed of rewriting rules and its ground state is a uniform superposition. Notice that this mapping is possible for all frustration free Stoquastic Hamiltonians \cite{Bravyi:2009:CSF:1898209.1898222}, due to the Perron-Frobenius theorem. To lower bound the gap of this Markov chain, we relate it to a high-level pair-displacing Markov chain, by the comparison theorem. This Markov chain, instead of recoloring adjacent pairs, removes a random pair and inserts a randomly colored pair to a random position. We use a generalized form of the approach from \cite{CriticalityWithoutFrustration} (see also \cite{MovaShor}, where a constructive proof was given) to lower bound its gap. We can assign canonical paths thanks to fractional matching and linear programming methods.
Here, instead of giving a closed-form solution for the fractional matching, we just show that a solution indeed exists.
As the last step, the canonical paths give us a lower-bound on the conductance and therefore a lower-bound on the gap. A generalization of our method may be of independent interest also in other contexts. 

{\em The gap lower-bound in the irreducible subspaces:} Similarly as we did in \cite{CriticalityWithoutFrustration}, we view the ``movement'' of the irreducible particles (the irreducible string), as a small perturbation and rely on the projection lemma. In our case we cannot ``catch'' an irreducible particle by an endpoint projector assigning a high expectation value to unbalanced (irreducible) strings. This was enough for Motzkin and Fredkin spin chains, where one could consider just one particle with weighted hopping on a line. Instead, we need to consider all the irreducible particles with weighted hopping or even long-range jumping. By a series of comparisons, we arrive to an exclusion process with Glauber dynamics, and Metropolis transitions \cite{MartinRandallDecompositionAdsorbingStaircaseWalk2006}, for which we know an inverse polynomial gap lower bound.

Finally, we also prove an {\em upper bound on the spectral gap}: $\bigO(n^{-2})$ 
thanks to the properties of PF words and their relation to colored Dyck words. This helps us to specify a ``twisted'' ground state -- a state with a small overlap with the ground state and a small expectation value, relying on the relation of PF words to Dyck words, the universality of Brownian motion and the convergence of Dyck random walks to Brownian excursions \cite{MovaShor,MovFSCGap}.

\vskip3pt\noindent\textbf{A unique ground state.}
When we do not use the cost term \eqref{HXXcost}, the Hamiltonian is made only from the rewriting projectors, is frustration free, and has a degenerate ground state. There is one zero-energy ground state in each invariant subspace with irreducible strings of length $k$ -- the uniform superposition over all words within the invariant subspace. Let us then include the pair-counting term \eqref{HXXcost} as a negative perturbation. Analytically (for constant-$k$) and numerically (for high $k$ growing with $n$), we show that the former ground state from the fully-reducible subspace gets a larger energy shift. The pair-counting perturbation thus selects a unique ground state.
Furthermore, we conjecture that the new, perturbed ground state can be so close to the unique superposition over all fully reducible words, that the scaling of its entanglement entropy across a middle cut remains unchanged.

In our analysis below, we will rely on \cite{CriticalityWithoutFrustration}, \cite{MovaShor} and \cite{FredkinSpinChain,ViolationOfCD}, but also on new specific tools.  The combinatorics are more difficult than before (Motzkin or Fredkin chains), but the surprising results about the entanglement of the ground state and other properties of the PF model make up for the extra effort. In particular, the PF model with local dimension $d$ has the same asymptotic behavior (entropy scaling with $\sqrt{N}$) as the colored Motzkin chain with $d-1$ colors, and thus with local dimension $1+2(d-1)$. The Motzkin chain with two types of brackets ($s=2$) has a square-root scaling of the entanglement entropy for local dimension $5$, while this behavior appears in the PF-model already for $d=3$. On the other hand, 
its behavior is also comparable to that of the Fredkin spin chain \cite{FredkinSpinChain,ViolationOfCD} 
with $d-1$ colors and local dimension $2(d-1)$.
However, the PF model's interactions only involve nearest-neighbors, instead of next-nearest-neighbors for the Fredkin chain.

We prove a lower bound on the gap (see Section~\ref{sec:gap}) in the fully reducible subspace -- an inverse polynomial of a large constant degree. However, our numerics indicate that it is on the order of $N^{-2.38}$ (without the pair-counting cost term, from numerics up to $N=14$). This indicates a $S \propto \Delta^{-\frac{1}{4.76}}$ tradeoff between the gap and the entropy in the fully reducible subspace.

\section{The qubit PF-model ($d=2$)}
\label{sec:pf2}

Before delving into the more interesting $d\geq 3$ PF-models in Section~\ref{sec:pf3}, let us investigate the simplest, $d=2$ (qubit) PF model, with easier combinatorial calculations and analytical answers.

Is this a model that we might have seen elsewhere before?
Observe that after relabeling the alphabet from $\{1,2\}$ to $\{0,1\}$,
our $d=2$ pair-flip Hamiltonian maps to an   
instance of the spin-$\half$ XXZ model (with an energy shift, and with one of the three Pauli terms with a different prefactor): 
\begin{align}
	H^{\textrm{PF}}_{\textrm{flip}} = \frac{1}{2} \sum_{i=1}^{N-1} \left(\ket{00}-\ket{11}\right)\left(\bra{00}-\bra{11}\right)_{i,i+1} 
	&= \frac{1}{4} \sum_{i=1}^{N-1} \left(\ii - X_i X_{i+1} + Y_i Y_{i+1} + Z_i Z_{i+1}\right).
	\label{PFasXXZ}
\end{align}
Moreover, when we flip every other site, we obtian the Heisenberg XXX model \cite{SamajIntegrableMBSII} (with an energy shift):
\begin{align}
	\frac{1}{2} \sum_{i=1}^{N-1} \left(\ket{01}-\ket{10}\right)\left(\bra{01}-\bra{10}\right)_{i,i+1} 
	&= \frac{1}{4} \sum_{i=1}^{N-1} \left(\ii - X_i X_{i+1} - Y_i Y_{i+1} - Z_i Z_{i+1}\right)
  \label{PFflip2}
\end{align}
Of course, much is already known about this model -- e.g. the gap \cite{XXZKomaNachtergaele} or the ground state entanglement entropy \cite{PhysRevA.71.012301}.
However, we hope interpreting the transitions as pair-flipping
will let us understand its structure from another viewpoint.

Furthermore, the additional pair-counting term \eqref{HXXcost}, can be expressed in terms of Pauli matrices as
\begin{align}
	-\delta \sum_{i=1}^{N-1} 
			\left(\ket{01}\bra{01}+\ket{10}\bra{10}\right)_{i,i+1} 
	= -\delta \sum_{i=1}^{N-1} 
		\frac{1}{2}\left(\ii-Z_i Z_{i+1}\right),
\end{align}
when we flip every other site as in \eqref{PFflip2}. This add weight to the $ZZ$ term in \eqref{PFflip2}, so the full PF model for $d=2$ is an instance of the XXZ model in the paramagnetic phase near the ferromagnetic point (for large $N$).

Let us go back to the $\{1,2\}$ alphabet of the PF model. The {\em fully-reducible} subspace of this PF model is spanned by all states one can get to from the all-1 state by pair flips. Each invariant subspace is labeled by its irreducible strings: 
\begin{align}
	\varnothing, &\qquad 12, \qquad 1212, \qquad 121212, \quad \dots \\
								&\qquad 21, \qquad 2121, \qquad 212121, \quad \dots \nonumber
\end{align}
listed here for a chain with even chain length, $N=2n$.

One of our tasks is to identify and count the basis states in each of these subspaces. We show in Appendix~\ref{sec:PFtools} how to map the PF model words to walks on a $d$-regular tree that return to the origin (the fully reducible words), or end at a distance $k$ from the origin (words with an irreducible string of length $k$), as illustrated in Figure~\ref{fig:stree}. 
For $d=2$, this means simply walks on a line, which are easy to count. 
There are $W_n^{(2)} = \binom{2n}{n}$ fully reducible words -- returning walks on a line.
More generally, for 
words of length $2n+k$ that reduce to a particular string of length $k$, we have
\begin{align}
	W^{(2)}_{n,k} = W_{n,k}= \binom{2n+k}{n}.
	\label{d2count2}
\end{align}

Just as for the general PF model, the ground state in each invariant subspace is the uniform superposition of all words that reduce to the same irreducible string. 
We choose to look in more detail at the uniform superposition in the fully reducible subspace \eqref{schmidtcut}. Later at the end of this Section, we show how to break the degeneracy and select a unique ground state, close to this superposition.

Ultimately, we want to understand how entangled this state is, so we will analyze its Schmidt decomposition and entanglement entropy across a cut in the middle of the chain.
Of course, we could also divide the chain into smaller blocks, but the computation would get a bit more complex. For convenience, we also choose to look at spin chains of length $N=4n$, so that the cut across the middle
divides the system into two blocks of length $2n$.
Let us calculate the Schmidt number and Schmidt coefficients over this cut.
We start with \eqref{schmidtcut} and recall that $W^{(2)}_{2n}$ is the number of all fully reducible words on a chain of length $4n$ (with $2n$ letter pairs), and $W^{(2)}_{n-k,2k}$ is the number of words on the half-chain of length $2n$ (with $n-k$ letter pairs) with irreducible strings of length $2k$.
For the ground state of the $d=2$ PF model, there are $2(2n)+1$ terms in the Schmidt decomposition of the uniform superposition of fully reducible words across the middle cut. Thanks to \eqref{d2count2}, we find the Schmidt coefficients are
\begin{align}
	\lambda_{2k}
	= \frac{W_{n,k}}{\sqrt{W_{2n}}}
	= \frac{\binom{2n}{n-k}}{\sqrt{\binom{4n}{2n}}},
\end{align}
for $k=0,1,\dots,n$, noting
that for $k \geq 1$ the $\lambda_k$'s come in pairs,
as there are two irreducible words of length $2k$: $1212\cdots$ and $2121\cdots$.

Squaring the Schmidt coefficients, we get a probability distribution $\{p_k = \lambda^2_{2k}\}$.

To account for the double degeneracy of the $\lambda_{2k}$'s for $k>0$, we can simply
view it as a distribution 
for $k=-n,\dots,n$.
Using the normal approximation of the binomial coefficients \cite{sedgewick2013introduction}, $\binom{2n}{n}=2^{2n}/\sqrt{\pi n}\left(1+\bigO(1/n)\right)$ and $\binom{2n}{n+k}=2^{2n}\left(e^{-k^2/n}/\sqrt{\pi n}+\bigO\left(1/n^{3/2}\right)\right)$, we receive

\begin{align}
	p_k & = \frac{e^{-\frac{2 k^2}{n}}}{\sqrt{n \pi/2}}\left(1+\bigO\left(\frac{1}{n}\right)\right).
\end{align}
The entanglement entropy across the middle cut is thus
the entropy of this distribution.
\begin{align}
	S(\{p_k\}) = -\sum_{k=-n}^n p_k \log_2 \frac{e^{-\frac{2 k^2}{n}}}{\sqrt{n \pi/2}}+\bigO(1)
	& =\log_2 \sqrt{n \pi/2}+\frac{2}{n\ln 2}\sum_{k=-n}^n k^2  p_k+\bigO(1) \nonumber\\
	& =\log_2 \sqrt{n \pi/2}+\frac{2}{n \ln 2}\Var(\{p_k\})  p_k+\bigO(1)\
	\label{eq:QubitEntropymeanAt0},
\end{align}
because the distribution is symmetric, so the expected value of $k$ is zero, and the expected value of $k^2$ is equal to the variance. 
In Appendix~\ref{sec:d2entropyapp}, we prove that 
$\Var(\{p_k\})=\frac{n^2}{4n-1}$. 
Therefore, 
\begin{align}
	S(\{p_k\})
	& = \half\log_2 \frac{n\pi}{2}+\bigO(1)\quad \textrm{bits},\label{eq:QubitEntropyVar}
\end{align}
for a chain of length $N=4n$.
This could be also alternatively derived by approximating the sum in \eqref{eq:QubitEntropymeanAt0} by integrals.

This basically gives us the behavior of the frustration-free qutrit Motzkin spin chain, or the qubit Fredkin spin chain with next-nearest-neighbor interaction, in the fully reducible subspace. A numerical investigation (not our 1/poly proof) shows a better scaling of the gap within this subspace. The qubit Fredkin chain gap scales approximately as $n^{-3.02}$ for $n \leq 22$ (our own new numerics), the qutrit Motzkin spin chain gap scales approximately like $n^{-2.91}$ for $n \leq 13$ \cite{CriticalityWithoutFrustration}, and the qubit PF model gap (without the pair counting term) in the fully reducible subspace is just the Heisenberg XXX model gap $1-\cos(\pi/n)=\Theta(n^{-2})$.

\subsection{Breaking the degeneracy by pair-counting}
\label{sec:PF2breakMAIN}

To break the ground state degeneracy between subspaces, we propose to add frustration to our system similarly to what we propose for the Fredkin chain at the end of Section~\ref{sec:fredkinreview}. This time, we add a pair-counting term \eqref{HXXcost}, 
\begin{align}
	H_{\textrm{cost}}^{\textrm{PF}} &= - \sum_{i=1}^{N-1} \sum_{t=1}^{2} \ket{tt}\bra{tt}_{i,i+1}. 
\end{align}
that counts the number of subsequent letter pairs (see Figure~\ref{fig:peakcount}c). In Appendix~\ref{sec:countPFpairs}, we exactly calculate the average number of pairs of identical letters for groundstates from different invariant subspaces, and show that it is the largest for the fully-reducible subspace, with a $\Theta\left(n^{-1}\right)$ gap to the next one. As discussed in Section~\ref{sec:PF2break}, we can then use the term \eqref{HXXcost} as a perturbation to select a unique ground state, close to the uniform superposition of fully reducible words.
We have numerical evidence for the robustness of its entanglement entropy properties,
but the full analysis of this remains open for future work.
However, it opens up interesting possibilities of investigating this model with periodic boundary conditions.

\section{A family of very entangled models: PF with $d\geq 3$}
\label{sec:pf3}

In this Section, we turn to the more complex family of PF models for $d\geq 3$. 
First, we discuss the recursive and asymptotic counting results proven in Appendix~\ref{sec:PFtools}, and their relationship to expressions appearing in the Fredkin chain and the bracket model. Second, we look at the Schmidt decomposition of the uniform superposition of fully reducible words \eqref{schmidtcut} for $d\geq 3$, and present upper and lower bounds on the entanglement entropy. We show that it is comparable to that of the Fredkin and Motzkin spin chains with $d-1$ colors and local dimensions $2d-2$ and $2d-1$, respectively. Therefore, we get a very entangled ground state in the fully reducible space already for the qutrit ($d=3$) model.
Third, we analyze the gap of the Hamiltonian.
Fourth, we discuss breaking the degeneracy by pair-counting terms,
presenting analytic results for a range of parameters and general recursive formulas for numerical investigation.

\subsection{Counting the words}

Counting the number of words in the PF model is equivalent to counting walks on $d$-regular trees (see Figure~\ref{fig:stree}).
In Appendix~\ref{sec:PFtools}, we derive a recursive formula \eqref{recursionPCcolored}
for $W_{n}^{(d)}$, the number of words with $n$ letter pairs that reduce to the empty string in the $d$-color PF-model. For $d=3,4,5$, these are known as the Online Encyclopedia of Integer Series (OEIS) sequences A089022, A035610, A130976 \cite{OEIS}. 
In Appendix~\ref{sec:PFtools}, we find the generating function \eqref{W0generate} of the series $\sum_n W_{n}^{(d)} z^n$. It highlights the  relationship \eqref{Vexact}
\begin{align}
	W_{n}^{(d)} = \frac{(d-1)R^{(d)}_n}{d} \,  C^{(d-1)}_n
\end{align}
between the number of fully reducible PF model words with $n$ letter pairs, and the number of $d-1$ colored Dyck paths with $n$ pairs \eqref{countcoloredCn}. 
The proportionality coefficients
$R_n^{(d)} = \sum_{j=0}^{\infty} \left(\frac{d-1}{d^2}\right)^j \frac{C_{n+j}}{C_n}$
monotonously grow with $n$ towards $\left(\frac{d}{d-2}\right)^2$, as shown in \eqref{Rnfirst}-\eqref{Rngrow}.

Counting words of length $2n+k$ that reduce to an irreducible string of length $k$, we
find upper \eqref{Wnkupper} and lower bounds \eqref{Wnkupper} in terms of the number 
of $d-1$ colored Dyck paths with $k$ extra steps:
\begin{align}	
	C_{n,k}^{(d-1)} 
	\leq W_{n,k}^{(d)} 
	\leq \frac{d(d-1)}{(d-2)^2} C_{n,k}^{(d-1)}.
	\label{Vnkbounds}
\end{align}
Note that this also works for $k=0$, i.e. the fully reducible words.
Assuming large $n$, we asymptotically get
\begin{align}
	W_{n}^{(d)} 
	\sim \frac{d(d-1)}{(d-2)^2}\, C_{n}^{(d-1)}\sim\frac{(4(d-1))^n}{\sqrt{\pi}n^{3/2}} \frac{d(d-1)}{(d-2)^2},
	\label{W0asym}
\end{align}
where when we say $f(n)\sim g(n)$ we mean $f(n)=g(n)+o(g(n))$.
These expressions will allow us to find good bounds on the entanglement entropy of the uniform superposition of fully reducible words \eqref{schmidtcut}.

\subsection{The Schmidt decomposition and the entropy}
\label{sec:schmidt}

Similarly to what we did for the $d=2$ model in Section~\ref{sec:pf2}, we now express the Schmidt coefficients for the uniform superposition of fully-reducible words \eqref{schmidtcut}, over a half-chain division, and use them to calculate the entanglement entropy. 
Note that we could do this for any division, but the counting would become a bit more tedious. Furthermore, for simplicity, let us again assume the chain length is $N=4n$, so the irreducible strings for halves of the chain also have even length. 

Let us start with the state \eqref{schmidtcut} in its Schmidt decomposed form.
This time, we have $d\geq 3$, and $I_{2k}=d(d-1)^{2k-1}$ different irreducible strings of length $2k$. Thus, we have $I_{2k}$ 
coefficients $\lambda_k$ of the same magnitude for each $k=0,\dots,m$, and Schmidt rank scaling exponentially with the system size.
The squares of the $\lambda_k$'s form a probability distribution.
The entanglement entropy of a $2n:2n$ cut is then 
\begin{align}
	S\left(\ketL{\textrm{PF}_{2n:2n}^{(d)}}\right) 
	= -\sum_i \lambda^2_i \log \lambda^2_i 
	= - \sum_{k=0}^{n} I_{2k} 
	\underbrace{\frac{\left(W^{(d)}_{n-k,2k}\right)^2}{W^{(d)}_{2n}}}_{\lambda_{2k}^2} \log \underbrace{\frac{\left(W^{(d)}_{n-k,2k}\right)^2}{W^{(d)}_{2n}}}_{\lambda_{2k}^2}. \label{lam2PF}
\end{align}
We will now obtain the scaling of the entanglement entropy by relating it to the results found in \cite{FredkinSpinChain,ViolationOfCD}. 

Let us prove
\begin{align}
S\left(\ketL{\textrm{PF}_{2n:2n}^{(d)}}\right)
 = \oo{S\left(\ketL{\D^{(d-1)}_{2n:2n}}\right)}=\oo{\sqrt{n}}. \label{entropybounds}
\end{align}
The first step is to relate the $\lambda^2_{2k}$ of the $d$ color PF model ground state \eqref{lam2PF} to the analogous coefficients $\tilde{\lambda}_{2k}^2$ of the $d-1$ colored Dyck ground state (the ground state of the corresponding Fredkin spin chain), 
\begin{align}
	\tilde{\lambda}^2_{2k}
		= \frac{\left(C_{n-k,2k}^{(d-1)}\right)^2}{C^{(d-1)}_{2n}}.
\end{align}
For this, we need to use the bounds \eqref{Vnkbounds} on 
$W_{n-k,2k}^{(d)}$ as well as the asymptotic scaling \eqref{W0asym} of $W_{2n}^{(d)}$, in terms of the corresponding colored Dyck path (Catalan) numbers $C_{n-k,2k}^{(d-1)}$ and $C_{2n}^{(d-1)}$.
We obtain (for $n\rightarrow\infty$)
\begin{align}
	c^{-1} \tilde{\lambda}^2_{2k} &\leq\lambda_{2k}^2\leq c \tilde{\lambda}^2_{2k}.
\end{align}
for a constant
$c=\frac{d(d-1)}{(d-2)^2}$.
Note that for $k=0$, the Schmidt coefficient appears just once in both models, while for 
$k\geq 1$ the degeneracy of $\lambda_{2k}^2$ is 
$I_{2k}=d(d-1)^{2k-1}$ for $d$ colored PF words,
and $\tilde{I}_{2k}=(d-1)^{2k} = \frac{d-1}{d}I_{2k}$ for $d-1$ colored Dyck paths.
We can thus lower bound the entanglement entropy of the uniform superposition of fully reducible words in the PF model as
\begin{align}
	S\left(\ketL{\PF_{2n:2n}^{(d)}}\right)
	& =-\sum_{k=0}^n I_{2k} \lambda^2_{2k} \log_2 \lambda^2_{2k}
	  \geq-\frac{d}{d-1}
	\sum_{k=0}^n 
	\tilde{I}_{2k}\, c^{-1}\tilde{\lambda}^2_{2k} \log_2 c\tilde{\lambda}^2_{2k}
	+ \frac{1}{d-1} \lambda^2_{0} \log_2 \lambda^2_{0}
	\nonumber\\
	& = \left(\frac{d-2}{d-1}\right)^2 S\left(\ketL{\D^{(d-1)}_{2n:2n}}\right) + \bigO(1)\quad\textrm{bits},
	\label{entropydown}
\end{align}
as the $k=0$ term contributes not more than a constant, which can be shown using \eqref{Vnkbounds} and the asymptotic expansion of $C_n$ \eqref{asymcatalan}. 

Similarly, we also get an upper bound on the entropy:
\begin{align}
S\left(\ketL{\PF_{2n:2n}^{(d)}}\right)
  &= -\sum_{k=0}^\chi I_{2k} \lambda^2_{2k} \log_2 \lambda^2_{2k} 
	\leq -\frac{d}{d-1} \sum_{k=0}^\chi \tilde{I}_{2k}\, c\tilde{\lambda}^2_{2k}
		\log_2 c^{-1}\tilde{\lambda}^2_{2k} \nonumber\\ 
	&= \left(\frac{d}{d-2}\right)^2 S\left(\ketL{\D^{(d-1)}_{2n:2n}}\right) + \bigO(1)\quad\textrm{bits}.
	\label{entropyup}
\end{align}
Together with \eqref{entropydown} and the entropy scaling for the colored Fredkin spin chain and colored Dyck paths \cite{MovaShor}, this implies \eqref{entropybounds}, what we set out to prove.

Finally, numerical investigations 
for $d=3$ show that the ratio of the PF-model and $d-1$ colored Fredkin chain ground state entropies approaches 1 ($0.9825$ for $N=8000$).
Meanwhile, our loose theoretical bounds give us a lower bound of $0.25$
and an upper bound of $9$.


\subsection{The energy gap}
\label{sec:gap}

In this Section, we analytically bound the energy gap of the PF model, without the cost term
$H_{\textrm{cost}}^{\textrm{PF}}$ \eqref{HXXcost}. Without it, the Hamiltonian of PF model is frustration free, but it has degenerate ground states -- uniform superpositions of all words sharing the same irreducible string. We will do this here for local particle dimensions $d\geq 3$ (spin-1 and higher). The qubit (spin-$\half$) PF model without $H_{\textrm{cost}}^{\textrm{PF}}$ is just the Heisenberg XXX model in disguise, and we already know its spectral gap exactly: $1-\cos(\pi/n)=\Theta(n^{-2})$ \cite{XXZKomaNachtergaele}.

It is usually difficult to determine the energy gap of a Hamiltonian exactly; here we will prove an inverse polynomial scaling bound. We split our proof into three parts. First, we show an inverse polynomial lower-bound in the fully reducible subspace (Theorem~\ref{thm:PFmodelGapLowerbound}). Second, we find a lower-bound also in the irreducible subspaces (Theorem~\ref{thm:GapLowerBoundIrreducible}). Finally, we show an $\bigO(n^{-2})$
upper-bound on the gap (Theorem~\ref{thm:PFmodelGapUpperBound}). 
This upper-bound rules out the possibility of describing the PF model by a relativistic CFT, whose gap must vanish with $\Theta (n^{-1})$ \cite[p. 412]{CFT}.

\subsubsection{The gap bound strategy and main tools}
Our strategy for bounding the gap is to relate the rewriting Hamiltonian to a stochastic matrix of a Markov chain and use the techniques for analysis of gap (or mixing time) of Markov chains \cite{LevinPeresWilmer2006}\footnote{An alternative option would be to see the rewriting Hamiltonian as a rescaled version of the unnormalized Laplacian of the graph induced by transition rules and use techniques from spectral graph theory \cite{chungspectral}.} -- 
the Canonical path method \cite{DiaconisStroockGeomBounds,Sinclair92improvedbounds} and the Comparison theorem \cite{diaconis1993}. 
In several steps, we will also split our Hamiltonian into two terms and rely on the projection lemma \cite{doi:10.1137/S0097539704445226}, which bounds the smallest eigenvalue of a sum of two Hermitian matrices when one of them is a small perturbation of the other.
Let us review these techniques before using them in the proofs of gap bound Theorems~\ref{thm:PFmodelGapLowerbound}, \ref{thm:GapLowerBoundIrreducible}, and \ref{thm:PFmodelGapUpperBound}.

We will deal with the gap bounds for our rewriting Hamiltonian in each of its invariant subspaces. We will take the restriction of the Hamiltonian and relate it to a particular Markov chain. The state space of the Markov chain $\Omega$ is the set of the basis states (words) spanning the invariant subspace. The chain's transition probabilities are given by a stochastic matrix
$P=\II-\beta H$, with a real positive parameter $\beta$ chosen sufficiently small to make the entries of $P$ nonnegative. This chain has a unique stationary state with a uniform distribution $\pi(x)=1/|\Omega|$. The gap of the Hamiltonian in the restricted subspace $\mathcal{S}$ can then be expressed as
$\Delta(H|_{\mathcal{S}})=(1-\lambda_2(P))/\beta$. The largest eigenvalue of $P$ equals 1 and corresponds to the stationary state, while the difference between the largest and the second largest eigenvalue of $P$ is the spectral gap $\Delta(P)=1-\lambda_2(P)$ of the Markov chain.
Such a ``quantum-to-classical'' mapping holds for all frustration-free stoquastic Hamiltonians \cite{Bravyi:2009:CSF:1898209.1898222}, thanks to the Perron-Frobenius theorem, which guarantees that a ground state can be chosen with real and nonnegative amplitudes in the basis in which the Hamiltonian is {\em stoquastic}. 
\begin{definition}[Stoquastic Hamiltonian \cite{BravyiEtAlStoquastic}]
	A local Hamiltonian is called stoquastic with respect to some basis (usually the standard basis), if and only if its terms have only real and non-positive off-diagonal matrix elements in this basis. 
\end{definition}
\noindent
We can relate such a Hamiltonian 
to a Markov chain with a stationary distribution $\pi(x)=\braket{x}{\psi}^2$ and transition probabilities:
\begin{align}
P(s,t)
=\delta_{s,t}-\beta\sqrt{\frac{\pi(t)}{\pi(s)}}\bra{s}H\ket{t},
\end{align}
and analyze the gap of this Markov chain instead. 
The underlying graph
has vertices corresponding to the states and edges corresponding to possible transitions between states, i.e. $E=\{(a,b):P(a,b)>0\}$.

Our first main technique is the Canonical path theorem. It provides a lower-bound on the gap of a Markov chain. To utilize it, we must design a family of paths between any two different states in the state space along the edges in the underlying graph, and then find a bound on the maximal congestion of any edge.

\begin{theorem}[Canonical path theorem \cite{DiaconisStroockGeomBounds,Sinclair92improvedbounds}] 
\label{thm:canonical}
Let $M$ be a reversible Markov chain over a finite state space $\Omega$ with a stochastic matrix $P$, and the stationary distribution $\pi$. For any choice of canonical paths $\Gamma=\{\gamma_{s,t}\}$ associating a path between every distinct states $s,t\in \Omega$ along the edges of the underlying graph of the Markov chain, the second largest eigenvalue $\lambda_2$ satisfies:
	\begin{align}
	1-\lambda_2(P) \geq \frac{1}{\rho l}, \label{Pgap}
	\end{align}
	where $l=\max_{s,t}|\gamma_{s,t}|$ is the maximum length of a canonical path and $\rho$ is the maximum congestion of an edge, given by
	\begin{align}
	\rho& = \max_{(a,b)\in E}\frac{1}{\pi(a)P(a,b)}\sum_{(a,b)\in\gamma_{s,t}}\pi(s)\pi(t).\label{congestion}
	\end{align}
\end{theorem}
Note that for a uniform distribution $1/|\Omega|$ and a constant probability of transition $p$, the congestion becomes
\begin{align}
\rho=\frac{1}{|\Omega|p}\max_{(a,b)\in E}|\{\gamma_{s,t}:\mbox{ canonical path }\gamma_{s,t}\mbox{ uses the edge } (a,b)\}|.\label{eq:rhoUniformConstProb}
\end{align}

The second importanant theorem in our toolkit for the gap lower-bound proof is the Comparison theorem of Diaconis and Saloff-Coste \cite{diaconis1993}. It compares Dirichlet forms of two reversible Markov chains acting on the same state space, but with a different set of transitions. It can be used to compare their spectral gaps. 
In our case, we will use it on a simple-to-analyze Markov chain $\tilde{M}$, whose ``big-step'' transitions can be decomposed into many ``small-step'' transitions of the Markov chain $M$, whose gap we want to bound. 
\begin{theorem}[Comparison theorem]
\label{thm:comparison} 
Let $P, \pi$ and $\tilde{P},\tilde{\pi}$ be the transition matrices and stationary distributions of two reversible Markov Chains over the same finite state space $\Omega$. Then their gaps are related as
	\begin{align}
	1-\lambda_2 (P)&
	\geq \left(\max_{s\in\Omega}\frac{\tilde{\pi}(s)}{\pi(s)}\right)\frac{1}{\mathcal{A}}(1-\lambda_2(\tilde{P})),
	\end{align}
	where the congestion ratio $\mathcal{A}$ is defined as follows. For any choice of paths $\Gamma=\{\gamma_{s,t}\}$ connecting any two different states $s,t$ such that $\tilde{P}(s,t)>0$ along the edges $E$ of the underlying graph of $P$:	
	\begin{align}\label{ComparisonA}
	\mathcal{A}
	&=\max_{(a,b)\in E}\frac{1}{\pi(a)P(a,b)} \sum_{(a,b)\in\gamma_{s,t}}|\gamma_{s,t}|\tilde{\pi}(s)\tilde{P}(s,t).
	\end{align}
\end{theorem}

Finally, we will utilize a variant of the Projection lemma \cite{doi:10.1137/S0097539704445226}, which bounds the lowest eigenvalue of a sum of two Hamiltonians if one of them has a spectral gap comparatively larger than the norm of the other Hamiltonian. 
It is useful to view the sum of the two terms as a main term and 
a small perturbation.
\begin{theorem}[Projection lemma \cite{doi:10.1137/S0097539704445226}] Consider a Hamiltonian $H=H_1+H_2$. Let $H_2$ have a non-empty zero eigenspace $\mathcal{S}$, and all its other eigenvalues be at least $J>2\|H_1\|$. Then,
	\begin{align}
	\lambda\left(H_1|_\mathcal{S}\right)
	\geq\lambda(H)
	\geq\lambda\left(H_1|_\mathcal{S}\right)-\frac{\|H_1\|^2}{J-2\|H_1\|}.
	\end{align}
\end{theorem}
Note that when we know that $H$ has a zero-energy eigenvector, 
we can use the Projection lemma to lower-bound the second smallest eigenvalue of a sum $H=H_1+H_2$ by considering $H$ restricted to the subspace orthogonal to the ground-space.

\subsubsection{A lower bound on the gap in the fully reducible subspace}
\label{sec:gap1proof}

Our first gap lower bound result comes from the fully reducible subspace. The whole Section~\ref{sec:gap1proof} is the proof of Theorem~\ref{thm:PFmodelGapLowerbound}.

\begin{theorem}[PF-model spectral gap lower-bound in the fully reducible subspace]\label{thm:PFmodelGapLowerbound} The gap of the PF model Hamiltonian for a chain of size $2n$, restricted to the fully reducible subspace (subset of all well nested/fully reducible PF words) is lower bounded by an inverse polynomial in $n$, for local particle dimension $d\geq 3$:
	\begin{align}
	\Delta_{H|_{\varnothing}}=\Omega\left(\frac{1}{n^{15/2}d}\right).
	\label{Hnothingexp}
	\end{align}
\end{theorem}

\begin{proof}
	In this proof, we will map our Hamiltonian to the pair-flip Markov chain, compare it to another, pair-displacing Markov chain, bound its gap, and analyze the relationships between the respective gaps. 
	
	We start by mapping the PF-model Hamiltonian restricted to the good subspace to the {\em pair flip} (PF) Markov chain induced by the transition rules of the PF model, with transition matrix
	\begin{align}
	P_{\textrm{PF}}=\mathbb{I}-\frac{1}{(2n-1)(d-1)}H|_{\varnothing}\label{IdlingChainP}.
	\end{align}
	Let us prove that this matrix indeed describes a reversible Markov chain with a unique stationary distribution. Let $\pi(s)=|\langle s|\psi \rangle|^2=1/W^{(d)}_n$, where $|\psi\rangle$ is the (unique) zero eigenvector of the fully reducible subspace of the PF model and $W^{(d)}_n$ is the number of fully reducible PF-model words with $n$ letter pairs. Then for any $s$,
	\begin{align}
	\sum_{t}P_{\textrm{PF}}(s,t)=1-\frac{1}{(2n-1)(d-1)}\langle s|(H|_{\varnothing})\sum_t |t\rangle=1-0=1,
	\end{align}
	therefore its columns sum to 1. When words $s$ and $t$ are connected by a pair-flip transition rule, there is only one possibility how to apply it, and the transition probability is (the 2 in the denominator comes from the definition of $H$)
	\begin{align}
		P_{\PF}(s,t)=\frac{1}{2(2n-1)(d-1)}.
		\label{PFtransitionexact}
	\end{align}
	Meanwhile, we have $P_{\PF}(s,s)\geq \frac{1}{2}$
	for self loops, since each word is connected to at most $(2n-1)(d-1)$ other words by transition rules and $\bra{s}\left(H|_{\varnothing}\right)\ket{s}\leq (2n-1)(d-1)$.
	Therefore, $P_{\textrm{PF}}$ is a stochastic matrix. Moreover, the uniform distribution $\pi(\cdot)$ is its unique stationary distribution:
	\begin{align}
		\sum_{t}\pi(s)P_{\textrm{PF}}(s,t)
		= \pi(t)
		- \sum\pi(s)\frac{\langle s|(H|_{\varnothing})|t\rangle}{(2n-1)(d-1)}
		= \pi(t).
	\end{align}
	
	The PF Markov chain is just a shift and rescaling of the PF Hamiltonian (restricted to the fully reducible subspace). Thus, we can use the gap $\Delta_{P_{\textrm{PF}}} =(1-\Lambda_2(P_{\textrm{PF}}))$

	calculate the gap of the Hamiltonian:
\begin{align}
	\Delta_{H|_{\varnothing}}=(2n-1)(d-1)\Delta_{P_{\textrm{PF}}}.
	\label{H0gap}
\end{align}

We do not analyze the PF Markov chain directly\footnote{It may turn out that it is possible to use an approach similar to card shuffling or Lozenge tailings.}. Its transitions involve locally recoloring a pair of letters to another pair of letters. Instead, we will turn to a high-level, {\em pair-displacing} (PD) Markov chain (see Definition \ref{PDMarkovchain} and Figure~\ref{fig:PFtoPDCanonicalPaths}), which allows removing an adjacent identical letter pair, inserting an arbitrarily colored pair on an arbitrary location.

\begin{definition}[The pair-displacing (PD) Markov chain] \label{PDMarkovchain}
It is a Markov chain defined over the PF model words $\PF^{(d)}_{n}$, with the following transition procedure:
	\begin{enumerate}
		\item Pick a position $i$ at random from $1$ to $2n-1$.
		If the particles at positions $i,i+1$ don't have the same color, do not change the state.
		\item Otherwise, if there is a pair at positions $i,i+1$, remove it. Then insert a randomly colored pair at a random position between $0$ and $2n-2$.
	\end{enumerate}
\end{definition}

First, let us find some properties of the PD Markov chain. Whenever two distinct words $s$ and $t$ can be obtained from one another by removing a pair and inserting a pair of any color at some position, their probability of transition is at least $P_{\PD}(s,t)\geq\frac{1}{(2n-1)^2d}$,  
as there are at most $2n-1$ positions where a pair could be removed, as well as inserted afterwards, with $d$ possible colors of the inserted pair.
On the other hand, the probability of transition is surely at most $\frac{1}{d}$ (the words differ in at least two letters, so the color of the inserted letters must be correctly chosen).
Thus, 
\begin{align}
	\frac{1}{d}\geq P_{\PD}(s,t)\geq\frac{1}{(2n-1)^2d}\label{eq:PPDbounds}.
\end{align}
These bounds will be useful for calculating the congestion and other properties of the PD chain. 
For completeness, the probability of self loops is $P(s,s)=1-\sum_{s'\not= s}P(s,s')$.

We need to understand the relationship between the {\em pair displacing} (PD) (Definition \ref{PDMarkovchain}) and the pair-flipping (PF) Markov chains. It turns out each transition of the PD chain can be built from a sequence of the PF chain transitions -- simple pair-flips. 
We show this in Lemma \ref{thm:PFtoPDMarkovChain}, as well as that their gaps are related as
	\begin{align}
	\Delta(P_{\PF})\geq\frac{1}{4 (2n-1)^4d^2} \Delta(P_{\PD}),\label{eq:PFgapfromPDgap}
	\end{align}
thanks to the Comparison theorem \cite{diaconis1993} (Theorem \ref{thm:comparison}).
This is similar to how Movassagh \cite{MovFSCGap} compared the gaps of the Fredkin and peak-displacing Markov chains.

	What remains is to analyze the (reversible) PD Markov chain and find a lower bound on its gap via canonical paths. 
	Two distinct words $s,t\in \PF^{(d)}_{n}$ have a transition between them when $s$ can be obtained from $t$ by removing a single adjacent pair from $t$ and inserting another pair of any color somewhere in $s$. 
	We get our inspiration for the lower-bound of the PD Markov chain gap from Section {\em Spectral gap: lower bound} of 
	\cite{CriticalityWithoutFrustration}, and \cite{MovaShor} (the extension to the multicolored case).
	There, a canonical path technique was used to lower bound the gap of a Markov chain on Dyck paths. The canonical paths were constructed by organizing the Dyck paths into a {\em supertree}: the root is the empty path, and the descendants of a vertex are paths that can be obtained by inserting a peak into the parent Dyck path. Moreover, the supertree is built so that no vertex has more than 4 descendants (or $4s$ for the $s$-colored case). The existence of a supertree was shown using a fractional matching theorem: converting a stochastic mapping of descendants to parents
	(a fractional matching between words of length $n$ and $n-1$) into an actual tree.
	
	Here, we also use the canonical path technique \cite{DiaconisStroockGeomBounds,Sinclair92improvedbounds} from Theorem~\ref{thm:canonical}. 
	To use it, we need to design a family 
	of canonical paths $\Gamma=\{\gamma_{s,t}\}$ associating a path along the edges of the underlying graph of the PD chain for every pair $s,t\in \Omega$.
	If we can calculate the congestion \eqref{congestion}, Theorem~\ref{thm:canonical} will give us a lower bound \eqref{Pgap} on the gap of the PD chain. 
	
	\begin{figure}[h]
		\begin{center}
			\includegraphics{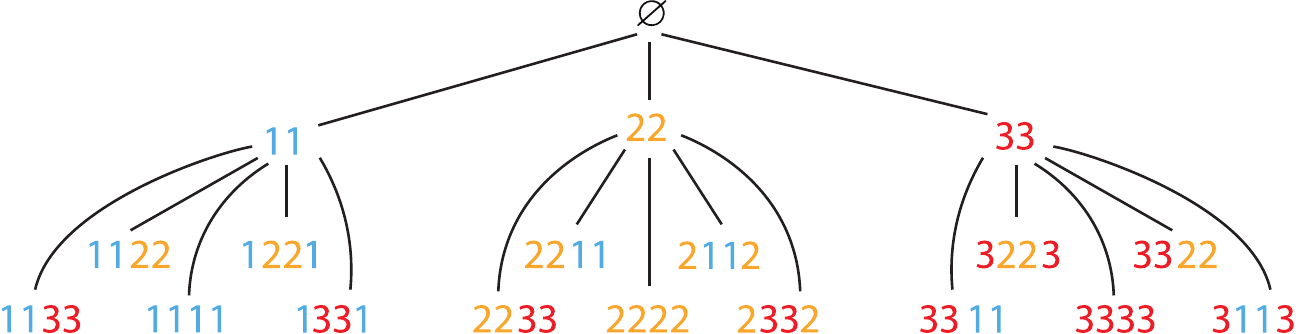}
			\caption{The first 3 levels of the 3-color pair-displacing supertree.
				The 0-th layer (empty) word has 3 children, while all words in layer 1 have 5 children.}\label{fig:PDSupertree}
		\end{center}
	\end{figure}
	
	Let us then construct the family of canonical paths for the $d$-colored \textit{PD Markov chain}. We start by arranging the words in a supertree. The PD supertree is a rooted ordered tree with vertices labeled by $d$-colored $\PF$ words. The $n$-th level of the tree contains all words from $\PF^{(d)}_n$. The root of the PD supertree is the empty word $\varnothing$ -- the only word belonging to $\PF^{(d)}_{0}$. Each descendant of a vertex is a word obtained by inserting a single adjacent pair to its parent word. Moreover, each vertex will have at most $4(d-1)$ descendants. Notice that this is sufficient as we know the upper bound \eqref{Ybounds} on the ratio between $\PF^{(d)}_{n+1}$ and $\PF^{(d)}_{n}$. The first 3 levels of the 3-color $\PD$ supertree are depicted in Figure \ref{fig:PDSupertree}.
	
	We show below in Theorem \ref{lem:PFnandnminus1} that this PD supertree exists: there is a function $g$ from $\PF^{(d)}_{n}$ to $\PF^{(d)}_{n-1}$, such that each element from $\PF^{(d)}_{n-1}$ has at least one and at most $4(d-1)$ pre-images. We assign each vertex representing word $v$ children labeled with words $g^{-1}(v)$.

	Using this supertree, we can then construct a set $\Gamma=\{\gamma_{s,t}\}$  of canonical paths between the vertices of the PD chain. 
	Our goal is to connect any pair of distinct PF Words $s,t\in \PF^{(d)}_{n}$ by a path $\gamma_{s,t}$ going through vertices $a_1,\ldots,a_l$, where $a_1=s$ and $a_l=t$, such that $(a_i,a_{i+1})$ forms an edge in the underlying graph.
	Moreover, we need to be careful so that no edge becomes ``overused'', so that Theorem~\ref{thm:canonical} implies an inverse-polynomial lower bound on the gap.

	Each step of the path $\gamma_{s,t}$ is associated with a word in $\PF^{(d)}_n$, which can be decomposed as $uv$, starting with $u=s$ and $v=\varnothing$ (empty word). We then repeat the following: remove one adjacent pair from $u$ and insert one adjacent pair to $v$, both according to the PD supertree, until we finish with $u=\varnothing$ and $v=t$. More formally, with the help of the PD supertree we generate two sequences
	\begin{enumerate}
		\item  {\em shrinking} the starting word $s\leadsto \varnothing$ with $S=s_{n},s_{n-1},\ldots ,s_{0}$, with $s_i\in\PF^{(d)}_{i}$ generated by going from $s$ towards the root in the supertree, and 
		\item {\em growing} the desired word $\varnothing\leadsto t$ with $T=t_{0},t_{1},\ldots ,t_{n}$, with $t_i\in\PF^{(d)}_{i}$, generated by going from the root to $t$.
	\end{enumerate}
	The steps of the canonical path are concatenations of words from sequences $S$ and $T$: step $i$ is the word 
	$s_{n-i}t_{i}\in \PF^{(d)}_n$.
	This process is illustrated in Figure \ref{fig:PDSupertreeCanonicalPath}.
	
	\begin{figure}[!h]
		\begin{center}
			\includegraphics{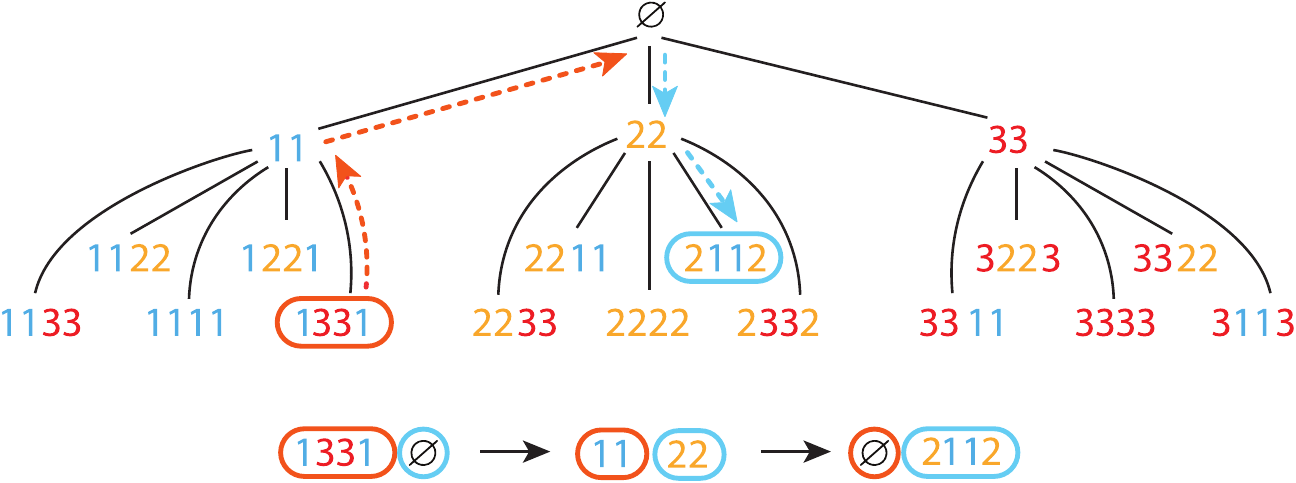}
			\caption{The PD supertree lets us assign canonical paths.
			Here, we depict the path between the words $1331$ and $2112$.
				According to the supertree, it takes 3 steps, to shrink $1331$ to $\varnothing$, and grow $2112$ from $\varnothing$. Concatenating the intermediate words we get the canonical path between the words $1331$ and $2112$. The transformations are valid transitions of the pair-displacing Markov chain -- they erase one pair and insert another pair at another location.
			}\label{fig:PDSupertreeCanonicalPath}
		\end{center}
	\end{figure}
	
	Let us prove a congestion bound for these cannonical paths, i.e. that they don't overuse any edge. Consider an edge on a canonical path $(s_{n-m}t_{m}, s_{n-m-1}t_{m+1})$, for a fixed $m\in\{0,\ldots,n-1\}$. How many canonical paths can use it?
This edge is used for all paths between the descendants of $s_{n-m}$ and $t_{m+1}$. The supertree up to level $n$ is built so that there are at most $(4(d-1))^{m}$ descendants of $s_{n-m}$ and at most $(4(d-1))^{n-m-1}$ descendants of $t_{m+1}$. Therefore, there are at most $(4(d-1))^{n-1}$ canonical paths using this edge. This gives us the following upper bound on the maximal congestion $\rho$:
	\begin{align}
	\rho &\leq \frac{(4(d-1))^{n-1}}{W^{(d)}_{n}P_{\PD}(a,b)}\leq\frac{(4(d-1))^{n-1}d (2n)^2}{W^{(d)}_{n}},
	\end{align}
	where we utilize the lower bound \eqref{eq:PPDbounds} on the transition probability $P_{\PD}(a,b)$.
	The asymptotic scaling of the number of fully-reducible PF model words \eqref{W0asym} of $W^{(d)}_{n}\sim \frac{(4(d-1))^n}{\sqrt{\pi}n^{3/2}} \frac{d(d-1)}{(d-2)^2}$ for $d\geq 3$ then gives us an asymptotic upper bound:
	\begin{align}
		\rho & =\bigO\left(\frac{n^{7/2}(d-2)^2}{(d-1)^2}\right).
	\end{align}
	The length of any path is at most $n$. Thus, \eqref{Pgap} implies an asymptotic lower bound on the gap of the PD chain:
	\begin{align}
		1-\Lambda_2(P_{\PD})&=\Omega\left(\frac{(d-1)^2}{n^{9/2}(d-2)^2}\right).
	\end{align}
	Now, thanks to the relationship between the gaps of the PF and the PD Markov chains \eqref{eq:PFgapfromPDgap} proven below in Lemma~\ref{thm:PFtoPDMarkovChain}, we have:
	\begin{align}
		\Delta(P_\PF)&=\Omega\left( \frac{(d-1)^2}{n^{17/2}d^2(d-2)^2}\right)
		=\Omega\left(n^{-17/2}d^{-2}\right).
	\end{align}
Recalling \eqref{H0gap}, this in turn gives us the asymptotic lower bound on the gap of the $\PF$ Hamiltonian in the fully reducible subspace and completes the proof of Theorem~\ref{thm:PFmodelGapLowerbound}:
	\begin{align}
		\Delta_{H|_{\varnothing}}=\Omega\left(\frac{1}{n^{15/2}d}\right)
		= \frac{1}{\textrm{poly}(n)}.
	\end{align}
\end{proof}

We will now fill the two gaps in the above proof of Theorem~\ref{thm:PFmodelGapLowerbound}. First, in Lemma~\ref{thm:PFtoPDMarkovChain} we relate the gaps of 
the PD Markov chain (with pair-displacing moves) and the gap of the PF Markov chain (with pair-flip moves). Second, in Lemma~\ref{lem:PFnandnminus1} we prove the existence of the desired supertree.

\begin{lemma}[Comparison theorem for the PF and PD Markov chains]
\label{thm:PFtoPDMarkovChain}
The gap of the PF Markov chain is lower bounded by the gap of the PD Markov chain: 
	\begin{align}
		\Delta(P_{\PF})\geq\frac{1}{4 (2n-1)^4d^2} \Delta(P_{\PD}).
		\label{PFtoPD}
	\end{align}
\end{lemma}

\begin{proof}We will use the Comparison theorem of \cite{diaconis1993} (Theorem \ref{thm:comparison})
	to relate the gaps of the two Markov chains. 
	In Theorem \ref{thm:comparison}, let $P$ be the PF chain and the $\tilde{P}$ be the PD chain.
	We want to find an upper bound on the congestion ratio $\mathcal{A}$ \eqref{ComparisonA} for the underlying graph of the PF Markov chain, and use it in the Comparison theorem. 

	We starb ty designing the canonical paths that build the pair-displacing moves of the PD chain from pair-flip moves of the PF chain.
	Let $i,i+1$ be the position of the pair we want to move and $j,j+1$ the position of the new pair. W.l.o.g., we consider $i\leq j$, and deal with the case $i>j$ symmetrically. We now show how to ``bubble'' the original pair from its starting position to the final position and recolor it to the final desired color, only using pair-flips. The path is generated by the following algorithm (illustrated in Figure \ref{fig:PFtoPDCanonicalPaths}): 
	\begin{enumerate}
		\item Start with $k=i$. There is a {\em current pair} at the current position $k,k+1$. Do the following with it:
		\item For any $k<j$, if the color of the {\em current pair} differs from the letter at position $k+2$, recolor the pair to this color.
		This ensures the positions $k+1,k+2$ hold a pair. Increase $k \leftarrow k+1$. If $k<j$, repeat step 2, otherwise go to step 3.
		\item Now we are at $k=j$ (the {\em current pair} is at the end position $j,j+1$). Flip the {\em current pair} to the desired color of the ``inserted'' pair and end the process.
	\end{enumerate}
	
	\begin{figure}[!h]
		\begin{center}
			\includegraphics{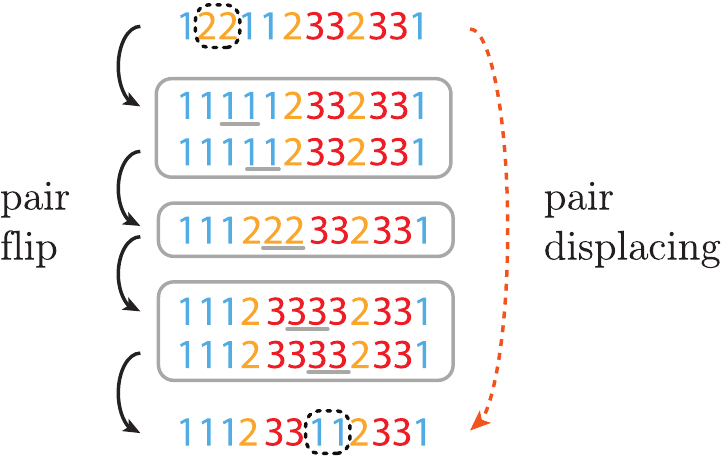}
			\caption{We can build a pair-displacing step (the big jump on the right) from a sequence of pair-flip moves (small arrows on the left). Here, we do this for the PD move $1\underline{22}112332331\, \rightarrow\, 111233\underline{11}2331$.
			} \label{fig:PFtoPDCanonicalPaths}
		\end{center}
	\end{figure}
	
		Let us consider our choice of canonical paths, and find bounds on the most ``overworked'' edge in the underlying graph of the PF Markov chain: $\alpha AAB \beta \leftrightarrow \alpha BBB \beta$ with letters $A \neq B$, and substrings $\alpha, \beta$ of length $|\alpha|=i$ and $|\beta|=2n-i-3$.
	Whenever this edge is part of a canonical path, this path involves a ``moving pair'' which had to start its journey in one of the $|\alpha|+1 \leq 2n-1$ initial positions (in the $\alpha$ substring)
	and has to end in one of the $|\beta|+1\leq 2n-1$ final positions (in the $\beta$ substring), and the initial and final pair could have $d$ possible colors. Therefore, there are at most $2(2n-1)^2 d^2$ canonical paths using any such edge, where the factor two comes from path symmetry. The length of the longest path is at most $2n-1$. 
	Recalling that the Markov chains are reversible, with a uniform unique stationary distribution $\pi(x)=\tilde{\pi}(x)=1/W^{(d)}_{n}$,
	and the bounds on the transition probabilities
	\eqref{PFtransitionexact} and \eqref{eq:PPDbounds},
	we obtain an upper bound on $\mathcal{A}$ in \eqref{ComparisonA}:
	\begin{align}
		\mathcal{A} 
		\leq & \max_{(a,b)\in E} 2(2n-1)(d-1) \sum_{(a,b)\in\gamma_{s,t}}\frac{2n-1}{d}  
		\leq 4 (2n-1)^4 d^2.
	\end{align}
	Theorem \ref{thm:comparison} then implies the claimed relationship\footnote{One might wonder why this bound is a factor of $n^{-1}$ worse than what Movassagh finds in his comparison of the Fredkin and peak-displacing chains \cite[in Lemma 1]{MovFSCGap}. We believe that this is a combination of factors: first, some of the estimates could be tighter, while, second, \cite{MovFSCGap} forgets an $n^{-2}$ factor from the estimation of the transition probability of their PD chain.} \eqref{PFtoPD} between the gaps of the PF an PD chains.
	
\end{proof}

We are ready for the last necessary ingredient for the proof of Theorem~\ref{thm:PFmodelGapLowerbound}, which we now state and prove.  
Lemma~\ref{lem:PFnandnminus1} implies the existence of the pair-displacing supertree with the desired properties (a limited number of children per parent).

First, we recall Lemma 5 of \cite{CriticalityWithoutFrustration}
and restate it as Lemma~\ref{lem:Lemma5Criticality},
as the proof of Lemma~\ref{lem:PFnandnminus1} involves an assignment between $\PF^{(d)}_{n}$ and $\PF^{(d)}_{n-1}$ from a stochastic map. Lemma~\ref{lem:Lemma5Criticality} is about a ``modified'' fractional matching in a bipartite graph $G=(Q\cup R, E)$ -- ``modified'' so that we allow vertices from $Q$ to have at most 4 matching pairs in $R$. We claim that there exists a ``modified'' matching without an exposed vertex if we can find a ``modified'' fractional matching.
\begin{lemma}\cite[Lemma 5.]{CriticalityWithoutFrustration}\label{lem:Lemma5Criticality}
	Let $G=(Q\cup R, E)$
	be a bipartite graph. Let $x=\{x_e\}_{e\in E}$ be a vector of real variables associated with edges; and $\delta(v)$ edges incident with a vertex $v$. Then for a nonempty matching polytope:
	\begin{align}
		\mathcal{P}=&
		\left\{x:x_e\geq 0\textrm{\, for all } e\in E,
		\sum_{e\in\delta(q)}x_e\leq 4, \sum_{e\in\delta(r)}x_e=1
		\textrm{,\, for all\, } q\in Q \textrm{ and } r\in R \right\}
	\end{align} 
	there exists a map $f: R\rightarrow Q$ such that (i) $f(r)=q$ implies $(q,r)\in E$, (ii) any vertex $q\in Q$ has at least one pre-image in $R$, and (iii) any vertex $q\in Q$ has at most four pre-images in $R$.
\end{lemma}
We can now prove the existence of the PD supertree of relationships between fully reducible PF-model words with $d$ colors -- the function $g$ which uniquely maps each descendant in the supertree to its parent word, with no parent having more than $4(d-1)$ children.
\begin{lemma}
\label{lem:PFnandnminus1} 
Let $\PF^{(d)}_{n}$ be the family of all $d$-colored, fully-reducible PF-model words of length $2n$. There exists a map $g: \PF^{(d)}_{n}\rightarrow \PF^{(d)}_{n-1}$ such that:
	\begin{enumerate}
		\item $g(a)=b$ implies we can obtain the word $b$ from the word $a$ by erasing a single (adjacent) pair of identical letters,
		\item any word from $\PF^{(d)}_{n-1}$ has at least 1 and at most $4(d-1)$ pre-images in $\PF^{(d)}_{n}$.
	\end{enumerate}
\end{lemma}
\begin{proof} 
	We will follow the proof of Lemma 3 of \cite{CriticalityWithoutFrustration}, showing the existence of a fractional matching between two sets of vertices. Then we turn to Lemma \ref{lem:Lemma5Criticality}, 
	where one can easily generalize the constant 4, the upper bound on the pre-images to a constant $c\geq 1$, with the corresponding matching polytope. This means a stochastic map $g$ will imply the existence of a function $g$ with the desired properties. 

	We will build the modified fractional matching inductively. On level $n$, we will look at the sets of fully reducible words (vertices) $Q=\PF^{(d)}_{n-1}$ with $n-1$ pairs and $R=\PF^{(d)}_{n}$ with $n$ pairs, and define $g:\PF^{(d)}_{n}\rightarrow \PF^{(d)}_{n-1}$ with the help of previous levels. We define the stochastic map (up to the normalization) as a collection of discrete random variables $g(b)\in \PF^{(d)}_{n}$ for all $b\in \PF^{(d)}_{n-1}$ satisfying:
	\begin{align}
		\Pr[g(b)=a]>0,\label{eq:PrG0}
	\end{align}
	if and only if $a$ can be obtained from $b$ by removing a single adjacent pair and
	\begin{align}
		\sum_{b\in\PF^{(d)}_n} \Pr[g(b)=a] \leq Y^{(d)}_n\leq 4(d-1),\label{eq:PrGUp}
	\end{align}
	We label $Y^{(d)}_n$ the ratio between $|R|$ and $|Q|$, i.e. $Y^{(d)}_n=|\PF^{(d)}_{n}|/|\PF^{(d)}_{n-1}|=W^{(d)}_{n}/W^{(d)}_{n-1}$. 
	We have shown in Appendix~\ref{sec:PF3start} that the ratio of $Y_i^{(d)}$ grows with $i$ \eqref{Yngrow}, 
	and that it is larger than the ratio of the successive number of Dyck paths $X^{(d)}_i$ \eqref{Ybounds}. Thus, for all $i\leq j$,
	\begin{align}
		X_i^{(d-1)}\leq Y_i^{(d)} \leq Y_j^{(d)}\leq 4(d-1). 
	\end{align}

	Let us now choose a stochastic mapping based on the structure of $\PF$ words. Recall that any $d$ colored PF word with $n$ pairs can be split into two blocks as $A\alpha A\beta$.
	The second, possibly empty block $\beta$ 
	is a $d$-colored PF word $\beta\in \PF^{(d)}_{n-i-1}$
	with $i\in\{0,\ldots,n-1\}$.
	The first block $\alpha$
	with $i$ pairs is enclosed within the pair of letters $A$ (one of $d$ possible colors). 
	In Section~\ref{sec:correspond}, we describe in detail how $\alpha$ corresponds to a $d-1$ colored Dyck word $\D(\alpha,A)\in\D^{(d-1)}_i$ \eqref{Dcorrespond}. Basically, we can draw the ``mountain'' profile of $\alpha$, and there are the same number of ways to color this profile, as there are $d-1$ colored Dyck paths with the same ``mountain'' profile.

	For all $d$ colors of the $AA$ pair, we specify $g(b)=a$ inductively for different structures of $b\in\PF^{(d)}_n$ by the following rules
	illustrated in Figure~\ref{fig:stochastic}:
	\begin{enumerate}
		\item\label{enum:gbaseCase} The base case is $n=1$: if $b=AA$, then $g(b)=\varnothing$, with probability 1.
		Here we have $Y_1=d$, as $Q=\{\varnothing\}$ and $R=\{AA,BB,\dots\}$.
		Next, for $n=2$, we have $Y_2 = d+d-1=2d-1$. The assignment is illustrated in Figure~\ref{fig:PDSupertree}.
		We set $q_{0,2}=0$ and $q_{1,2}=1$, resulting in parent assignments $g(ABBA)=AA$ and $g(AABB)=AA$, and exactly $2d-1$ children per parent.
		\item\label{enum:AAbeta} Right-extreme case. If $b=AA\beta$, then $g(b)= AA g(\beta)$, with probability 1. 
		\item\label{enum:AalphaA} Left-extreme case. If $b=A\alpha A$, 
		there is only one block in the word. The substring $\alpha$ corresponds to the Dyck word $\D(\alpha,A)\in \D^{(d-1)}_{n-1}$ with base color $A$ \eqref{Dcorrespond}. As shown in \cite{CriticalityWithoutFrustration}, there exists a stochastic parent assignment function $f$ for Dyck paths, also easily extendable to the colored case.
		Thanks to the 1-to-1 mapping between $\alpha$ and Dyck paths, we can define the parent assignment function $\overline{f}(\alpha,A)=\PF(f(\D(\alpha,A),A)$, using the Dyck path parent assignment $f$, and the map \eqref{PFmap} back from Dyck paths to a PF model block. We thus choose the parent assignment as $g(b)=A\overline{f}(\alpha,A)A$, with probability 1. 
		\item General case. If $b=A\alpha A\beta$ with nonempty $\alpha$ and $\beta$,
		i.e. $|\alpha|=i> 0$ and $|\beta|=n-i-1 > 0$, then
		\begin{enumerate}
			\item $g(b)=A \overline{f}(\alpha,A)A\beta$, with probability $q_{i,n}$,		
			\item and $g(b)=A\alpha Ag(\beta)$, with probability $1-q_{i,n}$.
		\end{enumerate}
		This comes from a word $A\gamma A\delta$ with $|\gamma|=i$ and $\delta =n-i-1$ 
		being assigned a parent $A\overline{f}(\gamma,A)A\delta$ with probability $q_{i,n}$ and a parent $A\gamma A g(\delta)$ with probability $1-q_{i,n}$,
		so that the total probability for having a parent for any word $A\gamma A\delta$ is 1.
	\end{enumerate}

	Note the function $f(\cdot)$ here is the stochastic assignment for $(d-1)$ colored Dyck paths,
	a straightforward generalization of the stochastic mapping $f(\cdot)$ from \cite{CriticalityWithoutFrustration}. The upper bound on its ``incoming'' parent probability is $X^{(d-1)}_{n}=\frac{C^{(d-1)}_n}{C^{(d-1)}_{n-1}}=\frac{4(d-1)(n-1/2)}{n+1}$, with closed-form parentage probabilities $p_{i,n}=\frac{i(i+1)(3n-2i-1)}{n(n+1)(n-1)}$. 
	
	\begin{figure}
		\begin{center}
			\begin{tabular}{cc|c|c}
				\multicolumn{2}{c|}{input $b\in \PF^{(d)}_{n}$}&probability&output $ g(b)\in \PF^{(d)}_{n-1}$\\
				\hline
				$AA\beta$ & & 1& $AAg(\beta)$ \\
				\hline
				$A\alpha A$ & & 1& $A\overline{f}(\alpha,A)A$\\
				\hline
				\multirow{2}{*} {$A\alpha A\beta$,} & \multirow{2}{*} {$|\alpha|=i$ and $0<i<n-1$} & $q_{i,n}$& $A \overline{f}(\alpha,A)A\beta$ \\
				& & $1-q_{i,n}$ & $A\alpha Ag(\beta)$ 
			\end{tabular}
		\end{center}
		\caption{The stochastic map $g(b)=a$ assigning parents to words $b$ is defined probabilistically. 
		Thanks to the correspondence $D(\alpha,A)$ \eqref{Dcorrespond} between the the block $\alpha$ and $d-1$ colored Dyck paths, we can rely on the known assignment of parents for colored Dyck paths $f$ from \cite{CriticalityWithoutFrustration} (extended to the colored case), as $\overline{f}(\alpha,A)=\PF(f(\D(\alpha,A),A)$.}
		\label{fig:stochastic}
	\end{figure}

	Let us now sum the ``incoming'' probabilities for each parent word, and see if we can make these all be smaller than $Y_n^{(d)}$. In fact, we will choose to saturate the inequalities instead. 
	In \cite{CriticalityWithoutFrustration}, we found a closed form solution, while here we just prove the existence of a list of probabilities $\{q_{i,n}\}_{i=0}^{n-1}$ with the desired properties.
	
	The base case was described above. Next, we will show how to inductively choose $q_{1,n},\ldots, q_{n-1,n}\in[0,1]$ so that $g(\cdot)$ satisfies equations \eqref{eq:PrG0} and \eqref{eq:PrGUp}. Let $g(\cdot)$ be defined for all PF words with length $n-1$. 
	Let us take a general parent word $a\in \PF^{(d)}_{n-1}$, see what its children words $b$ could be, and count the total ``incoming'' parentage probabilities from them for the parent $a$.
	In general, any word $a$ with $n-1$ pairs is built from two blocks as $A\alpha A \beta$. The first block, $\alpha$, is made from $i$ pairs, is enclosed in a letter pair $A\cdots A$, and corresponds to a colored Dyck path (word) $\overline{\alpha}$ from $\D^{(d-1)}_i$. The second block, $\beta$, is a fully-reducible PF-model word $\beta\in \PF^{(d)}_{n-i-2}$,
	made from $n-i-2$ pairs. 
	According to the prescription illustrated in Figure~\ref{fig:stochastic}, we have $g(b)=a$ iff
	\begin{enumerate}
		\item $b=A\alpha' A\beta$, when $\alpha'$ is a child of $\alpha$,
		meaning the corresponding Dyck word $\overline{\alpha'}=\D(\alpha',A)\in\D^{(d-1)}_{i+1}$ 
		is a child of the corresponding Dyck word $\overline{\alpha}$, as assigned by $\overline{f}(\overline{\alpha'})=\overline{\alpha}$ from \cite{CriticalityWithoutFrustration} extended to the colored case.
		The probability of assigning $g(b)=a$ will be $q_{i+1,n}$,
		as the block $\alpha'$ is made from $i+1$ pairs.
		\item $b=A\alpha A\beta'$ for $\beta'\in \PF^{(d)}_{n-i-1}$ such that $g(\beta)=\beta'$,
		relying on the choice $g(\beta)=\beta'$ for shorter words ($\beta$ is built from $n-i-2$ pairs).
		The probability of assigning $g(b)=a$ will then be $1-q_{i,n}$, as the block $\alpha$ is made from $i$ pairs.
	\end{enumerate}
	Both events are mutually exclusive. In terms of sums of ``incoming'' probabilities for the parent word $a=A\alpha A\beta$, it means
	\begin{align}
		p_{A\alpha A \beta} 
		= \sum_{b\in PF_n^{(d)}} \Pr[g(b)=A\alpha A\beta]
		=q_{i+1,n}\sum_{\overline{\alpha'}} \Pr[\overline{f}(\overline{\alpha'})=\overline{\alpha}]
		\,+\,(1-q_{i,n})\sum_{\beta'} \Pr[g(\beta')=\beta]. \label{psum}
	\end{align}
	We require $p_{A\alpha A \beta} \leq Y^{(d)}_n$. Thus, we need to show (relying on the levels below as before) that
	our $q_i$'s obey
	\begin{align}
		q_{i+1,n} X^{(d-1)}_{i+1} + (1-q_{i,n})Y^{(d)}_{n-i-1} \leq Y^{(d)}_n.
		\label{qinequality}
	\end{align}
	There are two boundary conditions. First, we need $q_{0,n}=0$, as there is no way to assign a parent to the empty first block $\alpha$. Second, $q_{n-1,n}=1$, as there is only one way to assign a parent to a word of the form $A\alpha A$ with an empty second block.
	
	We will choose to saturate the inequality \eqref{qinequality},
	and show that this is consistent with the boundary conditions, producing a sequence of probabilities $q_{i,n}$.
	We claim that we can start with $q_{i,n}=0$, use
	\begin{align}
		q_{i+1,n} X^{(d-1)}_{i+1} + (1-q_{i,n})Y^{(d)}_{n-i-1} = Y^{(d)}_n, \label{qequality}
	\end{align}
	and calculate a sequence of probabilities $q_{i+1,n}$ for $i=0,\dots,n-2$, ending with $q_{n-1,n}=1$.
	Let us prove that this is indeed what happens,
	as here we do not present a closed-form solution as in \cite{CriticalityWithoutFrustration} for Dyck paths. 
	
	We want to show that $\{q_{i,n}\}_{i=1}^{n-1}$ can be chosen as a sequence of probabilities, i.e. $0\leq q_i \leq 1$, while satisfying the boundary conditions.
		First, we can rewrite
	\eqref{qequality} as
	\begin{align}
		q_{i+1,n} = \frac{1}{X^{(d-1)}_{i+1}}\left(
		Y^{(d)}_n - Y^{(d)}_{n-i-1}
		+ 
		q_{i,n} Y^{(d)}_{n-i-1}
		\right), 
		\label{qeqn}
	\end{align}
	because $Y^{(d)}_n - Y^{(d)}_{n-i-1}\geq 0$ and $q_{i,n}\geq 0$ implying $q_{i+1,n}\geq 0$,
	so the $q_i$'s are nonnegative.

	Second, we show that $q_{i,n}\leq 1$.
	Let us imagine that for some $j$, \eqref{qeqn} produces $q_{j+1,n} > 1$. 
	We could then simply set $q_{j+1,n}=q_{j+2,n} = \dots = 1$,
	which would uniquely assign parents to all children words whose first block has $j'+1 \geq j+1$ pairs, while the parent words with first block with $j' \geq j$ pairs would each have exactly $Y_{j'} \leq Y_n$ children, and thus would not be overworked. 
	However, because $Y_{j<n} < Y_n$, we can conclude that this situation can happen only for $j=n-2$. 
	
	Finally, we realize that $q_{j+1,n}=1$ must happen for $j=n-2$, as then we have $q_{n-1,n}=1$, satisfying the second boundary condition, and saturate all the inequalities, with incoming probability exactly $Y_n$ per parent.
	
	Therefore, a ``modified'' fractional matching of the layers $\PF^{(d)}_{n-1}$ and $\PF^{(d)}_{n}$ with the desired properties exists. Then Lemma~\ref{lem:Lemma5Criticality} implies that such a ``modified'' matching exists too.
	Thus, there is a way to assign the supertree for the pair-displacing chain in such a way that each word has a parent, and no word has more than $4(d-1)$ children.
\end{proof}
	
	For the interested reader, we now show how to calculate the probabilities $q_{i,3}$ at the first nontrivial level ($n-3$) for the modified fractional matching. We can count $W_2 = d(2d-1)$, and $W_3=d(4d^2-3d+2)$, which implies $Y_3=\frac{W_3}{W_2} = \frac{5d^2-6d+2}{2d-1}$.
	In the layer with $n-1=2$ pairs, words of the type $AABB$ must have $d+(d-1)$ children of the type $AABBCC$ and $AABCCB$, as we start with $q_{0,3}=0$. This gives each such word $Y_3 - (2d-1)$ incoming probability to assign from children words of the type $ACCABB$, with $d-1$ choices for the letter $C$. This results in $q_{1,3}=\frac{d-1}{2d-1}$.
	In turn, we have the parent words $ACCA$, which each receive $(1-q_{1,3})$ incoming probability from $d$ words of the type $ACCABB$, leaving them incoming probability $Y_3-d(1-q_{1,3}) = 2(d-1)$ to be assigned. This is exactly what is needed to assign parents to words $ACCDDA$ and $ACDDCA$, which means $q_{2,3}=1$ as desired.


\subsubsection{A lower bound on the gap in the irreducible subspaces.}

In the previous Section, we have shown a gap lower bound for the PF Hamiltonian restricted to the fully reducible subspace. Let us now proceed with the subspaces whose words reduce to a nonempty word and prove Theorem \ref{thm:GapLowerBoundIrreducible}: the PF model Hamiltonian has an inverse polynomial scaling of the energy gap also in these subspaces. 
This proof is inspired by the gap lower bound proof for the ``bad'' subspaces in Motzkin chains \cite{CriticalityWithoutFrustration}.

\begin{theorem}[PF-model gap lower-bound for irreducible subspaces]\label{thm:GapLowerBoundIrreducible} Let $x=X_1X_2\ldots X_k$ be the irreducible string for an irreducible subspace. The PF model Hamiltonian restricted to this subspace has at least an inverse polynomial spectral gap $\Delta(H|_{\textrm{ir:}x})=n^{-\bigO(1)}$.
\end{theorem}

\begin{proof}
	Let us start with a review of the structure of irreducible words. All PF words with the irreducible string $x=X_1X_2\ldots X_k$, made from irreducible letters $X_1,X_2,\ldots, X_k$, can be written as $\alpha X_1 \beta_1 X_2\ldots X_{k} \beta_{k}$, with a ``mountain'' profile as in Figure~\ref{fig:reduction}.
	Here, $\alpha$ is a PF word and each $\beta_1,\ldots \beta_k$ corresponds to a Dyck path (see Section~\ref{sec:correspond} for the precise notion of correspondence),
	as the base layer of the ``mountain'' profile of $\beta_i$ does not contain the letter $X_i$.
		
	Notice that by applying the PF model rewriting rules, the letters of the irreducible string can jump over a long distance, unlike in the Motzkin spin chain or the Fredkin spin chain. There, an irreducible letter can move only by one or two positions, respectively. To illustrate this, let us consider the word $\underline{2}\,1133$ with the irreducible letter $\underline{2}$. If we flip the $33$ to $22$, we obtain $2112\,\underline{2}$, with the irreducible letter jumping all the way to the end.
	
	Each term in the Hamiltonian is a local projector whose terms induce several transitions between states. Observe that the Hamiltonian as a matrix would be the same if it were made from nonlocal projector terms -- one for each specific transition. Let us then analyze this version of the Hamiltonian, and split the terms into three mutually exclusive groups.
	\begin{enumerate}
		\item The first part, $H_{\textrm{ir:\sout{move}}}$, contains all the projectors involving transitions that do not affect the position of the irreducible letters,
		\item the second part, $H_{\textrm{ir:hop}}$, involves transitions that change the position of an irreducible letter by two. We call this part {\em hopping}.
		\item the third part, $H_{\textrm{ir:jump}}$, relates the states that change the position of an irreducible letter by more than two. We call this part {\em jumping}.
	\end{enumerate}
	\begin{align}
	H=H_{\textrm{ir:\sout{move}}}+H_{\textrm{ir:hop}}+H_{\textrm{ir:jump}}
	\end{align}
	Each of the three parts of $H$ is a sum of projectors. 
	Let us remove the jumping term and denote the Hamiltonian without jumping
	\begin{align}
	\Hir=H_{\textrm{ir:\sout{move}}}+H_{\textrm{ir:hop}}.
	\end{align}
	Observe that $H$ and $\Hir$ have the same ground state (the uniform superposition of all states with the same irreducible string), as any big jump can be decomposed into small hops. 
		By removing the jumping part, we are not changing the ground state.
		Moreover, we can only lower the energy of the excited states.
		Effectively, this removes connections in the underlying graph of transitions (while keeping the whole graph in one connected component). 
	
	Next, let us bound the gap of $\Hir$ by treating the hopping term as a small perturbation. For all $0<\epsilon\leq1$ we define the perturbed Hamiltonian:
	\begin{align}	
		\Hir^\epsilon=H_{\textrm{ir:\sout{move}}}+\epsilon H_{\textrm{ir:hop}}.
	\end{align}
	Clearly, $H_{\textrm{ir}}$ and $H_{\textrm{ir}}^\epsilon$ have the same ground state, as both $H_{\textrm{ir}}$ and $H_{\textrm{ir}}^\epsilon$ are sums of the same projectors, just with different prefactors.
	Because $H_{\textrm{ir}}\geq H_{\textrm{ir}}^\epsilon$ (i.e. $H_{\textrm{ir}}- H_{\textrm{ir}}^\epsilon$ is positive),
	their gaps are related as $\Delta(H_{\textrm{ir}})\geq\Delta(H_{\textrm{ir}}^\epsilon)$.
	
	Let us now analyze the gap of $H^0_\textrm{ir}=H_{\textrm{ir:\sout{move}}}$, the Hamiltonian without the perturbation term. The positions of the irreducible letters are not affected by the unperturbed Hamiltonian, therefore the ground subspace of the $H_{\textrm{ir}}^0|_{\textrm{ir:}x}$ in the subspace of the states with the irreducible string $x=X_1\ldots X_k$ is spanned by the normalized states:
	\begin{align}
		\ket{\psi_{x_1,x_2,\ldots,x_k}}
		&=\ket{\PF^{(d)}_{x_1-1}}\otimes\ket{X_1}_{x_1}\otimes\ket{\PF(\D^{(d-1)}_{x_2-x_1-1},X_1)}\otimes\ket{X_2}_{x_2}\ldots\otimes\ket{X_k}_{x_k}\ket{\PF(\D^{(d-1)}_{n-x_k},X_k)},\label{eq:Hir0Gspace}
	\end{align}
	where $x_i\in\{1,\ldots,n\}$ is the position of the $i$th irreducible letter and $\ket{\PF^{(d)}_{m}}$ and $\ket{\PF(\D^{(d-1)}_{m},X)}$ are normalized states denoting uniform superpositions over words from $\PF^{(d)}_{m}$ and the corresponding Dyck words from $\PF(\D^{(d-1)}_{m},X)$, respectively. The spectral gap of $H_{\textrm{ir}}^0|_{\textrm{ir:}x}$ can be computed separately on the intervals between the irreducible letters. 
	We have shown in Theorem~\ref{thm:PFmodelGapLowerbound}
	that the spectral gap of the PF Hamiltonian in the fully reducible subspace is lower bounded by an inverse polynomial.
	The gap of the PF Hamiltonian acting on the subclass of fully reducible words corresponding to $d-1$ colored Dyck walks is also lower-bounded by an inverse polynomial, using a mapping to the peak-displacing Markov chain \cite{CriticalityWithoutFrustration, MovaShor}. Therefore we conclude that $\Delta(H_{\textrm{ir}}^0|_{\textrm{ir:}x})=n^{-\bigO(1)}$.
	
	Let us now turn on the perturbation term by setting $\epsilon>0$. The first order effective Hamiltonian acting on the ground subspace \eqref{eq:Hir0Gspace} describes a weighted non-crossing hopping of the $k$ irreducible letters. The blocks between the irreducible letters are balanced words (of even length) and the perturbation allows hopping of the irreducible letters by two, changing the number of pairs of the two neighboring balanced words by one. We can thus consider the effective Hamiltonian $H_\textrm{eff}$ acting as the hopping of $k$ irreducible letters (particles) on the chain of length $l=(2n-k)/2+k$.
	The hopping amplitudes depend on the distances to the neighboring irreducible letters.
	
	We define the effective Hamiltonian over basis states $\ket{j_1}\ket{j_2}\ldots\ket{j_k}$, where $j_i\in \{1,\ldots, l\}$, and $j_i<j_{i+1}$ for all $i\in\{1,\ldots,k-1\}$. The value $j_i$ represents the position $x_i$ of the irreducible letter $X_i$, and only the number of pairs in the balanced words between the irreducible letters need to be considered. To simplify the notation, let us label the empty space the particle $i$ has to the left by $s_i$ and $s_{k+1}$ the empty space the particle $k$ has to the right. Recalling the structure of the irreducible PF words, the substring before the first irreducible letter is just a fully reducible PF word and the subsequent substrings correspond to Dyck walks. Therefore, we can split the Hamiltonian into the hopping of the first and the remaining ($2$nd, \dots, $k$th) irreducible letters.
	\begin{align}\label{eq:Heff}
		H_{\textrm{eff}}&=\Gamma'+\sum^{k}_{i=2} \Gamma_{i},
	\end{align}
	where $\Gamma'=\ketbra{\gamma'}$ and $\Gamma_i=\ketbra{\gamma_i}$ are projectors onto the states:
	\begin{align}
		\ket{\gamma'}_{1,2}
		  &=a_{s_2}\ket{j_1\otimes j_2}-a_{s_1+1}'\ket{(j_{1}+1)\otimes j_{2}},
		\label{alphas}\\
		\ket{\gamma_i}_{i, i+1, i+2} 
		  &=a_{s_{i+2}}\ket{j_{i}\otimes j_{i+1}\otimes j_{i+2}}-a_{s_{i+1}+1}\ket{j_{i}\otimes (j_{i+1}+1)\otimes j_{i+2}},
		\quad \mbox{ for } 2\leq i\leq k-1, \label{alpha2}\\
		\ket{\gamma_k}_{k-1,k}
	  	 &=a_{s_{k+1}}\ket{j_{k-1}\otimes j_k}-a_{s_{k}+1}\ket{j_{k-1}\otimes (j_{k}+1)}, \label{alpha3}
	\end{align}
	where the amplitudes in \eqref{alpha2} and \eqref{alpha3} are the same $a$'s as in \eqref{alphas}, thanks to the structure $\alpha X_1 \beta_1 X_2\ldots X_{k} \beta_{k}$ of the PF words with irreducible strings.
	
	Let us calculate the amplitudes $a_m$ and $a'_m$ for the first projector $\Gamma'$, involving the hopping of the irreducible letter $X_1$. The projector $\Gamma'$ is the effective term coming from transitions connecting PF model words $\alpha X_1AA\beta\gamma$ and $\alpha X_1X_1X_1 \beta\gamma$, where $A\in\{1,\ldots,d\}\setminus X_1$, and $\gamma$ is a substring starting with the second irreducible letter $X_2$ (if there is one), or the empty word. The position of the letter $X_1$ is $x_1$ and the number of pairs in $\alpha$ and $\beta$ are $s_1$ and $s_2-1$, respectively. Let us denote the part of $H_{\textrm{ir:hop}}$ connecting hopping of the $i$th irreducible letter between positions $x_i$ and $x_i+2$ by $H_{\textrm{ir:hop}}^{i,x_i}$. 
	Recalling the normalized states \eqref{eq:Hir0Gspace},
	we can compute the amplitudes from the expectation values of the respective hopping terms:
	\begin{align}
		\bra{\psi_{x_1,x_2,\ldots,x_k}}H_{\textrm{ir:hop}}^{1,x_1}\ket{\psi_{x_1,x_2,\ldots,x_k}}&=(d-1)\frac{C^{(d-1)}_{s_2-1}}{2C_{s_2}^{(d-1)}},\\
		\bra{\psi_{x_1+2,x_2,\ldots,x_k}}H_{\textrm{ir:hop}}^{1,x_1}\ket{\psi_{x_1+2,x_2,\ldots,x_k}}&=(d-1)\frac{W^{(d)}_{s_1}}{2W^{(d)}_{s_1+1}},\\
		\bra{\psi_{x_1,x_2,\ldots,i_k}}H_{\textrm{ir:hop}}^{1,x_1}\ket{\psi_{x_1+2,x_2,\ldots,i_k}}&=-(d-1)\sqrt{\frac{W^{(d)}_{s_1}}{2W^{(d)}_{s_1+1}}\frac{C^{(d-1)}_{s_2-1}}{2C_{s_2}^{(d-1)}}},\\
		\bra{\psi_{x_1+2,x_2,\ldots,i_k}}H_{\textrm{ir:hop}}^{1,x_1}\ket{\psi_{x_1,x_2,\ldots,i_k}}&=\bra{\psi_{x_1,x_2,\ldots,i_k}}H_{\textrm{ir:hop}}^{1,x_1}\ket{\psi_{x_1+2,x_2,\ldots,i_k}}.
	\end{align}
	This means the amplitudes in \eqref{alphas}, as well as in \eqref{alpha2} and $\eqref{alpha3}$, are 
	\begin{align}
		a_{m}&=\sqrt{\frac{(d-1)C^{(d-1)}_{m-1}}{2C^{(d-1)}_{m}}} \qquad a_{m}'=\sqrt{\frac{(d-1)W^{(d)}_{m-1}}{2W^{(d)}_{m}}}.
	\end{align}
The unique, zero-energy ground state of the effective Hamiltonian can then be written as
\begin{align}
	\ket{g}=\frac{1}{\sqrt{W^{(d)}_{l,k}}}\sum_{j_1=1}^{l-k}\sum_{j_2=j_1+1}^{l-k+1}\sum_{j_3=j_2+1}^{l-k+2}\ldots\sum_{j_k=j_{k-1}+1}^{l}\sqrt{W^{(d)}_{s_1}C^{(d-1)}_{s_2}\ldots C^{(d-1)}_{s_{k+1}}}\ket{j_1}\otimes\ket{j_2}\otimes\ldots\otimes\ket{j_k}.
\end{align}
	
	We apply the projection lemma \cite{doi:10.1137/S0097539704445226} to the orthogonal complement of the PF state with the irreducible string $x$ in the irreducible subspace $\textrm{ir:}x$ and receive:
	\begin{align}
		\Delta(H^\epsilon_{\textrm{ir}}|_{\textrm{ir:}x})
		\geq\epsilon \Delta(H_{\textrm{eff}})-\frac{\bigO(\epsilon^2)\|H_{\textrm{ir:hop}}|_{\textrm{ir:}x}\|^2}{\Delta(H^0_{\textrm{ir}}|_{\textrm{ir:}x})-2\epsilon\|H_{\textrm{ir:hop}}|_{\textrm{ir:}x}\|}.
	\end{align}
	We can choose\footnote{Due to the large exponent in the inverse polynomial lower bound \eqref{Hnothingexp}, $\epsilon$ needs to be rather small.} $\epsilon$ as $n^{-\bigO(1)}$ so that $2\epsilon\|H_{\textrm{ir:hop}}|_{\textrm{ir:}x}\|\ll\Delta(H^0_{\textrm{ir}})$ to receive:
	\begin{align}
		\Delta(H^\epsilon_{\textrm{ir}}|_{\textrm{ir:}x})
		\geq\epsilon\Delta(H_{\textrm{eff}})-\bigO(\epsilon^2)n^{\bigO(1)}.
	\end{align}
	To finish the proof, we now need to find an inverse polynomial lower-bound on the gap of the hopping Hamiltonian $H_\textrm{eff}$.
	
	First we map the effective Hamiltonian to the stochastic matrix describing the following {\em PF weighted hopping (PFH)} Markov chain: 
	\begin{align}
		P_{\textrm{PFH}}((j_1,\ldots,j_k),(j_1',\ldots,j_k'))&=\delta_{(j_1,\ldots,j_k),(j_1',\ldots,j_k')}-\frac{1}{ 2k}\sqrt{\frac{\pi(j_1',\ldots,j_k')}{\pi(j_1,\ldots,j_k)}}\bra{j_1\otimes\ldots\otimes j_k}H_\textrm{eff}\ket{j_1'\otimes\ldots\otimes j_k'}.
	\end{align}
	Its stationary state is $\pi(j_{1},\ldots,j_k)=\braket{j_1,\ldots, j_k}{g}^2$. Since $\ket{g}$ is the unique ground state of $H_\textrm{eff}$, each row sums to 1 and $\pi(\cdot)$ is a unique stationary state. By simple calculations, we receive the following probabilities. It is easy to see that $0\leq P_{\textrm{PFH}}((j_1,\ldots,j_k),(j_1',\ldots,j_k'))\leq 1$ and $P_{\textrm{PFH}}((j_1,\ldots,j_k),(j_1,\ldots,j_k))\geq\frac{1}{2}$.
	
The transition probabilities for the hopping of the first irreducible letter ($i=1$) are the most interesting, as it has a PF word on the left, and a Dyck path on the right:
	\begin{align}
		P_{\textrm{PFH}}((j_1,\ldots,j_k),(j_1+1,\ldots, j_k))& =0-\frac{1}{4k}\frac{\braket{(j_1+1)\otimes\ldots\otimes j_k}{g}}{ \braket{j_1\otimes\ldots\otimes j_k}{g}}\bra{j_1\otimes\ldots\otimes j_k}H_{\textrm{eff}}\ket{(j_1+1)\otimes\ldots\otimes j_k} \nonumber\\
		& =\frac{1}{4k}\frac{\sqrt{W^{(d)}_{j_1}C^{(d-1)}_{j_{2}-j_{1}-2}}}{\sqrt{W^{(d)}_{j_1-1}C^{(d-1)}_{j_{2}-j_{1}-1}}}a_{s_2}a'_{s_1+1}\nonumber\\
		& =\frac{1}{4k}\frac{\sqrt{W^{(d)}_{j_1}C^{(d-1)}_{j_{2}-j_{1}-2}}}{\sqrt{W^{(d)}_{j_1-1}C^{(d-1)}_{j_{2}-j_{1}-1}}}\sqrt{\frac{(d-1)C^{(d-1)}_{j_2-j_1-2}}{2C^{(d-1)}_{j_2-j_1-1}}}\sqrt{\frac{(d-1)W^{(d)}_{j_1-1}}{2W^{(d)}_{j_1}}}\nonumber\\
		& =\frac{d-1}{4k}\frac{C^{(d-1)}_{j_{2}-j_{1}-2}}{C^{(d-1)}_{j_{2}-j_{1}-1}}=\frac{d-1}{4k}\frac{C^{(d-1)}_{s_2-1}}{C^{(d-1)}_{s_2}}.
	\end{align}
The remaining derivations are very similar; we only provide the results.
The probability of transition for the jump of the first irreducible letter to the left is
\begin{align}
	P_{\textrm{PFH}}((j_1+1,\ldots,j_k),(j_1,\ldots,j_k))&=\frac{d-1}{4k}\frac{W^{(d)}_{s_{1}-1}}{W^{(d)}_{s_{1}}}, \label{e89}
\end{align}
and	for $2\leq i\leq k$, the transition probabilities for an irreducible letter between two blocks that correspond to Dyck-paths are:
\begin{align}
	P_{\textrm{PFH}}((j_1,\ldots,j_i,\ldots,j_k),(j_1,\ldots,j_i+1,\ldots,j_k))
	&=\frac{d-1}{4k}\frac{C^{(d-1)}_{s_{i+1}-1}}{C^{(d-1)}_{s_{i+1}}},\\
	P_{\textrm{PFH}}((j_1,\ldots,j_i+1,\ldots,j_k),(j_1,\ldots,j_i,\ldots,j_k))&=\frac{d-1}{4k}\frac{C^{(d-1)}_{s_i-1}}{C^{(d-1)}_{s_i}}. \label{e91}
\end{align}

	The way the first irreducible letter hops here is different from the rest, see \eqref{e89} \& \eqref{e91}. However, we can compare the properties of this chain to another one, where the all the particles behave the same.
 In other words, let us now consider the {\em Catalan weighted hopping (CH)} chain, hopping on the same state space, weighted by colored Catalan weights, with $C$'s instead of $W$'s in \eqref{e89}. Note that this analysis is also applicable for a lower bound on the gap in the irreducible subspaces of the Fredkin chain, which was not explicitly performed in \cite{MovFSCGap},
but it would be necessary if one wanted to break the ground state degeneracy by peak-counting instead of endpoint terms.

	Let us denote the transition matrix of the Catalan weighted chain by $P_{\textrm{CH}}$ and its stationary state $\pi_{\textrm{CH}}(\cdot)$.
	Notice that the stationary state $\pi_{\textrm{CH}}(j_1,j_2,\ldots j_{k})=C_{s_1} C_{s_2}\ldots C_{s_{k+1}}/C_{l,k}$ and the transition probabilities  are the same for the colored and uncolored case, thanks to \eqref{countcoloredCn}: $\frac{1}{4k} (d-1)C_{m-1}^{(d-1)}/ C_m^{(d-1)}=\frac{1}{4k} C_{m-1}/ C_m$.
We can compare the transition probabilities \eqref{e89} of PF weighted hopping to Catalan weighted hopping, thanks to the relationship between $d-1$ colored Catalan numbers and $d$ colored PF numbers  \eqref{Vnkbounds},
with the constant of proportionality 
$c=\frac{d(d-1)}{(d-2)^2}$:
	\begin{align}
		\frac{1}{c}\,
		P_{\textrm{CH}}(s,t)
		\leq
		& 
		P_{\textrm{PFH}}(s,t)
		\leq c\, 
		P_{\textrm{CH}}(s,t),
		\label{eq:RelatingProbabilitiesPFHoppingCatalanHopping}
	\end{align}
i.e. the transition probabilities of PF weighted hopping are between $c^{-1}$ and $c$ times the transition probabilities of Catalan weighted hopping.
Their stationary states must then also be similar, as this implies
	\begin{align}
		\frac{1}{c}\leq\frac{\pi_{\textrm{PFH}}(j_1,\ldots,j_k)}{\pi_{\textrm{CH}}(j_1,\ldots,j_k)}\leq c
		.
		\label{eq:ComparisonPFHoppingCatalanHoppingPi}
	\end{align}

We can learn something about Catalan weighted hopping
from a similar Markov chain with the same stationary state, analyzed in \cite{MartinRandallDecompositionAdsorbingStaircaseWalk2006,MartinRandallAdsorbingStaircaseWalks} -- let's call it {\em Catalan weighted Metropolis hopping (CWMH)}. Its spectral gap was lower-bounded by $\Omega(n^{-19/2})$ \cite[Lemma 6.8.]{MartinRandallDecompositionAdsorbingStaircaseWalk2006}. 

The CWMH has Metropolis transition probabilities $P_{\textrm{CWMH}}(s,t)=\frac{1}{4k}\min\left\{1,\pi_{\textrm{CH}}(t)/\pi_{\textrm{CH}}(s)\right\}$ for two states $s$ and $t$ differing by a hop of a single particle, 0 for states that differ more, and self loop probabilities $P_{\textrm{CWMH}}(s,s)=1-\sum_{t\not=s}P_{\textrm{CWMH}}(s,t) \geq\frac{1}{2}$. As we mentioned earlier, the stationary state of this Markov chain is 
	$\pi_{\textrm{CWMH}}(s)=\pi_{\textrm{CH}}(s)$.
	Let us compare these transition probabilities with those of Catalan weighted hopping. Consider two states $s=(j_1,\ldots ,j_i,\ldots j_{k})$ and $t=(j_1,\ldots ,(j_i+1),\ldots j_{k})$, which differ only by the position of particle $i$. The probability of transition is
	$P_{\textrm{CH}}(s,t)=\frac{1}{4k}\min\left\{1,\frac{C_{s_i+1}C_{s_{i+1}-1}}{C_{s_i}C_{s_{i+1}}}\right\}$, similarly for $P_{\textrm{CH}}(t,s)$. 
	Because Catalan numbers satisfy $C_{n+1}\leq 4 C_n$, we get
\begin{align}
	\frac{1}{4}P_{\textrm{CWMH}}(s,t) 
		\leq P_{\textrm{CH}}(s,t),\qquad 
	\frac{1}{4}P_{\textrm{CWMH}}(t,s) 
		\leq P_{\textrm{CH}}(t,s).
	\label{eq:CatalanWHoppingAndCatalanBlockWeightedHopping}
\end{align}
	These bounds will allow us to relate the gap of the PF weighted hopping chain to the gap of the CWMH chain.
	
		The spectral gap of a reversible Markov chain satisfies \cite{LevinPeresWilmer2006}:
\begin{align}
	\Delta = 1-\lambda_2
	= \inf_{\substack{f:\Omega\rightarrow \R,\\ \Var_\pi(f)\not=0}}\frac{\DirichletForm_\pi(f,f)}{\Var_\pi(f)}
	= \inf_{\substack{f:\Omega\rightarrow \R,\\ 
	\Var_\pi(f)\not=0}}\frac{\frac{1}{2}\sum_{x,y\in\Omega}(f(x)-f(y))^2\pi(x)P(x,y)}{\frac{1}{2}\sum_{x,y\in\Omega}(f(x)-f(y))^2\pi(x)\pi(y)}.
	\label{Diriform}
\end{align}
Let $\Delta_{\textrm{CWMH}}$ be the spectral gap of the CWMH.
We will show that 
$\Delta_{\textrm{PFH}}\geq \frac{1}{4 c^4}\Delta_{\textrm{CWMH}}$. 
Let $\DirichletForm_{\pi}(f,f)$ and $\tilde{\DirichletForm}_{\tilde{\pi}}(f,f)$ be two Dirichlet forms associated with the PF weighted hoping and CWMH chains, respectively. Combining \eqref{eq:RelatingProbabilitiesPFHoppingCatalanHopping}, \eqref{eq:CatalanWHoppingAndCatalanBlockWeightedHopping}, and \eqref{eq:ComparisonPFHoppingCatalanHoppingPi} for any $x,y\in\Omega$ we find 
$\frac{1}{4c^2}\pi_{\textrm{CWMH}}(x)P_{\textrm{CWMH}}(x,y)\leq\pi_{\textrm{PFH}}(x)P_{\textrm{PFH}}(x,y)$ 
and 
$c^2\, 
	\pi_{\textrm{CWMH}}(x)\pi_{\textrm{CWMH}}(y)\geq 
		\pi_{\textrm{PFH}}(x)\pi_{\textrm{PFH}}(y)$. 
	Hence for all $f:\Omega\rightarrow\R$ such that $\Var_{\tilde{\pi}}(f)\not=0$ holds, we have
\begin{align}
	\frac{1}{4c^4}\frac{\tilde{\DirichletForm}_{\tilde{\pi}}(f,f)}{\Var_{\tilde{\pi}}(f)}\leq \frac{\DirichletForm_{\pi}(f,f)}{\Var_\pi(f)},
\end{align}
and thus also $\Delta_{\textrm{PFH}}\geq \frac{1}{4 c^4}\Delta_{\textrm{CWMH}}$, thanks to \eqref{Diriform}.
	
	Plugging in the inverse-polynomial lower bound on the CWMH gap \cite[Lemma 6.8.]{MartinRandallDecompositionAdsorbingStaircaseWalk2006} completes the proof. The spectral gap of the PF model Hamiltonian in any  unbalanced subspace labeled by the irreducible string $x$ is thus lower bounded by an inverse polynomial $\Delta(H|_{\textrm{ir:x}})=n^{-\bigO(1)}$.
\end{proof}

\subsubsection{An upper bound on the gap}

In this Section, we prove an $\bigO(n^{-2})$ upper bound on the gap of the PF model. As mentioned in \cite{MovaShor}, such a system can not be described by a relativistic CFT -- for that, its gap would have to vanish as $\Delta(H) = \frac{2\pi\Delta}{n}$, with $\Delta$ the scaling dimension \cite[p. 412]{CFT}.

Our upper-bound is analogous to the approach used for the Fredkin spin chain in \cite{MovFSCGap} and 
Motzkin spin chain in \cite{MovaShor}. The trick is to ``twist'' the uniform superposition of the individual Dyck walks (Motzkin walks) by phases proportional to the area under them, 
and show that the resulting state has low-enough energy and low overlap with the original ground state. 
In \cite{MovFSCGap, MovaShor}, the proof is based on the universality of Brownian motion, and on the convergence of Dyck random walk to the Brownian excursions. Here, we analyze the properties the PF Hamiltonian and utilize the comparison of PF words and $d-1$ colored Dyck paths \eqref{Vnkbounds}.

\begin{theorem}[PF-model gap upper-bound in the fully reducible subspace]\label{thm:PFmodelGapUpperBound} 
The gap of the pair-flip model is upper bounded by the gap in its fully reducible subspace, which is upper bounded as 
$\Delta(H|_{\varnothing})=\bigO(n^{-2})$.
\end{theorem}
\begin{proof}
	Similarly as in \cite{MovFSCGap} and \cite{MovaShor}, we will choose a state $\ket{\phi^{(d)}_n}$ with a small overlap with the PF model ground state $|\braket{\phi^{(d)}_n}{\PF^{(d)}_n}|^2\leq \kappa<1$, where $\kappa=\bigO(1)$.
	We will show that $\bra{\phi^{(d)}_n}(H|_{\varnothing})\ket{\phi^{(d)}_n}=\bigO(n^{-2})$, which will imply an upper bound on the gap $\Delta(H|_{\varnothing})\leq (1-\kappa)^{-1}\bra{\phi^{(d)}_n}H|_{\varnothing}\ket{\phi^{(d)}_n}$.
	
	The state used to upper bound the gap of Fredkin spin chain in \cite{MovFSCGap} was:
	\begin{align}
		\ket{\phi^{(d-1)}_{\textrm{FR},n}}=\frac{1}{\sqrt{C^{(d-1)}_n}}\sum_{p\in \D^{(d-1)}_n}e^{2\pi i \tilde{A}_p\tilde{\theta}}\ket{p},
		\label{FRtwist}
	\end{align}
	where the $\tilde{A}_p$ was the area below the Dyck walk $p$, and $\tilde{\theta}$ a constant set to $\frac{n^{-3/2}}{\sqrt{10/3-\pi}}$.
	For the PF model, we will use a similar state, made from PF words, and involving the area under the corresponding mountain profile of the words (see Figure~\ref{fig:wordspathstrees}).
	\begin{align}
		\ket{\phi^{(d)}_{n}}=\frac{1}{\sqrt{W^{(d)}_n}}\sum_{p\in \PF^{(d)}_n}e^{2\pi i \tilde{A}_p\tilde{\theta}}\ket{p}.
	\end{align}
Let us calculate its overlap with the PF ground state. We split $\ket{\phi^{(d)}_{n}}$ into two parts: 
\begin{enumerate}
\item States corresponding to $d-1$ colored Dyck walks -- restricting to PF states which contain only $d-1$ colors on the base level $\ket{(\PF^{(d)}_n|_{\D})}$.
\item States corresponding to the rest of the PF words $\ket{(\PF^{(d)}_n|_{\PF\setminus \D})}$. 
\end{enumerate}
We will now show that $|\braket{\phi^{(d)}_{n}}{\PF^{(d)}_n}|^2 \leq \kappa<1$ and $\kappa\in\bigO(1)$.
\begin{align}
	|\braket{\phi^{(d)}_{n}}{\PF^{(d)}_n}|
	&\leq |\braket{\phi^{(d)}_{n}}{(\PF^{(d)}_n|_{\D})}|
	 + |\braket{\phi^{(d)}_{n}}{(\PF^{(d)}_n|_{\PF\setminus \D})}|\\
	&= \frac{1}{W^{(d)}_n}\left(
		\left|\sum_{p\in \D^{(d-1)}_n} e^{2\pi i \tilde{A}_p\tilde{\theta}}\right|
		+ \left|\sum_{p\in \PF_n^{(d)}\setminus\{\PF(x,A):x\in \D^{(d-1)}_n\}} e^{2\pi i \tilde{A}_p\tilde{\theta}}\right|
		\right)\nonumber\\
	&\leq \frac{1}{W^{(d)}_n}\left(C^{(d-1)}_n|\braket{\phi^{(d-1)}_{\textrm{FR},n}}{\D^{(d-1)}_n}| + (W_n^{(d)}-C_n^{(d-1)}) \right)\nonumber\\
	&= 1 - \frac{C_{n}^{(d-1)}}{W_n^{(d)}}
		\left(1-|\braket{\phi^{(d-1)}_{\textrm{FR},n}}{\D^{(d-1)}_n}|\right)
		\leq 1- \frac{(d-2)^2}{d(d-1)} \left(1- \sqrt{\kappa_{\textrm{FR}}}\right), \nonumber
	\end{align}
where we rely on \eqref{Vnkbounds} and recall that $\kappa_{\textrm{FR}}<1$ is a constant upper bound on the magnitude of the overlap of the state \eqref{FRtwist} and the Fredkin chain ground state.
Thus, $|\braket{\phi^{(d)}_{n}}{\PF^{(d)}_n}|^2$ is upper bounded by some constant $\kappa$ strictly less than 1.

	Let us explore the expectation value of the energy of the state $\ket{\phi^{(d)}_n}$. The Hamiltonian $H|_{\varnothing}$ relates two words if and only if they differ by applying one of the transition rules. The local rules can leave the area under the walk the same (when we recolor a peak), or change it up or down by two (when a peak becomes a valley, as in $1221 \leftrightarrow 1111$). If the area is the same and the words differ just by recoloring, the contribution to the expected value is 0. For each area change by two, we receive the contribution of $1-\cos{4\pi\tilde{\theta}}$ to the expectation value. To see that, consider two words with a different area under the PF walk connected by the transition rule $AA\leftrightarrow BB$ with the Hamiltonian term $\frac{1}{2}\ketbra{AA-BB}$. The number 1 is for connecting the word containing $AA$ (the same for $BB$) to itself. The term $-\cos{4\pi\tilde{\theta}}=-\frac{1}{2}e^{4 i\pi\tilde{\theta}}-\frac{1}{2}e^{-4i\pi\tilde{\theta}}$ comes from connecting the word with $AA$
	to the word with $BB$, with a different area under the paths, $-\frac{1}{2}e^{4 i\pi\tilde{\theta}}$ for increasing and $-\frac{1}{2}e^{-4i\pi\tilde{\theta}}$ for decreasing the area by two.
\begin{align}
	\sum_{j=1}^{2n-1}\bra{\phi^{(d)}_n}H_{j,j+1}\ket{\phi^{(d)}_n}
	&= \frac{1}{W_n^{(d)}}\sum_{j=1}^{2n-1}\#_{j,j+1}(1-\cos (4\pi\tilde{\theta})) \nonumber\\
	& \leq \frac{1}{W_n^{(d)}}\sum_{j=1}^{2n-1}\#_{j,j+1} 8 \pi^2 \tilde{\theta}^2
	= \frac{8\pi^2n^{-3}}{W_n^{(d)}(10/3-\pi)} \sum_{j=1}^{2n-1}\#_{j,j+1},
\end{align}
where $\#_{j,j+1}$ is the number of area-changing transitions involving positions $j,j+1$. Consider a word $vAAw$ with a pair on this position. Since $A\in\{1,\ldots,d\}$, there are $d$ such words.
In $d-1$ of them, $AA$ is a peak, while in the remaining one, $AA$ is a valley.
These $d$ words are connected to each other by a full graph of transitions, but only $d-1$ of these change a peak to a valley, the others recolor the $AA$ peak.
Therefore, for each position $j,j+1$, the total number of words with a transition there
is an upper bound on the number of area-changing transitions. This is no more than the total number of words $W_n^{(d)}$.
Summing over all positions $j$, we get a factor of $n$, giving us an upper bound 
\begin{align}
	\bra{\phi^{(d)}_n}(H|_{\varnothing})\ket{\phi^{(d)}_n}
	= E = \bigO(n^{-2}).
\end{align}

The overlap of the new twisted state with the ground state is thus bounded, and its energy is low. This implies that even if it were made only from the ground state and the first excited state $\ket{\psi_1}$ as 
\begin{align}
	\left(\sqrt{\kappa} \bra{\PF^{(d)}_n} + \sqrt{1-\kappa}\bra{\psi_1}
	\right)
	(H|_{\varnothing}) 
	\left(\sqrt{\kappa} \ket{\PF^{(d)}_n} + \sqrt{1-\kappa}\ket{\psi_1}
	\right)
	&\leq E, \nonumber\\
	(1-\kappa) \bra{\psi_1}(H|_{\varnothing}) \ket{\psi_1}
	&\leq E,
\end{align}
the energy of $\ket{\psi_1}$ and thus the gap of $H|_{\varnothing}$ must be upper bounded by $\frac{E}{1-\kappa} = \bigO\left(n^{-2}\right)$, which we wanted to prove.

\end{proof}


\subsection{Making the ground state unique.}
\label{sec:pairs}

Let us recall and adapt what we proposed for the qubit PF model in Section~\ref{sec:PF2breakMAIN}. We claim that we can break the ground state degeneracy of the $d\geq 3$ PF model by adding a little bit of frustration in the form of a translationally invariant pair-counting term \eqref{HXXcost}
\begin{align}
	H_{\textrm{cost}}^{\textrm{PF}} &= - \sum_{i=1}^{N-1} \sum_{t=1}^{d} \ket{tt}\bra{tt}_{i,i+1} 
\end{align}
as a  perturbation -- with a small prefactor $\delta=1/\textrm{poly}(n)$.
For all words (computational basis states), this term counts the number of subsequent identical letter pairs, as illustrated in Figure~\ref{fig:peakcount}c). Moreover, it keeps the structure of the invariant subspaces intact.

Thanks to the former ground states coming from different invariant subspaces, the first-order perturbation theory correction to the energy of each former ground state is $\delta \bra{\psi_0} H_{\textrm{cost}}^{\textrm{PF}} \ket{\psi_0}$. Thanks to the ground states being uniform superpositions of basis states from a subspace, this correction is proportional to the average number of pairs in all words from the subspace. 

\begin{figure}[h]
	\begin{center}
		\includegraphics[width=16.5cm]{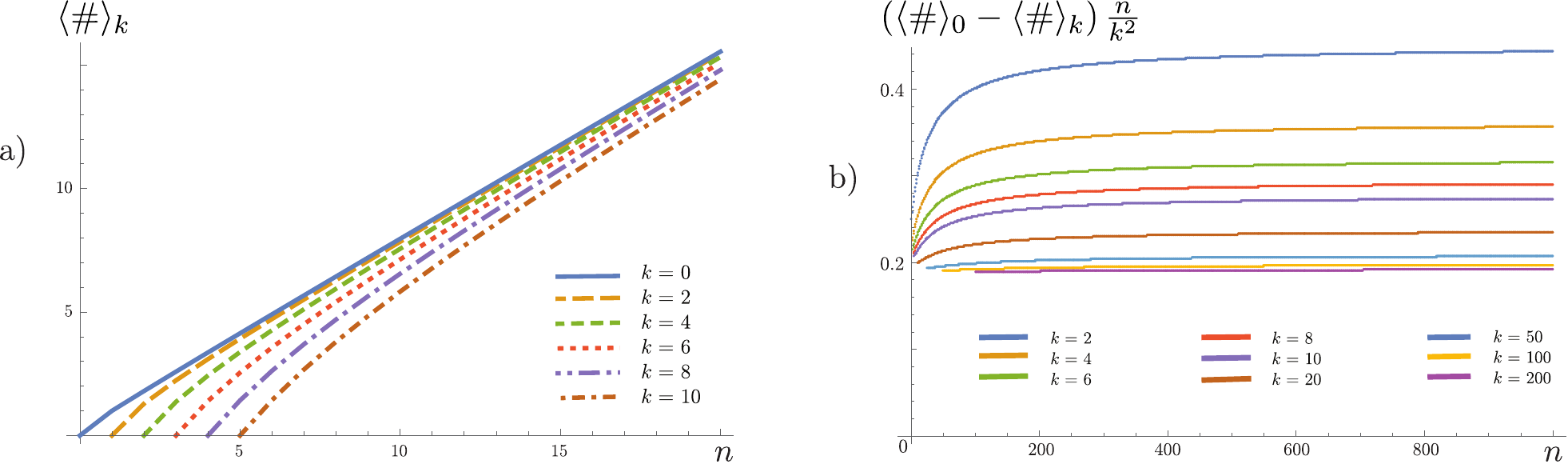}
		\caption{Numerics for the average pair number in $d=3$ PF model words of length $2n$ with $k$ irreducible letters.
		a) The average number of pairs for $n\leq 20$ and low $k$ converges to the $\langle \# \rangle \sim \frac{3}{4}(n+1)$ scaling shown in \eqref{fn32}. 
		b) A closer look at the rescaled gap between the average number of pairs in different subspaces $\left(\langle \# \rangle_{0} - \langle \# \rangle_{k} \right) \frac{n}{k^2}$ for large $n\leq 1000$.
		For clarity, we only plot the values for $k\leq 200$, the higher values of $k$ result in esentially horizontal lines. Observe that the number of pairs in the subspace with $k$ irreducible letters
		is $\Theta\left(k^2/n\right)$ lower than in the fully reducible subspace.}\label{fig:pairnumerics}
	\end{center}
\end{figure}

In Appendix~\ref{sec:GFasymptopairs}, we utilize analytic combinatorics to prove that for constant $k$, the average number of pairs is the largest for the fully-reducible subspace, with a $\Theta\left(n^{-1}\right)$ gap to the closest one, and further scaling as $\Theta\left(3k^2/16n\right)$ for subspaces with $k$ irreducible letters. 
On the other hand, in Appendix~\ref{sec:countPFpairs} we provide a recursive formula for calculating the average number of pairs of identical letters for groundstates from different invariant subspaces. 
Our numerics up to $n = 1000$ (a chain of length $2n=2000$) for all $k$ (shown in Figure~\ref{fig:pairnumerics} up to $k=200$)
indicate that the analytic behavior proven for constant $k$ works for large $k$ as well.

The pair-counting perturbation term thus selects a unique ground state, close to the uniform superposition of fully reducible words.
The original Hamiltonian has an inverse-polynomial gap in each of the invariant subspaces. The gap of the perturbed Hamiltonian will thus remain lower bounded by a small inverse polynomial, depending on the scaling of the inverse polynomial $\delta$.
 
\begin{figure}
\begin{center}
	\includegraphics[width=8cm]{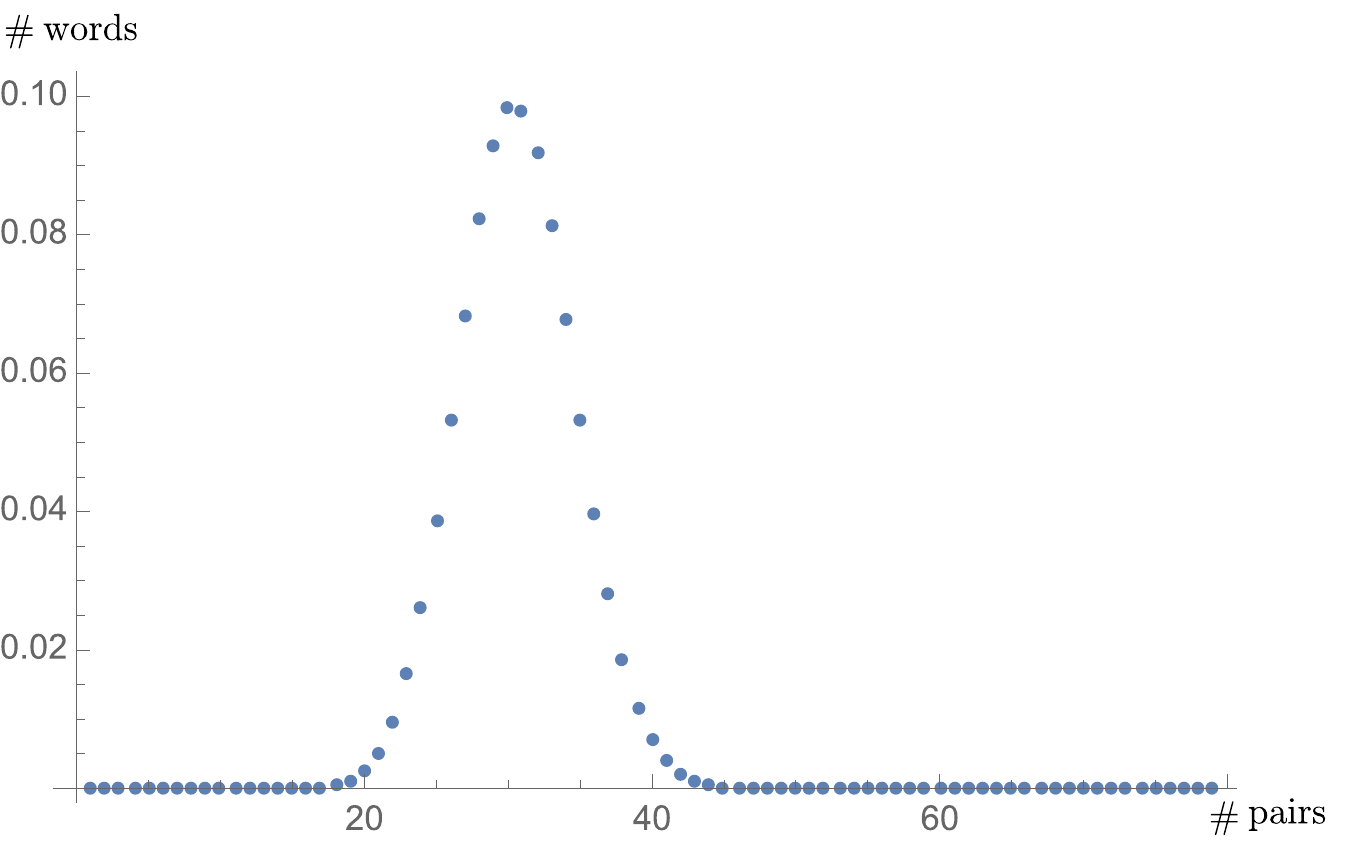}%
	\caption{
	A histogram of the number of pairs in fully reducible, $d=3$ PF model words with length $2n=80$. This distribution is concentrated around the value $3(n+1)/4$ \eqref{fn32}, and has width scaling as $\sqrt{n}$ \eqref{pairvariance}. Note that the maximum possible number of neighboring letter pairs in a word of length $2n$ is $2n-1$ ($AAA\cdots A$).}
	\label{fig:pairdistribution}
\end{center}
\end{figure}
As for the $d=2$ case, for large $n$, the distribution of words with a given number of pairs is narrowly centered around the average number of pairs, as illustrated in Figure~\ref{fig:pairdistribution}, with variance $\bigO(n)$ \eqref{pairvariance}, as shown in Appendix~\ref{sec:GFasymptopairs}. We numerically observe the robustness of the ground state entanglement entropy properties -- for small $\delta$, the ground state remains sufficiently close to the original uniform ground state. Nevertheless, an analytic proof of this remains a direction for future work, together with the possibility of understanding the $d\geq 3$ PF model on a ring.

\section{Discussion and open problems}

Our goal in this paper was to probe the boundaries of what is possible in terms of the gap/entanglement entropy tradeoff in translationally invariant, low-dimensional qudits chains, especially rewriting Hamiltonians with projector terms. We have made a breakthrough with our qutrit PF model. 

We discovered the PF model by realizing that we don't need ``movement'' terms to make a model interesting. We started with another rewriting Hamiltonian: a ``pair-creation'' interaction with particles that appear and disappear in pairs, and also move as the bracket ``particles'' in the Motzkin chain. Thus model has local dimension $d+1$, while its ground state entropy is comparable to the Motzkin chain with $d-1$ colors (i.e. local dimension $2d-1$). This is why the PF model, which receives an effective particle type from parity, is preferable.

There is still a number of open questions about the PF model.
We investigated the degeneracy breaking by preferring pairs analytically (for a range of parameters) and numerically. There is more work left to do for a full analytic solution for the expected number of pairs and the impact of adding such terms as a perturbation on entanglement robustness.
Analogous analytic calculations with a high degree of accuracy are required to analyze the degeneracy breaking impact of peak-counting for the Fredkin chain and particle-counting for the Motzkin chain.

The model we investigated includes frustration. We already know that the ground state entanglement entropy scaling of the $d=2$ PF model can not be achieved in frustration-free systems for qubits. However, we don't yet know what is the best thing in terms of the entanglement entropy/gap tradeoff possible for frustration-free qutrit systems.

Our model is truly translationally invariant, while the Motzkin and Fredkin chain included boundary terms. It would be very interesting to take these models to a system with periodic boundary conditions, with degeneracy broken by the particle-, peak/valley-, or pair-counting terms. However, the calculations would complicate quite a bit, while the precise tuning of the degeneracy-breaking terms would have to be addressed.

It would also be interesting to see if a parametrization similar to the deformed Motzkin/Fredkin spin chain is possible also for the PF model, since the transition rules are rather symmetric.

Moreover, thanks to the symmetry of transition rules, we believe tighter gap lower bounds could be obtained by different techniques.

Let us conclude with speculation about further directions. Would like to focus on possible applications -- investigating the error correcting properties and robustness for this class of models. Second, we are interested in whether the Hamiltonian can be obfuscated while keeping the terms simple, so that someone with the knowledge of the transformation can still prepare the ground state of such a system. In particular, we are interested in formulating a computational problem about rewriting Hamiltonians that would be natural for an intermediate computational class, i.e. verifiable on a less-than-universal quantum computer.

\section*{Acknowledgements}
We thank Ramis Movassagh, Peter Shor, Sergey Bravyi, Dorit Aharonov for interesting discussions.
LC has done the work on this paper during his PhD studies at the Faculty of Mathematics, Physics and Informatics of the  Comenius University in Bratislava.
DN's research has received funding from the People Programme (Marie Curie Actions) EU's 7th Framework Programme under REA grant agreement No. 609427. This research has been further co-funded by the Slovak Academy of Sciences.
LC and DN were also supported by the Slovak Research and Development Agency grant QETWORK APVV-14-0878 and VEGA OAQS 2/0130/15.


\bibliographystyle{plain}	
\bibliography{PFmodelpaper}


\appendix

\section{From well-bracketed words to reducible words: technical tools}
\label{sec:technical}

In this technical Appendix we introduce the basic facts about well-bracketed strings, Catalan numbers, Dyck and Motzkin paths that were necessary for the calculations in the bracket model \cite{CriticalityWithoutFrustration, MovaShor}.
We then build on these and prove several results about the words of the pair-flip model.

Throughout the paper, we label $\D_n^{(s)}$ the set of all Dyck paths with $s$ colors, $\M_n^{(s)}$ the set of all Motzkin paths with $s$ colors, and $\PF_n^{(d)}$ the set of all PF model words with $d$ colors.

\subsection{Catalan numbers}
The famous Catalan numbers OEIS A000108 \cite{OEIS} sequence of integers  1, 1, 2, 5, 14, 42, 132, 429, 1430, 4862, \dots, appears in many combinatorial problems \cite{koshy2008catalan} -- e.g. in the ballot problem, counting non-crossing partitions, or certain paths on grids. 
They play a major role in the understanding of well-bracketed words and the bracket model \cite{CriticalityWithoutFrustration}. Let us look at them in more detail, and sketch several techniques for working with these sequences (recursive relations, generating functions, asymptotic scaling).

\begin{figure}
\begin{center}
\includegraphics[width=10cm]{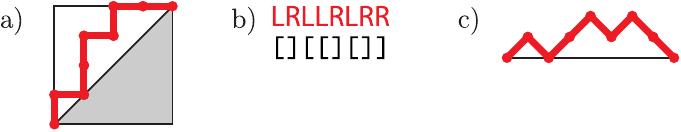}%
\caption{
The Dyck paths can be seen as a) up- and right- moving paths on an $n\times n$ square grid that do not go below the diagonal, b) well-bracketed words, and c) mountain profiles.
}
\label{fig:pathdyck}%
\end{center}
\end{figure}

\subsubsection{Dyck paths}
One of the interpretations of the Catalan number $C_n$ is the number of {\em Dyck paths}: up- and right-moving paths on an $n\times n$ grid that do not go below the diagonal as in Figure~\ref{fig:pathdyck}a). 
It is a combinatorial quantity appearing in many different contexts (voting, walks on trees, polygon dissection, Young tableaux, permutations, mountain ranges, well-bracketed expressions), with a rich literature on the subject. \cite{koshy2008catalan}.
Let us look at the connection between Dyck paths, well-bracketed expressions and mountain ranges.

One can also view Dyck paths as well-bracketed strings, as in Figure~\ref{fig:pathdyck}b) by replacing an ``up'' move by a left bracket ``['' and a ``right'' move by a right bracket ``]''. On the other hand, by turning the image around by 45 degrees as in Figure~\ref{fig:pathdyck}c), we can also view the Dyck paths as mountain ranges. When these mountain ranges are also allowed to have plateaus as in Figure~\ref{fig:pathmotzkin}c), the paths are called {\em Motzkin paths}.

Counting Dyck paths can be done for example by Andr\'{e}'s reflection\footnote{Following \cite{AndreReflect}, it is the number of all paths on an $n\times n$ grid, minus the number of ``bad'' paths that go over the diagonal. To count those, take each ``bad'' path, and reflect it (move right $\leftrightarrow$ move up) from the point after it crosses the diagonal. It will necessarily end at $(n+1,n-1)$. Note that all up- and right-moving paths on an $(n+1)\times(n-1)$ grid can be obtained in this fashion. Therefore,  $C_n = \binom{2n}{n} - \binom{2n}{n-1}$.} method:
\begin{align}
	C_n = \binom{2n}{n} - \binom{2n}{n-1}= \frac{1}{n+1}\binom{2n}{n}. \label{catalan}
\end{align}
We can use this to derive a handy recursive expression
\begin{align}
	C_{n+1} = \frac{1}{n+2}\binom{2(n+1)}{n+1} 
	= \frac{(2n+2)(2n+1)}{(n+2)(n+1)^2} \binom{2n}{n} 
	= \frac{4n+2}{n+2}\, C_n,
	\label{Cnrecursive}
\end{align}
and prove two facts about the ratio of successive Catalan numbers:
\begin{align}
	X_n &= \frac{C_{n}}{C_{n-1}} = \frac{4n-2}{n+1} \leq 4 \\
	X_{n+1} &= \frac{C_{n+1}}{C_{n}} = \frac{4n+2}{n+2} \geq \frac{4n-2}{n+1} \geq X_{n}.
	\label{Xngrow}
\end{align}

We can view the Catalan numbers $\{C_n\}_{n=0}^{\infty}$ as coefficients in the series $\sum_{n=0}^\infty C_n z^n$. We claim that it is equal to the generating function 
\begin{align}
		C(z) = \sum_{n=0}^\infty C_n z^n 
		= \frac{1-\sqrt{1-4z}}{2z}.
		\label{Cgenerate}
\end{align}
We can prove this using the recursive relation
\begin{align}
	C_n = \sum_{i=0}^{n-1} C_{n-i-1} C_{i},
	\label{recursiveCn}
\end{align} 
saying that for $n\geq 1$, any well bracketed word with $n$ pairs can be be built as $(u)v$, from two
shorter well-bracketed words $u$ with $i$ pairs and $v$ of length $n-i-1$ pairs.
Thus, we can write
\begin{align}
	C(z) &= \sum_{n=0}^{\infty} C_n z^n 
		= 1 + \sum_{n=1}^{\infty}  \sum_{i=0}^{n-1} C_i C_{n-1-i} z^n 
	 = 1 + z \sum_{m=0}^{\infty} \sum_{i=0}^{m} C_i z^i C_{m-i} z^{m-i} 
	= 1 + z \left(C(z)\right)^2.
\end{align}
This is a quadratic equation for $C(x)$, whose solution is \eqref{Cgenerate}.

One can also derive (see \cite{FlajoletSedgewick}, p.384, Figure VI.3) a precise asymptotic large-$n$ expansion for the Catalan numbers:
\begin{align}
	C_n = \frac{4^n}{\sqrt{\pi n^3}} \left(
		1 - \frac{9}{8n}
		+ \frac{145}{128n^2}
		- \frac{1155}{1024n^3}
		+ \frac{36939}{32768n^4}
		- \frac{295911}{262144n^5}
		+ O\left(n^{-6}\right)  \right).
		\label{asymcatalan}
\end{align}
This result uses the asymptotic expansion of the term $-(1-4z)^{\frac{1}{2}}$ in the generating function, which has the general form, called {\em standard function scale} \cite{FlajoletOdlyzkoSingularityAnalysis,FlajoletSedgewick}). For any $\alpha \not= 0,-1,-2, \ldots$ and $\zeta\not=0$, we have a full asymptotic expansion in descending powers of $n$
\begin{align}
	[z^n](1-\zeta z)^{-\alpha} &= \frac{\zeta^n n^{\alpha-1}}{\Gamma(\alpha)}\left(1 + \frac{\alpha(\alpha-1)}{2n} + \frac{\alpha(\alpha-1)(\alpha-2)(3\alpha-1)}{24n^2}+\ldots\right). \label{standardscaling}
\end{align}

This {\em singularity analysis} \cite{FlajoletOdlyzkoSingularityAnalysis, FlajoletSedgewick} method works for functions with an algebraic (or algebraic-logarithmic) singularity.  First, find the essential singularity $\zeta$, exponential growth factor, and check whether the function is analytic in a $\Delta$-region (Camembert or Pac-man shaped -- a circular region with an acute angle wedge cutting off the singularity\footnote{Formally, 
		given two numbers $\phi$ and $R$, such that $0<\phi<\pi/2$ and $R>1$, the $\Delta$-region is defined by $\{z:|z|<R,z\not=1, |\arg(z-1)|>\phi\}$. We can shift this region to a singularity $\zeta\not=0$ by $z\mapsto\zeta z$.
		}). Next, expand the function near the singularity, and finally, use the standard function scale to find the sub-exponential scaling as well. We expand the generating function around the singularity ($z=\frac{1}{4}$): $C(z)\sim 2-2\sqrt{1-4z}$ which gives us $[z^n]C(z)\sim\frac{4^n}{\sqrt{\pi n^3}}$. 

To conclude, we have found a closed form expression \eqref{catalan} for the number of Dyck paths of length $2n$ (with $n$ pairs), and understand its asymptotic scaling \eqref{asymcatalan} as well.


\subsubsection{Motzkin paths}

\begin{figure}
\begin{center}
\includegraphics[width=11cm]{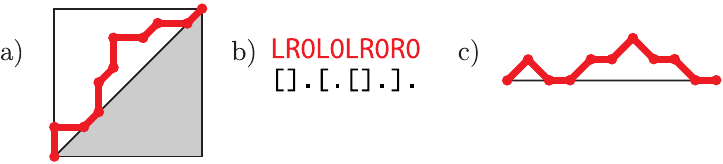}%
\caption{
The Motzkin paths can be seen as a) up-, right-, and diagonally- moving paths on an $n\times n$ square grid that do not go below the diagonal, b) well-bracketed words with spaces, and c) mountain profiles with plateaus.
}
\label{fig:pathmotzkin}%
\end{center}
\end{figure}

When we generalize the Dyck paths, allowing the mountains to have flat sections, we get the {\em Motzkin paths} (see Figure~\ref{fig:pathmotzkin}). This also has an interpretation in the bracket model, where we add another possible letter ``0'' corresponding to a space in the bracket model word.
To get the number of Motzkin paths $M_n$ of length $n$, we just count in how many ways we can take shorter Dyck paths of length $2m$ (with $m$ pairs) and pad them with zeros:
\begin{align}
	M_n = \sum_{m=0}^{\lfloor \frac{n}{2}\rfloor} \binom{n}{2m} C_m.
\label{Mn}
\end{align}

We can also investigate them using a recursive equation that we derive here. Motzkin paths (well-bracketed words with spaces) can be built either from $0$ and another Motzkin path $w$ as $0w$, or from a bracket pair with a Motzkin path $v$ inside and a Motzkin path $u$ outside as $[v]u$:
\begin{align}
	M_n = M_{n-1} + \sum_{i=0}^{n-2} M_{n-i-2} M_{i},
	\label{recursiveMn}
\end{align} 
In terms of generating functions, this recursive equation reads
\begin{align}
	M(z) &= 1 + z M(z) + z^2 M(z) M(z), 
	\label{recursivegeneratingMn}
	\end{align} 
with the 1 taking into account what happens for $n=1$. The solution of the quadratic equation with positive values for $M_n$ is
\begin{align}
	M(z) &= \frac{1-z-\sqrt{1-2z-3z^2}}{2z^2}
	\label{genMn}
\end{align} 
After rewriting $\sqrt{1-2z-3z^2} = \sqrt{(1-3z)(1+z)}$ and expanding around $z=\frac{1}{3}$, one can show using singularity analysis that the asymptotic scaling of $M_n=[z^n]\left(M(z)\right)$ is
\begin{align}
	M_n = 3^n\frac{3\sqrt{3}}{2\sqrt{\pi n^3}} 
	\, \left(1 + O(n^{-1})\right).
\end{align}
Thus, even though we do not have a closed form for the Motzkin numbers $M_n$, we can understand how they scale with growing $n$.


\subsubsection{Dyck and Motzkin paths that don't end at zero height (words with extra brackets)}

\begin{figure}
\begin{center}
\includegraphics[width=11cm]{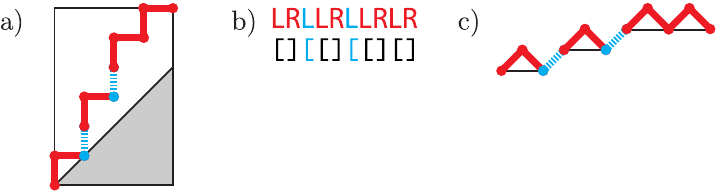}%
\caption{
A generalization of Dyck paths to a) up- and right- moving paths on an $n\times (n+k)$ grid that do not go below the 45$^\circ$ line, b) bracketed words with $k$ extra brackets, and c) mountain profiles that end at level $k$ above the starting point.
}
\label{fig:pathdyckNK}%
\end{center}
\end{figure}

One can generalize Catalan numbers in a simple way. Let us count the number of up- and right-direction only paths that don't go below the diagonal line on an $n\times (n+k)$ grid (with $k$ extra vertical steps) as in Figure~\ref{fig:pathdyckNK}. 
The total path length is $2n+k$.
Andr\'{e}'s reflection method gives us
\begin{align}
	C_{n,k} = \binom{2n+k}{n} - \binom{2n+k}{n-1}= \frac{k+1}{n+k+1}\binom{2n+k}{n}.
	\label{Pk}
\end{align}
We can find the generating function $C^{k}(z) = \sum_{z=0}^{\infty} C_{n,k} z^n$ 
in an iterative way, assuming that we know $C^k(z)$. 
Our base of induction will be $C^0(z) = C(z)$, because walks with no extra steps are counted by $C_{n}^{0} = C_n$.
Let us count the walks above a diagonal on an $n \times n+k$ grid (with $k$ extra steps). 
Each of them is composed from a walk with $k-1$ extra (vertical) steps, one extra (vertical) step, and then a simple Dyck path, implying a recursive relation similar to \eqref{recursiveCn}:
\begin{align}
	C_{n,k} = \sum_{i=0}^{n} C_{n-i,k-1} C_i. \label{Pnkrecursion}
\end{align}
We can use it to obtain the generating function
\begin{align}
	C^{k}(z) &= \sum_{n=0}^{\infty} C_{n,k} z^{n} 
	= \sum_{n=0}^{\infty} \sum_{i=0}^{n} C_{n-i,k-1} C_i z^{n-i+i} 
	= C^{k-1}(z) C(z).
\end{align}
Therefore, by induction starting with $C^{0}(z)=C(z)$, we get
\begin{align}
	C^{k}(z) = (C(z))^{k+1} = \left(\frac{1-\sqrt{1-4z}}{2z}\right)^{k+1}. 
	\label{Ck}
\end{align}

In the Motzkin spin chain, we also need to calculate the number of Motzkin paths that have $k$ extra brackets. Similarly to \eqref{Mn}, this can be done by counting the ways to pad Dyck paths with extra steps:
\begin{align}
	M_{n}^{k} = \sum_{m=0}^{\lfloor \frac{n-k}{2}\rfloor} \binom{n}{2m+k} C_{m,k}.
\label{Mnk}
\end{align}

There are several possible approaches (e.g. the saddle point method) to obtain the asymptotic scaling of the coefficient of the $z^n$ term of \eqref{Ck}, or of the generating function for \eqref{Mnk} which we don't show here.
We advise the reader to be very careful about the various error guarantees and requirements on the scaling of $k$ with $n$. 

\subsubsection{Coloring the brackets}
The bracket model discussed above can be easily generalized to several types (colors, species) of brackets. For example, we could allow two types of brackets,
allowing words such as $\red{[}\,\red{]}\,()$ 
in addition to $()\,()$. 
For the Dyck paths (in the Fredkin spin chain), this requires local dimension $d=2s$, while for the Motzkin paths (in the Motzkin chain, including the empty spaces) require local dimension $d=2s+1$.

Let us count well-bracketed words made from brackets of $s$ different types.
These calculations are required for the analysis of the colored bracket model in \cite{MovaShor}.
Each such word can be mapped to a generalization of a Dyck path -- a mountain range, where each pair of up- and down- moves have the same color (type). All that changes in the coefficients is counting the possibilities of assigning the colors
for the bracket pairs, i.e.
\begin{align}
	C^{(s)}_n &= s^n C_n, \qquad 	C^{(s)}_{n,k} = s^n C_{n,k},
	\label{countcoloredCn}
\end{align}
where $C^{(s)}_{n,k}$ counts the number of words with $s$ bracket colors that have a specifically colored sequence of $k$ extra brackets.
The expressions \eqref{Xngrow} can thus be simply generalized for the colored Dyck paths as
\begin{align}
	X_n^{(s)} &= \frac{C_{n}^{(s)}}{C_{n-1}^{(s)}} 
	= \frac{C_{n}^{(s)}}{C_{n-1}^{(s)}} = s \frac{4n-2}{n+1} \leq 4s, \\
	X_{n+1}^{(s)} &= \frac{C_{n}^{(s)}}{C_{n-1}^{(s)}} = s\frac{C_{n+1}}{C_{n}} \geq s\frac{C_{n}}{C_{n-1}} \geq X^{(s)}_{n}.
	\label{Xdngrow}
\end{align}

We can also use \eqref{countcoloredCn} to find the generating functions for the $d$-colored well-bracketed words and $k$-extra bracket words, which become
\begin{align}
	C^{(s)}(z) &= C(sz), \qquad 	C^{(s),k}(z) 
	= C^k(sz) = \left(C(sz)\right)^{k+1}.\label{ColorCGF}
\end{align}
However, things get slightly more complicated when we look at colored Motzkin paths (with flat sections/empty spaces/zeros). For $s$-colored Motzkin paths, the recursion \eqref{recursiveMn} has an extra $s$ in front of the $M^2(z)$ term, so the generating function becomes
\begin{align}
	M^{(s)}(z) = \frac{1-z-\sqrt{1-2z+(1-4s)z^2}}{2sz^2}.
		\label{generatingMn}
\end{align}
It is a generating function with algebraic singularity at 
$z = \frac{-1+2\sqrt{s}}{4s-1} = \frac{1}{1+2\sqrt{s}}$.
It is easy to check that this function is analytic in $\Delta$-region. Therefore we can use the singularity analysis to obtain the asymptotic scaling

\begin{align}
	M^{(s)}_n\sim \left(1+2\sqrt{s}\right)^n \frac{\sqrt{(1+2\sqrt{s})^3}}{2 s^{3/4}\sqrt{\pi n^3}}\left(1+\bigO\left(\frac{1}{n}\right)\right).
\end{align}

Finding the scaling for colored Motzkin paths with $k$-extra brackets is a bit more complicated, especially when we would like it to work up to $k = \Theta(\sqrt{n})$, necessary in the calculations in the Motzkin spin chain.	

This concludes the technical Section on the Catalan numbers and counting in the bracket model. Next, we will look at the pair-flip model and its words.

\subsection{The pair-flip (PF) model and its words: returning walks}
\label{sec:PFtools}

In this Section we provide a detailed discussion of the pair-flip model
with the alphabet $\{1,\dots,d\}$, and calculations of its combinatorial properties.
 
The rules of the PF model \eqref{HXXflip} allow us to change the color of an adjacent letter pair, e.g. $\dots 11 \dots \leftrightarrow \dots 22 \dots$.
Such rewriting moves preserve the {\em irreducible string} of a word.
We specify the reduction procedure in Definition~\ref{def:reduction}, and illustrate it in Figure~\ref{fig:reduction}. Two other examples of irreducible strings are illustrated in Figure~\ref{fig:wordspathstrees}a), where the letters are fully matched, and the word is fully reducible, and Figure~\ref{fig:wordspathstrees}d), where the word reduces to the irreducible string ``32''.

The reduction process pairs letters uniquely. One could imagine another way of pairing the letters in words like $1111$ -- matching the middle pair and the outer pair. However, we choose to always match the letters by reading from left to right, as dictated by Definition~\ref{def:reduction}.

\begin{figure}
\begin{center}
\includegraphics[width=16.5cm]{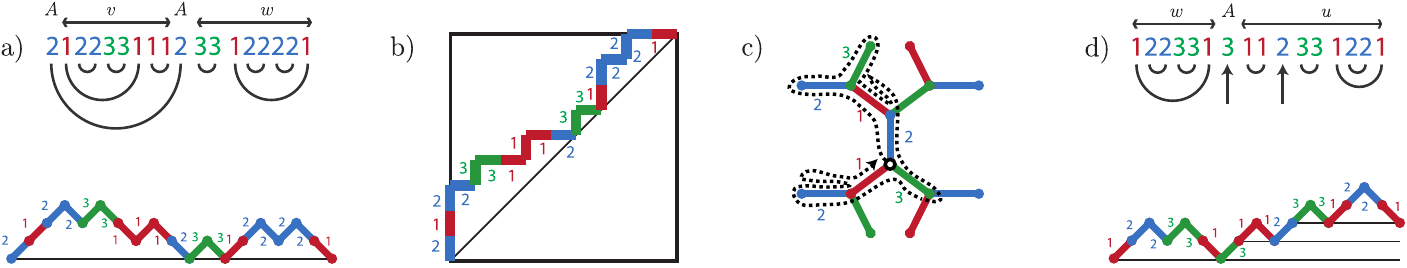}%
\caption{
Different ways to view words in the pair-flip model. a) A fully reducible word from the alphabet $\{1,2,3\}$. The semicircles depict a unique way of pairing (reducing) the letters when reading from left to right. 
We also depict a division of the word into two blocks as $AvAw$. 
b) The same fully reducible word viewed as a mountain range/Dyck path on a square grid: move up when a new letter is introduced, and move right when encountering its partner.
c) The same word viewed as returning walk on a 3-regular colored tree.
 d) A word that reduces to the nonempty string $32$ is made from a fully reducible word $w$, the letter $3$, and the word $u$ divided by the letter $2$ into two $d-1$ colored Dyck paths.
Its mountain profile ends two steps above the floor.}
\label{fig:wordspathstrees}%
\end{center}
\end{figure}

The fully reducible words in the PF model can be viewed as words whose letters can be paired and reduced, colored paths on a grid, colored ``mountain'' profiles, or colored trees (see Figure~\ref{fig:wordspathstrees}) or returning walks on a $d$-regular graph (see Figure~\ref{fig:stree}). In those, each vertex has $d$ edges, each of which has a different color/label $\{1,\dots,d\}$. Fully reducible words correspond to walks that return to the starting point. This is because at the starting point, we have $d$ ways of growing the stack by going ``outwards''. At all other points we have one way to pop the stack and ``go back inside for one step'', or $d-1$ options of adding a letter to the stack. This also means when we end up at distance $k$ from the center, our stack has $k$ letters -- and the walk corresponds to a word that reduces to a $k$-letter irreducible string. We will use this observation to derive a recursive property of the words, as well as a generating function allowing us to count them.

\subsubsection{The relationship of PF words to colored Dyck paths}
\label{sec:correspond}

Viewing PF words as colored ``mountains''/Dyck paths as in Figure~\ref{fig:wordspathstrees}b), 
we realize there are restrictions on their possible coloring\footnote{This results in complicated counting, as a mountain range with $n$ pairs that touches the ground $t$ times (i.e.  is made of $t-1$ fully reducible blocks) can be colored in $d^{t-1} (d-1)^{n-(t-1)}$ ways, while the standard colored bracket model's mountains with $n$ pairs are colorable by $d$ colors straightforwardly in $d^n$ ways. For example, for $d=3$, the bracket model word $LLRLRR$ or $(\,(\,)\,(\,)\,)$ can be colored in $3^3=27$ ways. Although it has the same mountain range profile, the PF-model word $ABBCCA$ can only be colored (have the letters $A,B,C$ chosen) in only $3\times 2 \times 2=12$ ways.}. The base layer pairs (up- and right- moves in Figure~\ref{fig:wordspathstrees}b) can take any of the $d$ colors, as there are $d$ possible directions of the walk from the initial point in Figures
~\ref{fig:wordspathstrees}c) or \ref{fig:stree}. However, each pair that is not touching the diagonal has only $d-1$ color options -- same as the number of ways to walk outwards from a non-central point in Figure~\ref{fig:stree}.  
This strongly suggests that the PF words are similar to $d-1$ colored Dyck paths, except for the points that touch the diagonal/return to the center. We will make this comparison precise in \eqref{Vexact}.

However, there are certain parts of the PF words that map to $d-1$ colored Dyck paths exactly -- the fully reducible subblocks that do not touch the diagonal. Consider a general, fully reducible PF word. We can write it as $A\alpha A\beta$, with some letter $A$, a PF model word $\beta$, and a block $\alpha$. We claim that this block $\alpha$, sitting at a higher nesting level, exactly corresponds to a $d-1$ colored Dyck word. 
We capture it by the following correspondence functions, useful in the proof of the energy gap Section \ref{sec:gap}. 
\begin{enumerate}
	\item The function $\PF(w, X)$ maps a $d-1$ colored Dyck path $w$ to a $d$ colored fully reducible PF word $q$, a substring of a single-block PF model word $XqX$ with base color $X$. 
	\item Inversely, the function $\D(q,X)$ maps the $d$-color, fully reducible substring $q$ of a 1-block PF word $XqX$ (with a base color $X$) to a $d-1$ colored Dyck word. 
\end{enumerate} 
We define these functions recursively. 
\begin{enumerate}
	\item We can split the Dyck path $w$ into balanced blocks enclosed in bracket pairs $w=L_{i_1}\alpha_1 R_{i_1}\ldots L_{i_k}\alpha_k R_{i_k}$, where $i_x$ is the color of the bracket enclosing the block $\alpha_x$.
	Then 
	\begin{align}
	\PF(w,X)\quad \mapsto \quad 
	A_{\tilde{i_1}}\PF(\alpha_1,A_{\tilde{i_1}})A_{\tilde{i_1}}\,
	\ldots \,
	A_{\tilde{i_k}}\PF(\alpha_k,A_{\tilde{i_k}})A_{\tilde{i_k}},
	\label{PFmap}
	\end{align}
	where $A_{\tilde{i_x}}$ is the $i_x$th element in the naturally ordered set
	$\{1,\ldots,d\}\setminus X$
	of PF colors without the base color $X$.
	Of course, when $w$ is empty, we simply map it to an empty PF word.
	\item We define the inverse function $\D$ similarly. Each fully reducible substring $q$ of a single-block PF word $XqX$ can be written as 
	$q=A_{j_1}\beta_1 A_{j_1}\, \ldots\, A_{j_l}\beta_l A_{j_l}$, with the color $j_x$ from the naturally ordered set
	$\{1,\ldots,d\}\setminus X$ of PF colors without the base color $X$. 
	We define recursively
	\begin{align}
	\D(q,X)\quad \mapsto \quad L_{\tilde{j_1}}\D(\beta_1,A_{j_1}) R_{\tilde{j_1}}
	\, \ldots \, L_{\tilde{j_l}}\D(\beta_l,A_{j_l}) R_{\tilde{j_l}},
	\label{Dcorrespond}
	\end{align}
	where $\tilde{j}_x$ denotes $j_x$th element in the naturally ordered set $\{1,\ldots,d\}\setminus X$.
\end{enumerate}

\begin{figure}
\begin{center}
\includegraphics[width=8cm]{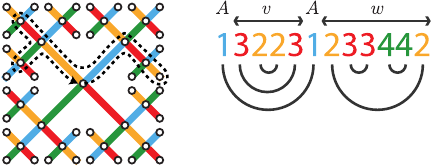}%
\caption{
The words in the PF-model can be viewed as walks on a $d$-regular tree, with $d$ differently colored edges for each vertex. When walking, we write down the colors of the edges and thus build the word. Fully reducible words correspond to walks that end at the initial point (the one depicted here reads $132231233442$). On the other hand, walks that end at distance $k$ from the initial point correspond to words that reduce to a $k$-letter string. The division $AvA w$ into two blocks is useful for getting the recursive relationship \eqref{recursionPCcolored} for counting the words.
}
\label{fig:stree}%
\end{center}
\end{figure}

With this in hand, we can turn to practical calculations.

\subsubsection{The simple case $d=2$.}
\label{sec:PF2start}

This is a qubit model ($d=2$), with the alphabet $\{1,2\}$.
For words of length $2n$, the irreducible words
are even-length strings with alternating $1$'s and $2$'s:
\begin{align}
	\varnothing, 12, 21, 1212, 2121, 121212, 212121, \dots,
\end{align}
and there are $1+2n$ of these.  

It is straightforward to count the numbers of such words.
Figure~\ref{fig:stree} becomes simply a line with alternating link colors. 
The number of fully reducible words of length $2n$ is thus
\begin{align}
	W^{(2)}_{n} = W_n = \binom{2n}{n},
	\label{reducible2}
\end{align}
the number of walks on a line that returns to the initial point after $2n$ steps.

Similarly, when counting 2-color words in the PF-model that reduce to a length-$k$ string of the type $1212\dots$ or $2121\dots$,
we are just counting the number of walks on a line that start at $0$ and end at position $k$ or $-k$. We obtain
\begin{align}
	W^{(2)}_{n,k} = W_{n,k} = \binom{2n+k}{n}.
		\label{reducible20k}
\end{align}
This was not difficult, but we will see below that adding just one more letter introduces way more complexity.

\subsubsection{The general case $d\geq 3$.}
\label{sec:PF3start}

The fully reducible words in the $d$-color PF-model can be built from two blocks as $AvAw$ as in Figure~\ref{fig:wordspathstrees}a) or Figure~\ref{fig:stree},
where $A$ is a single letter from the alphabet $\{1,\dots,d\}$, and the two $A$'s are a pair matched by the reduction procedure. 
The word $w$ is some fully reducible PF-model word, while $v$ has more restrictions -- we can view it as a walk that starts one vertex away from the center in Figure~\ref{fig:stree}, and can't return to the center. This means $v$ corresponds to a $d-1$ colored Dyck path, and we already know how to count those \eqref{countcoloredCn}. Thus, 
\begin{align}	
	W^{(d)}_{n} = \sum_{i=0}^{n-1} d C^{(d-1)}_{n-1-i}  W^{(d)}_{i}
	= \sum_{i=0}^{n-1} d (d-1)^{n-1-i} C_{n-1-i}  W^{(d)}_{i}, \label{recursionPCcolored}
\end{align}
with $d$ choices for the letter $A$, and counting $(d-1)^{n-1-i} C_{n-1-i}$ possible words $v$ and $W^{(d)}_{i}$ possible words $w$.
Note that this works for $n\geq 1$, and $W^{(d)}_{0}=1$ (the empty word). 
For $d=3$, this results in the sequence 1, 3, 15, 87, 543, 3543, 23823, \dots, also known as OEIS\cite{OEIS} A089022.

We can use this recursive relation to find the generating function for the series $\{W^{(d)}_{n}\}_{n=0}^\infty$ as follows:
\begin{align}
	W^{(d)}(z) &= 
	 1+ \sum_{n=1}^{\infty} W^{(d)}_{n} z^n
	= 1+ \sum_{n=1}^{\infty} 
	\sum_{i=0}^{n-1}
	d \underbrace{(d-1)^{n-1-i} C_{n-1-i}}_{C^{(d-1)}_{n-1-i}} 
				W^{(d)}_{i} z^n \\
	&= 1+ dz \sum_{n=1}^{\infty} \sum_{i=0}^{n-1} C^{(d-1)}_{n-1-i} W^{(d)}_{i} 
		z^{n-1}  
	= 1+ d z \sum_{m=0}^{\infty} \sum_{i=0}^{m} C^{(d-1)}_{m-i} W^{(d)}_{i} z^i z^{m-i} \\
	&= 1+ dz C^{(d-1)}(z) W^{(d)}(z),
\end{align}
Solving the above equation for $W^{(d)}(z)$ and using \eqref{ColorCGF} gives us
\begin{align}
	W^{(d)}(z) &
	= \frac{1}{1-dz {C}^{(d-1)}(z)} 
	= \frac{1}{1-dz \left(\frac{1-\sqrt{1-4(d-1)z}}{2(d-1)z}\right)} 
	= \frac{2(d-1)}{d-2+d\sqrt{1-4(d-1)z}}. \label{W0generate}
\end{align}
This formula has been originally proven in the context of walks on $d$-regular trees by \cite{McKay1983aa} and has also a short proof in \cite{2009arXiv0903.1877R}, with references to \cite{CartierHarmoniTrees} and \cite{10.2307/1993160}.

Let us now use the generating function to prove a formula for $W^{(d)}_{n}$, the number of fully reducible words made from $n$ pairs of letters $1,\dots,d$. We will use another form of \eqref{W0generate}: $W^{(d)}(z) = \left(\frac{d}{2}\right) \frac{\sqrt{1-4(d-1)z}}{1-d^2 z} -  \left(\frac{d-2}{2}\right)\frac{1}{1-d^2z}$, and plug in the expansions
\begin{align}
	\left(1-d^2 z\right)^{-1} &= \sum_{n=0}^{\infty} d^{2n} z^n,\\
	\sqrt{1-4(d-1)z} &= 1-2(d-1)z \sum_{n=0}^{\infty} C_n (d-1)^n z^n 
		 = 1-2(d-1)\sum_{n=1}^{\infty} C_{n-1} (d-1)^{n-1} z^{n}, \\
	\frac{\sqrt{1-4(d-1)z}}{1-d^2z} &= 
			1 - \sum_{n=1}^{\infty} \left(
				d^{2n} - 2(d-1) \sum_{i=1}^{n} d^{2(n-i)} C_{i-1} (d-1)^{i-1} \right)z^n \nonumber\\
			& = 1 - \sum_{n=1}^{\infty} d^{2n} \left(
				 1 - \frac{2(d-1)}{d^2} 
					\sum_{i=0}^{n-1} \left(\frac{d-1}{d^2}\right)^{i} C_{i}
					\right) z^n. \label{expandratio}
\end{align}
Let us extract the coefficient in front of $z^n$ in the formal power series $W^{(d)}(z)$:
\begin{align}
	W^{(d)}_{n} &= [z^n]\left(W^{(d)}(z)\right)
	= [z^n]\left(\left(\frac{d}{2}\right) \frac{\sqrt{1-4(d-1)z}}{1-d^2 z} 
			-  \left(\frac{d-2}{2}\right)\frac{1}{1-d^2z}\right)\\
	&= \frac{d}{2} d^{2n}
			\left(
				1- \frac{2(d-1)}{d^2} \sum_{i=0}^{n-1} \left(\frac{d-1}{d^2}\right)^{i} C_{i} 
		\right) - \frac{(d-2)}{2}d^{2n}
	\\
	&= d^{2n}\left( 1- \frac{d-1}{d} \sum_{i=0}^{\infty} \left(\frac{d-1}{d^2}\right)^{i} C_{i}
	+ \frac{d-1}{d} \sum_{i=n}^{\infty} \left(\frac{d-1}{d^2}\right)^{i} C_{i} \right) \label{Vexpress}
\end{align}
Recalling the generating function for Catalan numbers \eqref{Cgenerate}, we can calculate the first of the sums as
\begin{align}
	\sum_{i=0}^{\infty} \left(\frac{d-1}{d^2}\right)^{i} C_{i} 
	&= C\left(\frac{d-1}{d^2}\right) 
	= \frac{d}{d-1}.
\end{align}
To get a handle on the second part, we express
\begin{align}
	\sum_{i=n}^{\infty}  C_i \left(\frac{d-1}{d^2}\right)^i 
	&=
	\left(\frac{d-1}{d^2}\right)^{n} C_n \underbrace{\sum_{j=0}^{\infty} \left(\frac{d-1}{d^2}\right)^j \frac{C_{n+j}}{C_n}}_{R^{(d)}_n},
\end{align}
where the term 
\begin{align}
	R^{(d)}_n = \sum_{j=0}^{\infty} \left(\frac{d-1}{d^2}\right)^j \frac{C_{n+j}}{C_n}
	\label{Rn}
\end{align}
will be important in several expressions later. For example, it lets us 
turn \eqref{Vexpress} into
\begin{align}
	W^{(d)}_{n} &=  \frac{(d-1)R^{(d)}_n}{d} \, (d-1)^n C_n \label{Vexact}.
\end{align}
We can interpret it as the number of fully reducible $d$-color PF-model words being a $R_n (d-1)/d$ multiple of the number of well-bracketed words for a $d-1$ colored bracket model.

Next, we can show that $R_n$ tends toward a constant, with an $O(n^{-1})$ error term.
Using  the approximation $C_{n+j}/C_n \approx 4^j$, true to order $O(n^{-1})$, we obtain for large $n$
\begin{align}
	\lim_{n\rightarrow \infty} R_n &= \sum_{j=0}^{\infty} \left(\frac{d-1}{d^2}\right)^j \frac{C_{n+j}}{C_n} 
		\approx \sum_{j=0}^{\infty} \left(\frac{4(d-1)}{d^2}\right)^j = \frac{1}{1-\frac{4(d-1)}{d^2}} 
		= \left(\frac{d}{d-2}\right)^2.
	\label{Rnfirst} \\
	W^{(d)}_{n} &= \left(\frac{d(d-1)}{(d-2)^2} + O\left(n^{-1}\right) \right)\, (d-1)^n C_n
	 = \left(\frac{d(d-1)}{(d-2)^2} + O\left(n^{-1}\right) \right)\, C^{(d-1)}_n. \label{Vapprox1}
\end{align}
Note that we now have $d\geq 3$, while we have already found $W^{(2)}_{n}$for $d=2$ in \eqref{reducible2}.
Also, because $C_{j+1}\leq 4 C_j$, the approximations \eqref{Rnfirst} and \eqref{Vapprox1} also give us upper bounds:
\begin{align}
	R^{(d)}_n \leq \left(\frac{d}{d-2}\right)^2,
	\qquad
	W^{(d)}_{n} \leq \frac{d(d-1)}{(d-2)^2}\, (d-1)^n C_n. \label{Vupper}
\end{align}

Note that the factor $R_n$ is growing with $n$, 
\begin{align}
	R^{(d)}_{n+1} = \sum_{j=0}^{\infty} \left(\frac{d-1}{d^2}\right)^j \frac{C_{n+1+j}}{C_{n+1}}
	\geq \sum_{j=0}^{\infty} \left(\frac{d-1}{d^2}\right)^j \frac{C_{n+j}}{C_{n}} = R^{(d)}_n, \label{Rngrow}
\end{align}
because 
\begin{align}
	\frac{C_{n+1+j}}{C_{n+1}}\geq \frac{C_{n+j}}{C_{n}} \label{Cnratiogrow}
\end{align} can be shown using a little algebra with the
exact expression for the Catalan numbers \eqref{catalan}.
Using $j=1$ this also implies that the ratio $X_n^{(d-1)} = \frac{(d-1)^n C_n}{(d-1)^{n-1}C_{n-1}}$
of the number of (colored) Dyck paths with $n$ and $n-1$ pairs is non-decreasing, i.e. 
\begin{align}
	X_{n+1}^{(d-1)}\geq X_{n}^{(d-1)}.	\label{XngrowD}
\end{align}

Next, we will show a few results about the ratio of successive numbers of fully reducible PF words.
For this, we will first derive a recursive expression for $R_n$ from \eqref{Rn}, keeping in mind the first term $R^{(d)}_0=\frac{d}{d-1}$:
\begin{align}
	R^{(d)}_{n+1} 
	&= \frac{C_n}{C_{n+1}} \left(\frac{d^2}{d-1}\right) \sum_{j=0}^{\infty} \left(\frac{d-1}{d^2}\right)^{j+1} 
			\frac{C_{n+1+j}}{C_n} 
	= \frac{C_n}{C_{n+1}} \left(\frac{d^2}{d-1}\right) \left(R^{(d)}_n - 1\right).
	\label{Rnrecursive}
\end{align}
Using this and \eqref{Cnrecursive} in \eqref{Vexact} gives us a recursive relationship for $W^{(d)}_{n}$:
\begin{align}
	W^{(d)}_{n+1} =  \frac{(d-1)R^{(d)}_{n+1}}{d}C_{n+1}^{(d-1)}
	& = d(d-1)(R^{(d)}_n-1)C_n^{(d-1)} 
	= d^2 W^{(d)}_{n} - d(d-1)C_{n}^{(d-1)}. \label{Vnrecursive}
\end{align}

We will now prove an upper and a lower bound on $Y_n^{(d)}$:
\begin{align}
	4(d-1) \geq Y_n^{(d)} = \frac{W_{n}^{(d)}}{W^{(d)}_{n-1}}
	= \frac{(d-1) R^{(d)}_{n+1} C_{n+1}}{R^{(d)}_n C_n} \geq X_n^{(d)}. \label{Ybounds}
\end{align}
on the ratio of successive numbers of fully reducible words.
The lower bound simply relies on \eqref{Rngrow} and the definition of $R^{(d)}_n$. For the upper bound, we want to show
$W_{n+1}^{(d)} \leq 4(d-1) W_{n}^{(d)}$.
Using the recursive relation \eqref{Vnrecursive}, this translates to
\begin{align}
	 d^2 W^{(d)}_{n} - d(d-1)C_{n}^{(d-1)} &\leq 4(d-1) W_{n}^{(d)}, \\
	 \frac{(d-2)^2(d-1)}{d} R^{(d)}_n C^{(d-1)}_n &\leq d(d-1)C^{(d-1)}_{n}. 
\end{align}
Because of \eqref{Vupper}, this is true, so we get the upper bound in \eqref{Ybounds}.

Finally, we will show that the ratio $Y_n^{(d)}$ is non-decreasing.
We want to show
\begin{align}
	Y_{n+1}^{(d)} &\geq Y_{n}^{(d)}, \label{Yngrow}
\end{align}
which means we must prove
\begin{align}
	\frac{(d-1) R^{(d)}_{n+1} C_{n+1}}{R^{(d)}_n C_n} &\geq \frac{(d-1) R^{(d)}_{n} C_{n}}{R^{(d)}_{n-1} C_{n-1}}, \\
	\frac{R^{(d)}_{n}-1}{R^{(d)}_{n}} &\geq \frac{R^{(d)}_{n-1}-1}{R^{(d)}_{n-1}}, \label{toproveY}
\end{align}
using \eqref{Rnrecursive} on both sides. The expression \eqref{toproveY} is true because of \eqref{Rngrow}.
Therefore, \eqref{Yngrow} is true.

We also need the asymptotic scaling of $W_{n}^{(d)}$. This can be again obtained from the generating function by singularity analysis \cite{FlajoletOdlyzkoSingularityAnalysis, FlajoletSedgewick}. The dominant singularity is at $1/(4(d-1))$.
Expanding the generating function around this singularity and considering only half integer powers of $(1-4(d-1)z)$ we receive
\begin{align}
	-\frac{2(d-1)d}{(d-2)^2}\,(1-4(d-1)z)^{1/2}-\frac{2(d-1)d^3}{(d-2)^4}\,(1-4(d-1)z)^{3/2}
	+O\left((1-4(d-1)z)^{5/2}\right),
\end{align}
which gives us the asymptotic approximation:
\begin{align}
	[z^n]W^{(d)}(z)
	\sim \frac{d(d-1)}{(d-2)^2} \frac{(4(d-1))^n}{\sqrt{\pi n^3}}\left(1
	-\frac{3(d+2)(3d-2)}{8(d-2)^2n}
	+O\left(n^{-2}\right)\right).
\end{align}
Observe that up to a constant prefactor, the leading term scales with $n$ exactly as the leading term for the $d-1$ colored Dyck paths $C_n^{(d-1)}=(d-1)^n C_n$,
recalling the asymptotic scaling of Catalan numbers $C_n$ from \eqref{asymcatalan}. 

Let us take one more step in this comparison. We can calculate the bounds on $R_n$ to higher precision, using higher order terms in the asymptotic expansion 
\eqref{asymcatalan}. For example, to get the order $\frac{1}{n}$, we calculate
\begin{align}
	R_n = \sum_{j=0}^{\infty} \left(\frac{d-1}{d^2}\right)^j \frac{C_{n+j}}{C_n} 
	&\approx 
		\sum_{j=0}^{\infty} \left(\frac{4(d-1)}{d^2}\right)^j 
		\left(\frac{n}{n+j}\right)^{\frac{3}{2}}
		\underbrace{\left(\frac{1-\frac{9}{8n}}{1-\frac{9}{8(n+j)}}\right)}_{1+O\left(j/n^2\right)} \nonumber\\
	&\sim
		\sum_{j=0}^{\infty} \left(\frac{4(d-1)}{d^2}\right)^j 
		\left(1-\frac{3j}{2n}\right)
	=
		\sum_{j=0}^{\infty} y^j - \frac{3}{2n}\sum_{j=0}^{\infty} j y^j \nonumber\\
	&= \frac{1}{1-y} - \frac{3}{2n}\frac{y}{(1-y)^2}  
	= 
	 \frac{1}{1-\frac{4(d-1)}{d^2}}
	 - \frac{3}{2n} \frac{4(d-1)}{d^2\left(1-\frac{4(d-1)}{d^2}\right)^2} \nonumber\\
	 &= \frac{d^2}{(d-2)^2}
	 -\frac{6d^2(d-1)}{(d-2)^4}\frac{1}{n},
		\label{RnHigherTerms}
\end{align}
where we labeled $y=\frac{4(d-1)}{d^2}$ and used $\sum_{j=0}^{\infty} y^j \sim (1-y)^{-1}$ 
and $\sum_{j=0}^{\infty} j y^j \sim y (1-y)^{-2}$.
This translates to an approximation
\begin{align}
	W^{(d)}_{n} &= \left(\frac{d(d-1)}{(d-2)^2}
	- \frac{6d (d-1)^2}{(d-2)^4}\frac{1}{n}
	+ O\left(n^{-2}\right)
	\right) \, C^{(d-1)}_n. \label{Vapprox2}
\end{align}
The number of fully reducible PF model words is thus asymptotically proportional to the number of $d-1$ colored Dyck paths.

\subsection{Words that reduce to something nonempty in the PF-model}

\begin{figure}
\begin{center}
\includegraphics[width=9cm]{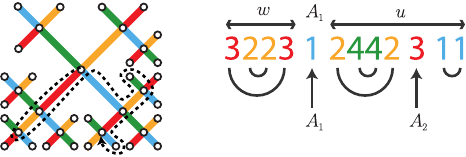}%
\caption{
The word $322312442311$ from the $d=4$ PF model reduces to the string $13$.
The words in the PF-model that reduce to a nonempty string $X_1\cdots X_k$ can be viewed as walks on a colored $d$-regular tree, ending at distance $k$ from the start (compare to to Figure~\ref{fig:stree}). We can also build them from blocks as $wX_1 u$, where $w$ is a fully reducible word, $X_1$ is the first extra letter, and $u$ is a word reducible to $k$ letters. Furthermore, $u$ is built from $k$ blocks of $d-1$ colored Dyck paths $v_1, \dots, v_k$
separated by the extra letters $X_2,\dots,X_k$ 
as $u=v_1 X_2 v_2 X_3 v_3 \dots X_k v_k$. This is useful for getting the recursive relationship \eqref{Wk}. 
}
\label{fig:Wk}%
\end{center}
\end{figure}

Next, let us look at words $w$ that reduce to a particular $k$-letter string $A_1A_2\cdots A_k$ and count them. These words can be built from a fully reducible word $w$, the letter $A_1$, and a word $u=b_1 A_2 b_2 A_3 \dots A_k b_k$, divided by the letters $A_2,\dots,A_k$ into $k$ blocks $b_j$. Each of the blocks $b_j$ is an $d-1$ colored Dyck path, as its base layer can not have the letter $A_j$. Observe that because $A_{j+1}\neq A_{j}$ (the $A$'s form an irreducible string), the number of the possible words $u$ is the same as the number of $d-1$ colored bracket words with $k-1$ extra brackets, defined in \eqref{Pk} and \eqref{countcoloredCn}.
Therefore, the recursive relation
\begin{align}
	W_{n,k}^{(d)} &= \sum_{i=0}^{n} W_{i}^{(d)} C_{n-i,k-1}^{(d-1)}
	\label{Wk}
\end{align}
holds for $n\geq 0$.
Recalling \eqref{Vexact}, this means
\begin{align}
	W_{n,k}^{(d)} &= \sum_{i=0}^{n} W_{i}^{(d)} C_{n-i,k-1}^{(d-1)} 
	= \sum_{i=0}^{n} \left(\frac{d-1}{d} \right)R_i (d-1)^{i} C_i C_{n-i,k-1}^{(d-1)}.
	\end{align}
Using the upper bound on $R_i$ \eqref{Vupper}, we can derive an upper bound on $W_{n,k}^{(d)}$ in terms of $C_{n,k}^{(d-1)}$, the number of $d-1$ colored Dyck paths with $n$ pairs and $k$ extra steps:
\begin{align}
	W_{n,k}^{(d)} &\leq  \left(\frac{d-1}{d} \right) \left(\frac{d}{d-2}\right)^2
	\sum_{i=0}^{n}   C^{(d-1)}_i C_{n-i,k-1}^{(d-1)}
	= \frac{d(d-1)}{(d-2)^2} C_{n,k}^{(d-1)}, \label{Wnkupper}
\end{align}
using the recursive relation \eqref{Pnkrecursion} for counting Dyck paths with extra steps, now with $d-1$ colors.
On the other hand, there is a straightforward, loose lower bound 
\begin{align}
	W_{n,k}^{(d)} &\geq C_{n,k}^{(d-1)}, \label{Wnklower}
\end{align}
as there are certainly more PF-model words than $d-1$ colored Dyck paths -- the PF words have $d$ choices for their base layer colors, compared to the $d-1$ for the Dyck paths.

We can also view \eqref{Wk} in terms of generating functions as
\begin{align}
	W^{(d),k}(z) &= \sum_{n=0}^{\infty} W^{(d)}_{n,k} z^n 
	= \sum_{n=0}^{\infty}  \sum_{i=0}^{n} W_{i}^{(d)} C_{n-i,k-1}^{(d-1)}  z^n 
	\nonumber\\
	& = W^{(d)}(z)
		C^{(d-1),k-1}(z) 
	= W^{(d)}(z) \left(C((d-1)x)\right)^{k}.
	\label{generateWk}
\end{align}
Plugging in what we know from \eqref{W0generate}, we have just shown
\begin{align}
	W^{(d),k}(z) = 
	\frac{2(d-1)}{d-2+d\sqrt{1-4(d-1)z}} 
	\left(\frac{1-\sqrt{1-4(d-1)z}}{2(d-1)z}\right)^k.
	\label{WKgenerate}
\end{align}
We discuss the asymptotic scaling of this generating function in detail in Appendix~\ref{subsec:GFAverageNoPairs}.
For now, the upper bound \eqref{Wnkupper} and lower bound \eqref{Wnklower} are enough to give us an estimate on the entanglement entropy of the ground state in the fully-reducible subspace. In Section~\ref{sec:schmidt}, we relate it to that of the $d-1$ colored Fredkin spin chain.

\section{Counting peaks and valleys}
\label{sec:peakcount}

In this Appendix, we first present the calculations of the average number of mountain peaks in Dyck paths, i.e. ()'s in well-bracketed words. We also calculate this for paths that have some extra steps, as well as for colored Dyck paths.
These calculations find use in removing boundary terms in the Fredkin spin chain model in Section~\ref{sec:fredkinreview}.
Finally, in Appendix~\ref{sec:pv}, we also count the valleys, i.e. )('s in Dyck paths. This result will be useful for the considerations of the PF model in Appendix~\ref{sec:countPFpairs}
and \ref{sec:GFasymptopairs}.

\subsection{Dyck paths and their peaks}
\label{sec:peaks0}

Our goal in this Section is to find the average number of mountain peaks in a Dyck path,
as illustrated by the solid circles in Figure~\ref{fig:peakcount}a). 
These correspond to substrings $LR$, or ``()'' if we view the path as a bracketed word. 
We are motivated to do this as we search for a fully translationally-invariant, degeneracy breaking term for the Fredkin spin chain, that would separate the ground states from well-bracketed and other subspaces.

To find the average number of peaks, we will first count the total number of peaks in all Dyck paths, and then divide this by the number of Dyck paths.

\subsubsection{Peaks in well-bracketed words}
Let us start by proving by induction that the total number of mountain peaks in Dyck paths with $n$ pairs is\footnote{This mountain-peak counting sequence starts with 1, 3, 10, 35, 126, 462, 1716, and is also known as OEIS A001700.}
\begin{align}
	\pi_n &= \binom{2n-1}{n} = \frac{1}{2}\binom{2n}{n} = \frac{n+1}{2}\, C_n, \qquad \textrm{for }n\geq 1, \textrm{and }\qquad
	\pi_0 =0, 
	\label{pin}
\end{align}
where $C_n$
are the Catalan numbers. The average number of peaks in well-bracketed words of length 
$N=2n$ (with $n\geq 1$ pairs of brackets) is thus
\begin{align}
	p_n = \frac{\pi_{n}}{C_{n}} 
	= \frac{\binom{2n-1}{n}}{\frac{1}{n+1}\binom{2n}{n}} 
	= \frac{n+1}{2} = \frac{N+2}{4}.
	\label{p0average}
\end{align}

We begin the proof of \eqref{pin} by decomposing the Dyck paths into smaller blocks. There are $C_{n+1}$ Dyck paths of length $2(n+1)$ with $n+1$ bracket pairs. Each such path can be written as $(u)w$, where $u$ is a Dyck path with $i$ pairs and $w$ is a Dyck path with $n-i$ pairs. Note also that the words $()w$ start with a peak.
The total number of peaks in all such words is then
\begin{align}
	\pi_{n+1} &= \sum_{|(u)w|=2(n+1)} \#_{(u)w\textrm{'s peaks}} 
	= \underbrace{C_n + \pi_n}_{\textrm{from }()w} 
	+ \underbrace{\sum_{i=1}^{n} \left(\pi_i C_{n-i} +  C_i \pi_{n-i} \right)}_{\textrm{from words built as }(u)w}
	= C_n + 2\sum_{i=1}^{n} \pi_i C_{n-i},  \label{pnplus1}
\end{align}
as $\pi_0=0$ and $\pi_n C_0=\pi_n$.

The expression \eqref{pin} is valid by inspection for $n=0$ ($\pi_0=0$) and $n=1$ ($\pi_1=1)$. Let us use \eqref{pnplus1} to prove the induction step for \eqref{pin}, assuming \eqref{pin} holds up to $n$:
\begin{align}
	\pi_{n+1} &= C_n + 2\sum_{i=1}^{n} \pi_i C_{n-i} 
	= C_n + 2\sum_{i=1}^{n} \frac{i+1}{2}C_i C_{n-i}
	= C_n - C_n + \sum_{i=0}^{n} (i+1)C_i C_{n-i} \nonumber\\
	& = \sum_{i=0}^{n} C_i C_{n-i} + \sum_{i=0}^{n} i C_i C_{n-i}
	= C_{n+1} + C_{n+1} \underbrace{\sum_{i=0}^{n} i \frac{C_i C_{n-i}}{C_{n+1}}}_{z_{n,1}}.
\end{align}
Here $z_{n,1}$ can be interpreted as the average length (number of bracket pairs) of a well-bracketed block $u$ in words formed as $u(w$ with $|u(w|=2n+1$. It is symmetric under the exchange of $u$ and $w$, which implies $z_{n,1} = n/2$. Therefore,  
\begin{align}
	\pi_{n+1} &= C_{n+1}\left(1 + \frac{n}{2}\right) = \frac{(n+1)+1}{2}\, C_{n+1}.
\end{align}
We have thus proven \eqref{pin} and \eqref{p0average}.

\subsubsection{Peaks in words with some extra brackets}
\label{sec:generalPNK}

Let us now directly calculate the number of peaks for a general number of extra brackets,
using \eqref{pin}. 
There are $C_{n,k}$ Dyck paths with $n$ pairs and $k$ extra steps \eqref{Pk}. Such words are built as $u_1 ( u_2 ( \cdots ( u_{k+1}$ from well bracketed blocks (Dyck paths) $u_1, \dots, u_{k+1}$. By symmetry, when going over all such words, counting the peaks of all of the blocks is $k+1$ times the number of the peaks of the first block. Thus,
\begin{align}
	\pi_{n,k} &= \sum_{z=1}^{k+1} \sum_{|u_1\cdots u_{k+1}|=2n}  \#_{u_z\textrm{'s peaks}} = (k+1) \sum_{|u_1\cdots u_{k+1}|=2n}  \#_{u_1\textrm{'s peaks}} 
	= (k+1) \sum_{i=1}^{n} \pi_i C_{n-i,k-1}.
\end{align}
Let us plug in \eqref{p0average} and continue as
\begin{align}
	\pi_{n,k} &= (k+1) \sum_{i=1}^{n} \frac{i+1}{2} C_i C_{n-i,k-1} 
	= -\frac{k+1}{2} C_{n,k-1} + \frac{k+1}{2} \sum_{i=0}^{n} (i+1) C_i C_{n-i,k-1} 
	\nonumber\\
	&= \frac{k+1}{2} \left( C_{n,k}- C_{n,k-1}\right) + \frac{k+1}{2} C_{n,k}
	\underbrace{\sum_{i=0}^{n} i \frac{C_i C_{n-i,k-1}}{C_{n,k}}}_{z_{n,k}},
	\label{znk}
\end{align}
where $z_{n,k}$ can be interpreted as the average length of the block $u_1$ in a word $u_1 ( u_2 (\cdots (u_{k+1}$
with $n$ bracket pairs and $k$ extra brackets.
Symmetry dictates that the average length of the block $u_1$ is the same as the average length of any block, thus $z_{n,k} = \frac{n}{k+1}$. 
Expressing $C_{n,k-1} 
= C_{n,k} \frac{k(n+k+1)}{(k+1)(2n+k)}$, we continue as
\begin{align}
	\pi_{n,k} 
	&	= C_{n,k} \left(
		\frac{k+1}{2} - \frac{(k+1)}{2} 
		\frac{k(n+k+1)}{(k+1)(2n+k)} + \frac{n}{2}
	\right) \nonumber\\
	&= 
	\frac{C_{n,k}}{2} \left(
		n+k+1 - \frac{k(n+k+1)}{(2n+k)}
	\right) 
	= \frac{C_{n,k}}{4} \left(
		2n+k+2 - \frac{k(k+2)}{2n+k}
	\right).
\end{align}
This translates to an average number of peaks
\begin{align}
	p_{n,k} &= \frac{\pi_{n,k}}{C_{n,k}} 
	= \frac{1}{2} + \frac{2n+k}{4} - \frac{k(k+2)}{4(2n+k)}
	= \frac{N+2}{4} - \frac{k(k+2)}{4N}, \label{pnkfull} 
\end{align}
expressed as a function of $N=2n+k$, the total number of letters in Dyck paths with $n$ pairs and $k$ extra steps. 

\subsection{Counting the average peaks can break ground state degeneracy}

Observe that for a fixed $N$, the number of peaks is smaller for $k>0$ than for $k=0$.
Thus, we have just proved that the average number of ``()'' peaks in well-bracketed words is at least $\oo{N^{-1}}$ greater than for words of other types of the same length. The difference is the smallest between $p_{n}$ and $p_{n-1,2}$, with the gap exactly $\frac{1}{2n}$ (comparing words with even lengths $N=2n$). 
Therefore, we can use the operator
\begin{align}
	H_{\textrm{peaks}} = - \delta \sum_i \ket{LR}\bra{LR}_{i,i+1}
\end{align}
as a perturbation to the Fredkin spin chain Hamiltonian to prefer the ground state close to the one from the well-bracketed subspace, as described in Section~\ref{sec:fredkinreview}.
Because the gap of $H_{\textrm{peaks}}$ is $\delta \oo{N^{-1}}$, we can choose $\delta= N^{-\alpha+1}$, so that the inverse-polynomial gap of the perturbation is much smaller than the inverse-polynomial gap of the Fredkin spin chain Hamiltonian.

\subsubsection{Coloring the Dyck paths and the Fredkin spin chain}

Let us generalize the previous calculations to models with several bracket types. We can color the paths, assigning each bracket pair (and each extra bracket) one of $s$ colors. There are $s^{n+k}$ versions of each Dyck path with $n$ pairs and $k$ extra steps. However, the total number of peaks also increases in the same way, as $\pi_{n,k}^{(s)} = s^{n+k}\pi_{n,k}$ (see also Figure~\ref{fig:peakcount}b). Therefore, the average number of peaks in colored Dyck paths is equal to the uncolored average \eqref{pnkfull} 
\begin{align} 
	p_{n,k}^{(s)} &= p_{n,k} = \frac{2n+k}{4} + \frac{1}{2} - \frac{k(k+2)}{4(2n+k)}
	= \frac{1}{2}+\frac{N}{4}-\frac{k(k+2)}{4N}, 
	\label{pnkfullcolor}
\end{align}	
for $n>0$. It decreases with growing $k$ for a fixed $N=2n+k$, the length of the words with $2n$ pairs and $k$ extra brackets.

Thus, we can also think about the Fredkin spin chain with several types of brackets, 
with the Hamiltonian
\begin{align}
	H^{\textrm{Fredkin}\,(s)} &=
	\sum_i \sum_{c=1}^{s} \sum_{c'=1}^{s} \ket{L_{c'} L_c R_c-L_c R_c L_{c'}}\bra{L_{c'} L_c R_c-L_c R_c L_{c'}}_{i,i+1} \label{multicolorFredkin}\\
 	&+ \sum_i \sum_{c=1}^{s} \sum_{c'=1}^{s} \ket{R_{c'} L_c R_c-L_c R_c R_{c'}}\bra{R_{c'} L_c R_c-L_c R_c R_{c'}}_{i,i+1}\nonumber\\
	&+ \sum_i \sum_{c=1}^{s} \sum_{c'\neq c} \ket{L_c R_c-L_{c'}R_{c'}}\bra{L_c R_c-L_{c'}R_{c'}}_{i,i+1}
	. \nonumber
\end{align}

Now, as we have just shown above, instead of boundary terms, we can use the peak-preferring Hamiltonian 
\begin{align}
	H_{\textrm{peaks}} = - \delta \sum_i \sum_{c=1}^{s} \ket{L_cR_c}\bra{L_cR_c}_{i,i+1}
\end{align}
as a perturbation, with $\delta = 1/\textrm{poly}(n)$, to energetically favor a state very close to the uniform superposition of well-bracketed words over states that are close to uniform superpositions of words with some extra brackets.

\subsection{Counting both peaks and valleys in Dyck paths}
\label{sec:pv}

As an intermediate step, let us count not only the peaks ``()'', but also the valleys ``$)($'' in Dyck paths with $n$ pairs, as in Figure~\ref{fig:peakcount}a).
Let us call the total number of peaks and valleys $\tau_n$ and the average number $t_n$. This does not have a direct application, but we will use it as a tool for counting pairs in the PF model in Section~\ref{sec:countPFpairs}.

Every well bracketed word with $f$ peaks also has $f-1$ valleys. Therefore\footnote{The sequence for $\tau_n$ begins with 1, 4, 15, 56, 210, 792, 3003; it is known as OEIS A001791.},
\begin{align}
	\tau_n &= 2\pi_n-C_n = \frac{2(n+1)}{2}C_n-C_n = nC_n, \label{tn}\\
	t_{n} &= \frac{\tau_{n}}{C_n} = n. \label{tnavg}
\end{align}
Note that he formula $\tau_n = nC_n$ works also for $n=0$, where $\tau_0 = 0$.

Let us continue with words with $n$ bracket pairs and a general number $k$ of extra brackets as in Section~\ref{sec:generalPNK}. Each such words can be built from Dyck paths $u_1,\dots,u_{k+1}$ as $u_1 ( u_2 ( \cdots ( u_{k+1}$. Note that when any of the blocks $u_1, \dots, u_k$ is not empty, it is followed by a valley, because there is an extra bracket that follows the block. The total peak/valley number for these words is then
\begin{align}
\tau_{n,k} 
&= k \underbrace{\sum_{i=1}^{n}C_{i}C_{n-i,k-1}}_{\textrm{a valley after a block}} + (k+1) \underbrace{\sum_{i=1}^{n} \tau_{i}C_{n-i,k-1}}_{\textrm{peaks/valleys in a block}}\nonumber\\
&= -kC_{n,k-1}
+ k \sum_{i=0}^{n}C_{i}C_{n-i,k-1}
-(k+1)\tau_0 C_{n,k-1} + (k+1) \sum_{i=0}^{n} \tau_{i}C_{n-i,k-1} \nonumber\\
&= - kC_{n,k-1} +k C_{n,k}+ (k+1) \sum_{i=0}^{n} i C_i C_{n-i,k-1}. 
\label{peakvalleykSTART} 
\end{align}
where we plugged in \eqref{tn}. 
Let us recall $z_{n,k}$ from \eqref{znk}, and 
$C_{n,k-1} = C_{n,k} \frac{k(n+k+1)}{(k+1)(2n+k)}$, which gives us
\begin{align}
	\tau_{n,k} 
	&= 	k \left( C_{n,k} - C_{n-1,k}\right) 
				+ (k+1)C_{n,k}z_{n,k} 
	= C_{n,k} \left(
				k -\frac{k^2(n+k+1)}{(k+1)(2n+k)} + n
			\right)	\nonumber\\
	&= \frac{1}{2} C_{n,k} 
	\left(
		(2n+k) + k - \frac{k^2(2n+k+k+2)}{(k+1)(2n+k)} 
		\right)
	= C_{n,k}
	\left( \frac{N}{2} + \frac{k}{2(k+1)} - \frac{k^2(k+2)}{2(k+1)N}
		\right), \label{peakvalleyk}\\
  t_{n,k} &= \frac{\tau_{n,k}}{C_{n,k}} 
	 = \frac{N}{2} + \frac{k}{2(k+1)} - \frac{k^2(k+2)}{2(k+1)N}. \label{tnkL}
\end{align}
This expression does not behave as we desired -- it is not a monotonous function of $k$. For a fixed and large total length $N=2n+k$ with $2n\gg k$, the value of $t_{n,k}$ can be larger than $t_{n,0}$.
Differentiating $t_{n,k}$ with respect to $k$, we find that it has a single maximum 
when $N=(2k^2+5k+4)k$, i.e. near $k = (N/2)^{\frac{1}{3}}$.
This means adding peak-and-valley counting terms breaks the degeneracy, but not in the way we intended -- it selects for a nearly-uniform ground state from a subspace with a specific number of extra steps. We conjecture that even this ground state might have a large entanglement entropy.
Furthermore, we will require this calculation when counting letter pairs in Appendix~\ref{sec:countPFpairs}.

Could this calculation have another application? For example, could counting peaks and valleys be useful for a system with {\em periodic boundary conditions}?
One can investigate this by also counting the possible valley/peak across the endpoints when we wrap the chain into a cycle. However, this complicates things quite a bit, as the classes of states that reduce to each other are now modified. For example, the words of the type $)))((($ become well bracketed, as the brackets can now also be paired around the cycle. We leave this investigation for later work.

\subsection{Counting the peaks and valleys of colored Dyck paths}

What happens to \eqref{peakvalleyk} when we color the brackets by $s$ colors as in Figure~\ref{fig:peakcount}b)? We can modify the previous calculation to reflect this. The number of peaks is multiplied by $s^n$, as the coloring the Dyck paths increases their number from $C_n$ to $C_{n}^{(s)} = s^{n}C_n$. We also remember that all but the final peaks of the words have valleys that follow them. However, as we color each bracket pair by some color, only $\frac{1}{s}$ of these actually count as valleys with both sides of the same color. Thus,
\begin{align}
	\tau_{n}^{(s)} = \sum_{\textrm{words}} \left(\textrm{\#peaks} + \frac{1}{s}
	\left(\textrm{\#peaks}-1\right)\right)
	= \frac{s+1}{s} \pi^{(s)}_n - \frac{1}{s}C_n^{(s)} = \frac{(s+1)n + (s-1)}{2s}\,C^{(s)}_{n}, \label{taucolor}
\end{align}
for $n>0$ and we need to remember that $\tau_0^{(s)}=0$. We can check that the special case for $s=1$ is what we had before, $\tau_{n}^{(1)} = \tau_n = n C_n$.

Let us calculate the peak/valley number for the colored Dyck paths that end at a level $k$ above the horizontal. Similar to \eqref{peakvalleykSTART}, we now have
\begin{align}
	\tau_{n,k}^{(s)} 
&= \frac{k}{s} \underbrace{\sum_{i=1}^{n}C_{i}^{(s)}C^{(s)}_{n-i,k-1}}_{\textrm{a valley after a block}} + (k+1) \underbrace{\sum_{i=1}^{n} \tau_{i}^{(s)}C^{(s)}_{n-i,k-1}}_{\textrm{peaks/valleys in a block}}.
\end{align}
The calculation continues as 
\begin{align}
	\tau_{n,k}^{(s)} 
	&= 
-\frac{k}{s} C_{n,k-1} + 
\frac{k}{s} \sum_{i=0}^{n}C_{i}^{(s)}C^{(s)}_{n-i,k-1}
+ s^n (k+1) \sum_{i=1}^{n} \frac{(s+1)i+s-1}{2s}\,C_i C_{n-i,k-1} \nonumber\\
&= \frac{ks^n}{s} \left( C_{n,k} - C_{n,k-1}\right) 
-\frac{s^n (k+1)(s-1)}{2s} C_{n,k-1}
+ s^n (k+1) \sum_{i=0}^{n} \frac{(s+1)i+s-1}{2s}\, C_i C_{n-i,k-1} \nonumber\\
&= s^n C_{n,k} \left(
		\left(
			\frac{k}{s} 
				+ \frac{(k+1)(s-1)}{2s}
		\right)
			\left( 1- \frac{C_{n,k-1}}{C_{n,k}}\right) 
		 + \frac{(k+1)(s+1)n}{2s(k+1)} 
	\right) \nonumber\\
& = s^n C_{n,k} \,\frac{1}{2s}
		\left(
			\left(
				2k + ks -k + s -1
			\right)
			\left(
				1 - \frac{k(n+k+1)}{(k+1)(2n+k)} 
			\right)
		+ (s+1)n
	\right) \nonumber\\
& = s^n C_{n,k} \,\frac{1}{2s}\left(
		k(s+1)+(s-1)
		- \frac{\left(k(s+1)+(s-1)\right)k(N+k+2)}{2N(k+1)}
		+ \frac{(s+1)(N-k)}{2}
	\right) \nonumber\\
& = s^n C_{n,k} 
	\left(
		\frac{(s+1)N+2(s-1)}{4s} + k \left(
						\frac{s+1}{4s}
						-\frac{\left(k(s+1)+(s-1)\right)}{4s(k+1)}
						\right)
		- \frac{k(k(s+1)+s-1)(k+2)}{4s(k+1)N}
	\right), \nonumber
\end{align} 
finally resulting in 
\begin{align}
	\tau_{n,k}^{(s)} 
	&= s^n C_{n,k} 
	\left(
		\frac{(s+1)N+2(s-1)}{4s} + \frac{2k}{4s(k+1)}
		- \frac{k(k+2)\left(k(s+1)+(s-1)\right)}{4s(k+1)N}
	\right),
\label{peakvalleykcolored}
\end{align}
for $n>0$, while $\tau^{(s)}_{0,k} = 0$, where we used $z_{n,k} = \frac{n}{k+1}$, and $C_{n,k-1} = C_{n,k} \frac{k(n+k+1)}{(k+1)(2n+k)}$. 
This result is consistent with \eqref{peakvalleyk} for $s=1$. 
Similarly to the uncolored version \eqref{tnkL}, 
the average number $t^{(s)}_{n,k} = \tau^{(s)}_{n,k}/C^{(s)}_{n,k}$
for a fixed chain length $N=2n+k$ is not monotonous with $k$.
Differentiation with respect to $k$ tells us it has a single maximum
when $N-1 = k^3+s(k+1)^3+2k^2+k$, i.e. near $k=\left(\frac{N}{s+1}\right)^{\frac{1}{3}}$.
Again, we challenge the reader to investigate whether a ground state from a subspace with such $k$ remains very entangled.

This concludes the investigation of the average number of peaks and valleys in $s$-color Dyck paths. In the next Section, we turn to the PF model with $d$ letters. Our goal is to find the average number of subsequent letter pairs in PF model words.
Recalling its mapping to a subset of $d$-colored Dyck paths, we realize identical letter pairs in the PF model correspond to peaks, as well as same-color valleys in Dyck paths, so we will be able to rely on what we derived in this Section.

\section{Counting neighboring identical letter pairs in the PF model}
\label{sec:countPFpairs}

In order to break the ground state degeneracy of the PF model, we suggest energetically preferring neighboring identical letter {\em pairs}. For this, we need to accurately count the average number of times a {\em pair} (e.g. \dots AA\dots) appears in PF model words with various irreducible strings (see also Figure~\ref{fig:peakcount}c). We will build on the results from Appendix~\ref{sec:peakcount}.

First, in Section~\ref{sec:countPF2}, we find a closed-form expression of the average number of neighboring letter pairs in the qubit ($d=2$) PF model.
Then in Section~\ref{sec:countPF3}, we prove a recursive expression for counting the pairs in models with $d\geq 3$, and sketch a direction for further analysis. The result does not have a closed form, but we can go quite far in using it numerically.
In particular, in 
Figure~\ref{fig:pairnumerics} we investigate the $d=3$ color PF model up to $n\leq 1000$.
We find strong evidence that 
\begin{align}
	\langle \# \rangle_{k=0}^{(3)} &= \frac{3}{4}(n+1) + \bigO\left(n^{-1}\right), \\
	\langle \# \rangle_{0}^{(3)}  - \langle \# \rangle_{k} &\sim \Theta\left(\frac{k^2}{n}\right),
\end{align}
i.e. that the average number of pairs for $d=3$ color PF model words of length $2n$ is approximately $3(n+1)/4$ in the fully reducible subspace, while it is on the order of $k^2/n$ lower in the subspaces with $k$ irreducible letters. The first claim is proven in \eqref{fn32}, while we show the second for constant $k$ in Appendix~\ref{sec:GFasymptopairs} \eqref{AVGfinalConst}.

\subsection{Neighboring identical letters in the $d=2$ (qubit) PF model words}
\label{sec:countPF2}

Let us start with the PF model with $d=2$, whose words are made from the alphabet $\{1,2\}$. First, we will look at the fully reducible words, and count their pairs, i.e. how many times the letter combinations $11$ and $22$ appear in them.
Viewing the words as their mountain profiles, this is essentially counting the peaks and valleys of the words, as in Figure~\ref{fig:peakcount}c).

We will prove by induction that the total number $\varphi_n$ of pairs ($11$ or $22$) in all fully reducible $d=2$ PF words made from $2n$ letters is\footnote{The integer sequence 2, 12, 60, 280, 1260, 5544, 24024, \dots, is known as OEIS A005430 (the Ap\'{e}ry numbers).} 
\begin{align}
	\varphi_n = nW_n. \label{phinguess}
\end{align}
and that the average number of pairs in such words is $f_n = \varphi_n/W_n = n$,
with 
	$W_n 
	= \binom{2n}{n}$,
the number of fully reducible 2-color PF model words
from \eqref{reducible2}.

In Section~\ref{sec:PF3start}, 
we have shown that for $n\geq 2$ one can rewrite the words into blocks as $AuAw$, made from a letter pair $A\in \{1,2\}$, a Dyck path $u$, and another $d=2$ PF model word $w$.
There are three cases depending on how the word is formed: $AAw$, $AuAw$ and $AuA$, with $u$ a Dyck path (whose base layer can not contain the letter $A$) and $w$ a fully reducible word. 
These give us 
\begin{align}
	\frac{1}{2}\varphi_{n+1} &= \underbrace{ W_{n}+\frac{1}{2}W_{n}
			+\varphi_{n} }_{\textrm{from $AAw$}}
	+  \underbrace{\sum_{i=1}^{n-1} \left(
				\tau_i W_{n-i} + \frac{1}{2}C_i W_{n-i} + C_i \varphi_{n-i}
			\right)
			}_{\textrm{from $AuAw$}}
	+ \underbrace{\tau_{n}
		}_{\textrm{from $AuA$}},
\end{align}
as the paths $AA w$ have an initial pair, a possible pair if the first letter of $w$ is $A$, and the pairs coming from the block $w$. Next, the $AuAw$'s have pairs in the block $u$,
a possible pair at the beginning of $w$, and pairs in $w$. Finally, we add the possible pairs of $u$ from the $AuA$'s. 
The initial $\frac{1}{2}$ comes from the two choices for $A\in \{1,2\}$.
We also realize that the number of pairs inside the block $u$ is equal to the number of peaks and valleys $\tau_i$ \eqref{tn} of the corresponding Dyck path.
Thus, we get a recursive expression
\begin{align}
 \varphi_{n+1} &= 3W_n + 2 \varphi_n + 2 \tau_n
	- 2\tau_n - W_n - C_n - 2\varphi_n
	+ 2 \sum_{i=0}^{n} \left(
				\tau_i W_{n-i} + \frac{1}{2}C_i W_{n-i} + C_i \varphi_{n-i}
			\right) \nonumber\\
 &= 2W_n - C_n	+ 2 \sum_{i=0}^{n} \left(
				\tau_i W_{n-i} + \frac{1}{2}C_i W_{n-i} + C_i \varphi_{n-i}
			\right)\nonumber\\
 &= 2W_n - C_n	+ 2 H_n + G_n + 2\sum_{i=0}^{n} \varphi_{i} C_{n-i}, 
	\label{phink}
\end{align}
where we label\footnote{We can actually use induction to prove $F_n = \sum_{i=0}^{n} \binom{2i}{i}\binom{2(n-i)}{n-i} = 4^n$, but this is not required for our proof of \eqref{phinguess}.}
\begin{align}
 G_n &= \sum_{i=0}^{n} C_i W_{n-i} 
		= \frac{1}{2}W_{n+1}, \label{WC} \\
 H_n &= \sum_{i=0}^{n} i C_i W_{n-i} 
		= \sum_{i=0}^{n} (i+1-1) C_i W_{n-i} 
		= \sum_{i=0}^{n} \left(W_i -C_i\right)W_{n-i}
		= F_n - G_n, \label{iWC} \\
	F_n &= \sum_{i=0}^{n} W_i W_{n-i}.
\end{align}
The reason for \eqref{WC} is composition: the $d=2$ fully reducible PF model words are built as $AuAw$, with $u$ a Dyck path and $w$ a PF model word, and there are 2 choices for the letter $A$.

We are now ready to prove \eqref{phinguess} by induction.
The base cases $n=0,1$ work by inspection, with $\varphi_0=0$, and $\varphi_1 = 2$ (there are $W_1=2$ such words, $11$ and $22$, and each has a pair).
For the induction step, we continue with \eqref{phink},
utilizing
$2W_n - C_n = 
		2\binom{2n}{n}-\frac{1}{n+1}\binom{2n}{n}
		= \frac{2n+1}{n+1} \binom{2n}{n} = \frac{1}{2}\binom{2n+2}{n+1} = \frac{1}{2}W_{n+1}=G_n$. Thus,
\begin{align}
	\varphi_{n+1} 
		&= G_n	+ 2(F_n-G_n) + G_n 
		+ 2\sum_{i=0}^{n} \left( n -(n-i)\right)W_i C_{n-i}.
\end{align}
Using \eqref{WC} and \eqref{iWC}, we obtain
\begin{align}
	\varphi_{n+1}
		&= 2F_n + 2n G_n
		- 2\left(F_n-G_n\right) 
		=  (2n+2)G_n = (n+1)W_{n+1},
\end{align}
which we wanted to prove.

\subsubsection{Neighboring letter pairs in the $d=2$ PF model words that reduce to nonempty strings}
Next, we will count the pairs in PF words that reduce to something nonempty.
There are $W_{n,k}$ such words \eqref{reducible20k}.
The words that reduce to $X_1 \dots X_k$ are built as $w X_1 u_1 \dots X_k u_k$, from an initial fully reducible word $w$, and then from $k$ blocks $u_i$ that correspond to uncolored Dyck paths (see Section~\ref{sec:correspond} for the mapping).
To count the number of pairs in these words, we count the pairs in the block $w$, the possible pair at the end of the block $w$, the pairs in the blocks $u_1,\dots,u_{k}$ as well as the pairs at the ends of the blocks $u_1,\dots,u_{k-1}$:
\begin{align}
	\varphi_{n,k} &= 
	\sum_{i=1}^{n} \left(
		\varphi_i C_{n-i,k-1}
		+ \frac{1}{2}W_i C_{n-i,k-1}
	  +	k  \tau_i W_{n-i,k-1}
		+ (k-1)  C_i W_{n-i,k-1} \right) \nonumber\\
	&= -\frac{1}{2}C_{n,k-1} - (k-1)W_{n,k-1} \nonumber\\
	&+
		\sum_{i=0}^{n} \left(
			\left( \varphi_i + \frac{1}{2}W_i \right) C_{n-i,k-1}
			+ \left(k  \tau_i +  (k-1)C_i \right) W_{n-i,k-1} \right).
			\label{phinkPF}
\end{align}

First, we will calculate or reexpress a few sums. 
When a PF model word $wX_1 u_1 X_2 \dots X_k u_k$ reduces to $X_1\cdots X_k$, we can rewrite it as $w' X_k u_k$ as well as $w X_1 u'$ and thus prove
\begin{align}
	\sum_{i=0}^{n} C_i W_{n-i,k-1} 
	 = \sum_{i=0}^{n} W_i C_{n-i,k-1} 
			= W_{n,k}. \label{CWk} 
\end{align}
The sum made from peak/valley numbers for Dyck paths and the number of PF words is
\begin{align}
	\sum_{i=0}^{n} \tau_i W_{n-i,k-1} 
			&= \underbrace{\sum_{i=0}^{n} W_i W_{n-i,k-1}}_{F_{n,k-1}} - \sum_{i=0}^{n} C_i W_{n-i,k-1} 
			= F_{n,k-1} - W_{n,k},
	\label{tauW}
\end{align}
because $\tau_i = i C_i = \left(i+1-1\right) C_i = W_i - C_i$.
Third, the combination of the pair numbers for fully reducible PF words and numbers of Dyck paths gives us:
\begin{align}
	\sum_{i=0}^{n} \varphi_i C_{n-i,k-1} 
			&= \sum_{i=0}^{n} i W_i C_{n-i,k-1} 
			= \sum_{i=0}^{n} \frac{ik}{n-i+k} W_i W_{n-i,k-1} 
			= \sum_{i=0}^{n} \left(\frac{k(n+k)}{n-i+k}-k\right) W_i W_{n-i,k-1}
			\nonumber\\
			&= (n+k)\sum_{i=0}^{n} W_i C_{n-i,k-1} - kF_{n,k-1}
			= (n+k) W_{n,k}  - kF_{n,k-1}.
	\label{phiP}
\end{align}
We also recall that
\begin{align}
	W_{n,k-1} &= \binom{2n+k-1}{n} = \frac{n+k}{2n+k}\binom{2n+k}{n} = 
	\frac{n+k}{2n+k}W_{n,k}, \\
	C_{n,k-1} &= 
		\frac{k}{n+k}\binom{2n+k-1}{n}
		=\frac{k}{2n+k}\binom{2n+k}{n} = \frac{k}{2n+k}W_{n,k}, \\
	C_{n,k} &= 
		\frac{k+1}{n+k+1}\binom{2n+k}{n}
		= \frac{k+1}{n+k+1}W_{n,k}.
\end{align}
Let us now plug \eqref{tauW} and \eqref{phiP} into \eqref{phinkPF}.
\begin{align}
	\varphi_{n,k} &= 
			-\frac{1}{2}C_{n,k-1} - (k-1)W_{n,k-1}
			+ (n+k) W_{n,k} - kF_{n,k-1}
			+ \frac{1}{2}W_{n,k} \nonumber\\
			&+ k \left(F_{n,k-1} - W_{n,k}\right)
			+ (k-1)W_{n,k} \nonumber\\
	&= W_{n,k} \left(
			-\frac{k}{2(2n+k)} -\frac{(k-1)(n+k)}{(2n+k)}
				+ (n+k) + \frac{1}{2} - k + k-1
	\right)\nonumber\\
	&=W_{n,k}
	\frac{-k-2(k-1)(n+k)+(2n+2k-1)(2n+k)}{2(2n+k)}
	= W_{n,k} \left(n + \frac{nk}{2n+k}\right).
\end{align}
This gives us the average number of subsequent letter pairs in PF model words of fixed length $N=2n+k$ that reduce to strings of length $k$:
\begin{align}
	f_{n,k} &= \frac{\varphi_{n,k}}{W_{n,k}} 
		\left(n + \frac{nk}{2n+k}\right) = \frac{N}{2}-\frac{k^2}{2N}. \label{fnk2}
\end{align}

\subsubsection{Breaking ground state degeneracy by counting pairs in the $d=2$ PF model}
\label{sec:PF2break}

The result \eqref{fnk2} immediately implies that the number of pairs (on even-length chains) is the largest for $k=0$ words, and monotonously decreases with growing $k$. The next highest number of pairs (for $k=2$) is at least $\oo{N^{-1}}$ smaller. Therefore, for the 2-color pair-flip (PF) model, we can now use a pair-counting term
\begin{align}
	- \delta \sum_{i=1}^{N-1} \left( \ket{11}\bra{11} + \ket{22}\bra{22}\right)_{i,i+1},
\end{align}
with a small $\delta=1/\textrm{poly}(N)$ as a perturbation to energetically prefer a unique ground state -- a state very close to the uniform superposition of all fully reducible words.

The new ground state comes from the fully reducible subspace. However, we still need to analyze how exactly its entanglement entropy responds to such a perturbation. 
We can show that the differences in amplitudes of transition-connected words can not be large, as this would manifest itself in the energy of this state, which is close to $-\delta \frac{N}{2}$, i.e. as close to zero as we wish.
Moreover, the distribution of the words with different pair numbers within a subspace is closely centered about the average. Therefore, even if the states with more pairs are slightly ``preferred'', the resulting perturbed ground state does not differ much from the original, its Schmidt decomposition structure remains largely intact, and the entropy scaling with chain length remains what it was. However, we leave a formal statement and proof of this as an open question. 


\subsection{The PF model with $d\geq 3$ colors}
\label{sec:countPF3}

Let us now look at the PF model with at least three different letters, which complicates the combinatorics quite a bit. Counting the number of words that reduce to a particular string is by itself is no longer an easy task. We do not have a simple formula for this, only iterative relations, and an understanding of its asymptotic scaling.

Let us derive the iterative relations for counting the fully-reducible words of length $N=2n$ (with $n$ letter pairs) in the $d$-colored PF model by looking at how they can be built. The very first letter pairs with another letter in such a word, so we can write it as $AuAv$, where $u$ is an $d-1$ color Dyck path whose base level can not start with the letter $A$ (such a letter would be paired with the initial $A$) and $v$ is a fully reducible word.
\begin{align}
	W_{n}^{(d)} &= \sum_{i=0}^{n-1} d C^{(d-1)}_{n-i-1} W_{i}^{(d)}
							= \sum_{i=0}^{n-1} d(d-1)^{n-i-1} C_{n-i-1} W_{i}^{(d)}.
\end{align} 
The number $W_{n}^{(d)}$ asymptotically scales like the number of $d-1$ colored Dyck paths, as we have shown in \eqref{Vexact}.
Meanwhile, words that reduce to a specific $k$-letter irreducible string can be built from as $vAu'$, with a fully reducible word $v$, and an $d-1$ colored Dyck path $u'$ with $k-1$ extra brackets, giving us the recursive relation\footnote{For example, it allows us to find $W_{1,1}^{(d)} = 2d-1$ and $W_{2,1}^{(d)} = 5d^2-6d+2$, as well as $W_{1,2}^{(d)} = 3d-2$ and $W_{2,2}^{(d)}=9d^2-13d+5$. For $d=3$, the sequence of fully reducible word numbers is the sequence OEIS A089022, for words that reduce to one letter, $W_{n,1}^{(3)}$ is OEIS A194723. We list the first few terms of these in Figure~\ref{fig:PFpairnumerics}.}
 \eqref{Wk}, 
$W_{n,k}^{(d)} = \sum_{i=0}^{n} W_{i}^{(d)} C^{(d-1)}_{n-i,k-1}$.

Let us now count the total number of subsequent letter pairs in such words. We start with the fully reducible words, just like in the previous Section. 
For $n=1$, the total number of subsequent letter pairs in all these words is simply $d$ times the number of pairs in the word $AA$, i.e. $d$. $\varphi_1^{(d)} = d$, and it is also easy to calculate $\varphi_2^{(d)} = 3d^2$. For $n\geq 2$, we can rewrite each fully reducible word in terms of blocks as
$AAv$, $AuAv$ or $AuA$, where $A$ is one of the $d$ letters (that is why we need to multiply the result by $d$), $u$ is a $d-1$ colored Dyck path, and $v$ is a fully reducible word.
\begin{align}
	\varphi_{n+1}^{(d)} &= \underbrace{ d W_{n}^{(d)}+\frac{d W_{n}^{(d)}}{d}
			+d \varphi_{n}^{(d)} }_{\textrm{from $AAv$}}
	+  \underbrace{\sum_{i=1}^{n-1} \left(
				d\tau_i^{(d-1)} W_{n-i}^{(d)} + \frac{dC_i^{(d-1)} W_{n-i}^{(d)}}{d} + dC_i^{(d-1)} \varphi_{n-i}^{(d)}
			\right)
			}_{\textrm{from $AuAv$}}
	+ \underbrace{d\tau_{n}^{(d-1)}
		}_{\textrm{from $AuA$}} \nonumber \\
	&= d W_{n}^{(d)} + \frac{W_{n+1}^{(d)}-dC_n^{(d-1)}}{d} + d \varphi_{n}^{(d)} + d\tau_{n}^{(d-1)}
	+  \sum_{i=1}^{n-1} \left(
				d\tau_i^{(d-1)} W_{n-i}^{(d)}  + dC_i^{(d-1)} \varphi_{n-i}^{(d)}
			\right) 
			\nonumber \\
	&= d W_{n}^{(d)} + \frac{W_{n+1}^{(d)} - d C_n^{(d-1)}}{d}
	+  \sum_{i=0}^{n} \left(
				d\tau_i^{(d-1)} W_{n-i}^{(d)}  + d C_i^{(d-1)} \varphi_{n-i}^{(d)}
			\right)
			, \label{pairsPFrecursive}
\end{align}
as $\sum_{i=1}^{n-1} dC_i^{(d-1)} W_{n-i}^{(d)} = W_{n+1}-dW_{n} - dC_n$,
and recalling $\tau_0^{(d-1)}=\varphi_0^{(d)}=0$ so that we can sum from $i=0$ to $i=n$.
We use this recursive formula for numerics, with the starting point $\varphi_1 = d$.

Similarly, for words that reduce to $k$ extra letters, we modify \eqref{phinkPF} to accomodate for $d$ colors, viewing the words as $v A_1 u_1 \dots A_k u_k$. We obtain the recursive formula
\begin{align}
	\varphi_{n,k}^{(d)} &= 
	\sum_{i=1}^{n} \left(
		\varphi_i^{(d)} C_{n-i,k-1}^{(d-1)}
		+ \frac{1}{d}W_i^{(d)} C_{n-i,k-1}^{(d-1)}
	  +	k  \tau_i^{(d-1)} W_{n-i,k-1}^{(d)}
		+ \frac{(k-1)}{(d-1)}  C_i^{(d-1)} W_{n-i,k-1}^{(d)} \right) 
			\label{pairsPFKrecursive}
\end{align}
This formula, together with the recursive formula for $W_{n,k}^{(d)}$ \eqref{recursionPCcolored} is well suited for numerical experiments\footnote{For the reader interested in checking their own numerical tools, we list the first few terms for the counting in Figure~\ref{fig:PFpairnumerics}.
}.
Our numerics up to $n=1000$ and $k=1000$ (see Figure~\ref{fig:pairnumerics}) indicate that 
\begin{align}
	f_{n}^{(3)}-f_{n-k,2k}^{(3)} 
	= \frac{\varphi_{n}^{(3)}}{W^{(3)}_{n}}-\frac{\varphi_{n-k,2k}^{(3)}}{W^{(3)}_{n-k}}
	= \oo{\frac{k^2}{N}}, \label{PFpairgap}
\end{align}
i.e. that the fully reducible words with length $N=2n$ have on the order of $\frac{k^2}{N}$
more subsequent letter pairs than words of the same length that have irreducible strings of length $2k$.
Later, in Appendix~\ref{sec:GFasymptopairs} we prove this behavior for constant $k$, 
and numerically probe that it persists for $k$ scaling with $n$ as well.

We thus claim we can use pair-counting
\begin{align}
	-\delta \sum_{i=1}^{N-1} \sum_{t=1}^{d} \ket{tt}\bra{tt}_{i,i+1}
\end{align}
with a small, inverse-polynomial $\delta$ as a perturbation to the PF model Hamiltonian to energetically prefer the ground state from the fully reducible subspace over the ground states from the other subspaces, with an inverse-polynomial resulting gap. 
Because we know that there is an inverse polynomial gap in each of those subspaces, such a perturbation means the ground state becomes unique.
Moreover, the new ground state is very close to the original one, similarly to what we claim at the end of Section~\ref{sec:PF2break}, and its entropy scaling remains proportional to $\sqrt{N}$.
We leave the exact statement and proof of its entanglement robustness under this perturbation as an open question.

In the next Section, we exhibit an analytical approach
to counting the pairs. However, it does not go all the way to analytically proving the gap \eqref{PFpairgap}, and more work is required in the future.


\subsection{A different approach: counting the PF model words using blocks}

For completeness, let us present another approach to counting the average number of subsequent letter pairs. It allows us to calculate the leading terms in $f_n^{(d)}$ and $f_{n-1,2}^{(d)}$. However, to find the gap \eqref{PFpairgap} between those two expressions analytically, we would have to calculate the next order term, which requires much more effort.

Nevertheless, the {\em block} approach is interesting on its own. There is yet another way to write down the PF model words -- dividing them into blocks: $A_1 a_1 A_1 \, A_2 a_2 A_2 \dots A_b a_b A_b$, with base-level letters $A_i$ chosen from $d$ letters, encapsulating blocks $a_i$ that correspond to $d-1$ colored Dyck paths. 
There are $d^b$ ways to choose the $A_i$'s, and we can count the ways to choose the $a_i$'s as the number of $(d-1)$-color Dyck paths $a_1 ( a_2 ( \cdots ( a_b$ with $b-1$ extra brackets (with fixed colors). Therefore, we have
\begin{align}
	W_{n}^{(d)} &= \sum_{b=1}^{n} d^b (d-1)^{n-b} C_{n-b,b-1} 
	= \underbrace{(d-1)^n}_{Y} \sum_{b=1}^{n} \underbrace{\left(\frac{d}{d-1}\right)^b}_{X^b}  C_{n-b,b-1} \nonumber\\
	&= Y \sum_{b=1}^{n} X^b C_{n-b,k-1}
	= Y \sum_{b=1}^{n} \frac{b}{n} X^b  \binom{2n-b-1}{n-b},
	\label{Wdifferent}
\end{align}
where we labeled $X= \frac{d}{d-1}$ and $Y = (d-1)^n$.
The sum thus goes over the different number of blocks in PF model words. 
It turns out that the average number of these blocks 
\begin{align}
	\bar{b}_{n}^{(d)} 
	= \frac{1}{W_{n}^{(d)}}\sum_{b=1}^{n} b\, d^b (d-1)^{n-b} C_{n-b,k-1}
	= \frac{\sum_{b=1}^{n} b X^b  C_{n-b,b-1}}{\sum_{b=1}^{n} X^b C_{n-b,b-1}}.
\end{align}
grows slowly towards a constant which we find below in \eqref{averageBLOCKSwn}.

\subsubsection{Counting blocks in Dyck paths}

Let us do a few calculations about Dyck paths.
We can view a Dyck path $A_1 a_1 A_1 \dots A_b a_b A_b$ as made from $b$ blocks, giving us
\begin{align}
	C_n = \sum_{b=1}^{n} C_n^{[b\textrm{\,blocks}]} = \sum_{b=1}^{n} C_{n-b,b-1},
\end{align}
because the number of Dyck paths with $n$ pairs made from exactly $b$ blocks
is the same as the number of Dyck paths of the form $a_1 A_1 a_2 A_2 \dots a_{b-1} A_{b-1} a_b$, with $n-b$ pairs and $b-1$ extra letters:
\begin{align}
	C_n^{[b\textrm{\,blocks}]} = C_{n-b,b-1}. \label{PCblocks}
\end{align}

On the other hand, we can view Dyck paths with $n$ pairs that end at a level $k$ above the horizon as
$A_1 a_1 A_1 \dots A_{b-k} a_{b-k} A_{b-k}
\, a_{b-k+1} A_{b-k+1}\,\dots a_{b-1} A_{b-1}\, a_b$,
with an irreducible string $A_{b-k}\dots A_{b-1}$ of length $k$.
The initial $b-k$ blocks might not be present.
Counting these words, we find their number is the same as the number of (well-bracketed) Dyck paths with $n+k$ pairs, which can be built as 
$A_1 a_1 A_1 \dots A_{b}a_bA_{b}$, with at least $k$ blocks. Therefore, we have
\begin{align}
	C_{n,k} = C_{n+k}^{[\geq k]} = \sum_{b=k}^{n+k} C_{n+k-b,b-1}
	= C_{n+k} - \sum_{b=1}^{k-1} C_{n+k-b,b-1}. \label{Pfromblocks}
\end{align}
Our goal is to use these expressions to find the scaling of $C_{n-b,b-1}$ that appears in the counting for $W_n^{(d)}$ \eqref{Wdifferent}. 
This implies
\begin{align}
	C_{n}^{[b\textrm{\,blocks}]} = C_{n-b,b-1} 
	&= C_{n-1} - \sum_{i=1}^{b-2} C_{n-1}^{[i\textrm{\,blocks}]}
	= C_{n-1} - \sum_{i=1}^{b-2} C_{n-1-i,i-1} \nonumber\\
	&= C_{n-(b-1),(b-1)-1} - C_{n-1}^{[b-2\textrm{\,blocks}]} \nonumber\\
	&= C_{n-b+1,b-2} - C_{n+1-b,b-3} 
	= C_{n}^{[b-1\textrm{\,blocks}]} - C_{n-1}^{[b-2\textrm{\,blocks}]}, \label{CBLOCKS}
\end{align}
valid for $b\geq 2$, with $C_{n}^{[0\textrm{\, blocks}]}=0$, $C_{n}^{[1\textrm{\,block}]} = C_{n-1}$, and $C_{n}^{[2\textrm{\,blocks}]} = C_{n-1}$.

\subsubsection{The scaling of $C_{n,k}$}

Using the block counting approach, we can find the asymptotic scaling of $C_{n,k}$, starting with \eqref{CBLOCKS}.
The function $C_{n}^{[k]}$ decreases with growing $k$, because
\begin{align}
	C_{n}^{[k]} - C_{n}^{[k+1]} = C_{n}^{[k]} - C_{n}^{[k]} + C_{n-1}^{[k-1]}
		= C_{n-1}^{[k-1]} > 0,
\end{align}
if $C_{n-1}^{[k-1]} > 0$, i.e. for $n\geq k$, and for $n=k$ we get $C_{n}^{[n]} = 1$ and $C_{n}^{[\geq n]} = 0$.

The recursive equation \eqref{CBLOCKS} has a solution $2^{-k}(c_1 k + c_2)$.
Thus, when we plug in the first two values to $C_{n}^{[1]}=C_n^{[2]}=C_{n-1}$, we obtain an upper bound
\begin{align}
		C_{n}^{[k]} \leq (2k)\, 2^{-k}\, C_{n-1}, \label{CNKapprox}
\end{align}
as the ratio
\begin{align}
		\frac{C_{n}^{[k]}}{C_{n}^{[k-1]}} 
		= \frac{C_{n-k,k-1}}{C_{n-k+1,k-2}}
		= \left(\frac{k}{k-1}\right) \left(\frac{n-k+1}{2n-k}\right)
		= \frac{1}{2}\left(\frac{k}{k-1}\right) \left(1-\frac{k-2}{2n-k}\right)
		\leq \frac{1}{2}\left(\frac{k}{k-1}\right),
\end{align}
for $k\geq 2$. 
In fact, \eqref{CNKapprox} is not only an upper bound but also a good approximation when $k\ll \sqrt{n}$.

\subsubsection{Scaling of fully reducible PF model words}

Now, we can find an upper bound on $W_{n,k}$ using \eqref{CNKapprox}.
\begin{align}
	W_{n}^{(d)} &= Y\sum_{b=1}^{n} X^b C_{n}^{[b\textrm{\, blocks}]} 
	\leq 2Y C_{n-1} \sum_{b=1}^{n} b \left(\frac{X}{2}\right)^b 
	\sim 2Y C_{n-1} \frac{x}{\left(1-x\right)^2} \nonumber\\
	&
	= 2 (d-1)^n C_{n-1} \frac{2d(d-1)}{(d-2)^2}
	= \frac{d(d-1)}{(d-2)^2} C_n^{(d-1)} + \bigO\left(n^{-1} C_n^{(d-1)}\right),
	\label{Wblocknumber}
\end{align}
with $x=\frac{X}{2}=\frac{d}{2(d-1)}$.
Note that we have have previously derived this result \eqref{Vexact} using generating functions.

We can now estimate the average number of blocks in the PF model words.
\begin{align}
	\bar{b}^{(d)}_n = 
	\frac{\sum_{b=1}^{n}bW_{n}^{(d)[b]}}{W_n^{(d)}}
	=\frac{\sum_{b=1}^{n} b X^b C_n^{[b]}
		}{\sum_{b=1}^{n}X^b C_n^{[b]}
		} 
	\sim \frac{\sum_{b=1}^{n} b^2 x^b}{\sum_{b=1}^{n} b x^b}  
	\sim \frac{1+x}{1-x} 
	= \frac{3d-2}{d-2}. \label{averageBLOCKSwn}
\end{align}
where 
$W_{n}^{[b]}$ stands for PF model words with $n$ pairs and $b$ blocks,
$x=\frac{X}{2}=\frac{d}{2(d-1)}$, and recalling
$\sum_{z=0}^{\infty} z x^z = \frac{x}{(1-x)^2}$ and $\sum_{z=0}^{\infty} z^2 x^z = \frac{x(1+x)}{(1-x)^3}$. This means the average number of blocks in $d=3$ color PF model words is $\bar{b}^{(3)}= 7$.

\subsubsection{Scaling of PF model words with $k$ extra letters}

Now for the more interesting part, let us find the scaling of $W_{n,k}^{(d)}$:
\begin{align}
	W_{n,k}^{(d)} &= \frac{1}{d^k} W_{n+k}^{[\geq k]}
	  = (d-1)^{n} \sum_{b=0}^{n} X^b C_{n,k}^{[b]} \nonumber\\
		&	= \frac{(d-1)^{n+k}}{d^k} \sum_{z=0}^{n} X^{k+z} C_{n+k}^{[k+z]} 
		= (d-1)^{n} \sum_{z=0}^{n} X^{z} C_{n+k}^{[k+z]} \\
		&\leq
			(d-1)^n \sum_{z=0}^{n} X^z 2(k+z) 2^{-(k+z)} C_{n+k-1} 
		= 2^{-k}\, 2(d-1)^{n}C_{n+k-1}
			\sum_{z=0}^{n} (k+z) x^z, \nonumber
\end{align}
with $x=\frac{X}{2}=\frac{d}{2(d-1)}$. Calculating the sum approximately, assuming large $n$, we obtain
\begin{align}
	W_{n,k}^{(d)} &\sim
	2^{-k}\, 2(d-1)^{n} C_{n+k-1}
	\left(\frac{k}{1-x} + \frac{x}{(1-x)^2}\right)
			\nonumber\\
	&=2^{-k}\,2 (d-1)^{n}C_{n+k-1} 
			\frac{4(d-1)^2}{(d-2)^2}
			\left(\frac{k(d-2)+d}{2(d-1)}\right) \nonumber\\
	&=2^{-k} 4 C_{n+k-1} (d-1)^{n} \frac{(d-1)\left(k(d-2) + d\right)}{(d-2)^2}  
	\sim \frac{(d-1)^n}{2^k} C_{n+k} \frac{(d-1)\left(k(d-2) + d\right)}{(d-2)^2},
	\label{WNKapprox}
\end{align}
valid for $k \ll \sqrt{n}$.
It might be possible to use a better approximation for $C_{n+k}^{[k+z]}$,
that would give us a result valid for higher $k$'s as well.

Next, we can calculate the average number of blocks that appear before the first extra letter for words that reduce to $k$ letters. Recalling \eqref{CNKapprox}, we have (here $W_{n,k}^{[b']}$ stands for words with $n$ pairs, $k$ extra letters and $b'$ leading blocks):
\begin{align}
	\bar{b}_k &= 
	\frac{\sum_{b'=0}^{n} b' W_{n,k}^{[b']}}{\sum_{b'=0}^{n} W_{n,k}^{[b']}}
		=
		\frac{(d-1)^{n+k} d^{-k} \sum_{b=k}^{n+k} (b-k) X^{b+k} C_{n+k}^{[b]}}{
			(d-1)^{n+k} d^{-k} \sum_{b=k}^{n+k} X^{b+k} C_{n+k}^{[b]}} 
		=
		\frac{\sum_{b=k}^{n+k} (b-k) X^b C_{n+k}^{[b]}}{
			\sum_{b=k}^{n+k} X^b C_{n+k}^{[b]}} 
			\label{averageBLOCKSwnk}\\
	&\sim \frac{\sum_{b=k}^{n+k} (b-k) 2b \left(\frac{X}{2}\right)^b }{ 
			\sum_{b=k}^{n+k} 2b\left(\frac{X}{2}\right)^b }
			\sim \frac{\sum_{z=0}^{n} z(z+k) x^z }{ 
			\sum_{z=0}^{n} (z+k)x^z }. \nonumber
\end{align}
Using approximate results for the sums, we obtain
\begin{align}
	\bar{b}_k &\sim \frac{\frac{x(1+x)}{(1-x)^3} + k \frac{x}{(1-x)^2}}{\frac{x}{(1-x)^2} + \frac{k}{1-x}}
	= \frac{x}{1-x} \left(
	\frac{1+x+k(1-x)}{x+k(1-x)}
	\right).
\end{align}
For $k=0$, this result is consistent with \eqref{averageBLOCKSwn}.
For $d=3$, we have $x=\frac{3}{4}$, so this translates to
\begin{align}
	\bar{b}_k^{(3)} \sim \frac{3(k+7)}{(k+3)} = 3 + \frac{12}{k+3}.
\end{align}

\subsubsection{Counting pairs in fully reducible PF model words}

For numerics, we can use the recursive formulas \eqref{pairsPFrecursive} and \eqref{pairsPFKrecursive} to count the pairs. However, we want to prove analytical bounds, so we have to do more work. 
Using the block-counting approach, we will now be able to derive the leading orders, which will not be enough to resolve the gap. For that, more work is needed. 

The number of pairs in PF model words that fully reduce is
\begin{align}
	\varphi_n = \sum_{b=1}^{n} \bigg(
	\underbrace{bd\,W_{n-1}^{[b-1]}}_{\textrm{peaks from 1-pair blocks}}
	+ \underbrace{b \sum_{i=1}^{n-b} d\tau_{i}^{(d-1)} W_{n-i-1}^{[b-1]}}_{\textrm{peaks within a nonempty block of size $i$}}
	+ \underbrace{\frac{b-1}{d}\,W_{n}^{[b]}}_{\textrm{valleys between blocks}}
	\bigg).
\end{align}
The PF words are made from a different number of blocks, i.e. $\sum_{b} W_n^{[b]} = W_n^{(d)}$. Recall that $\tau_i^{(d-1)} = (g i + h) C_i^{(d-1)}$, with $g=\frac{d}{2(d-1)}$ and $h=\frac{d-2}{2(d-1)}$.
Label $\bar{b}$ the average number of blocks in PF model words. It is a constant plus $\bigO\left(n^{-1}\right)$, as we have shown. We obtain
\begin{align}
	\varphi_n &= d \sum_{b=1}^{n} b\,W_{n-1}^{[b-1]} 
	- dh \sum_{b=1}^{n} b\,W_{n-1}^{[b-1]} 
	+ \sum_{b=1}^{n} 
	b \sum_{i=0}^{n-b+1} (gi+h) d C_{i}^{(d-1)} W_{n-i-1}^{[b-1]}
		+ \frac{\bar{b}-1}{d}\,W_{n}^{(d)} \nonumber\\
	&= d(1-h) \sum_{b=1}^{n} (b-1+1)W_{n-1}^{[b-1]}
	   + \sum_{b=1}^{n} b \left(g \frac{n-b}{b}+h \right) W_n^{[b]}
		+ \frac{\bar{b}-1}{d}\,W_{n}^{(d)}
	\\
	&= d(1-h) \left(\bar{b}+1\right) W_{n-1}^{(d)}
	+
	\left(gn+(h-g)\bar{b}+ \frac{\bar{b}-1}{d}\right) W_n^{(d)}  \nonumber
\end{align}
The average number of pairs is thus
\begin{align}
	f_n^{(d)} = \frac{\varphi_n^{(d)}}{W_n^{(d)}} 
	&= gn+(h-g)\bar{b}+ \frac{\bar{b}-1}{d}
	+ d(1-h)\left(\bar{b}+1\right) \frac{W_{n-1}^{(d)}}{W_n^{(d)}}
	\nonumber\\
	&\sim gn+(h-g)\bar{b}+ \frac{\bar{b}-1}{d}
	+  \frac{d(1-h)\left(\bar{b}+1\right)}{4(d-1)}
	 + \bigO\left(\frac{1}{n}\right),
\end{align}
with the definition of the average number of blocks $\bar{b}$ from \eqref{averageBLOCKSwn}.
This scales linearly with $n$, has a constant, and an error term of order $\oo{n^{-1}}$.
For $d=3$, we have $g=\frac{3}{4}$ and $h=\frac{1}{4}$ and $\bar{b}\approx 7$. This translates to
\begin{align}
	f_n \sim \frac{3n}{4} - \frac{7}{2} + 2 + \frac{9}{4} = \frac{3n}{4} + \frac{3}{4}
	= \frac{3N}{8} + \frac{3}{4}.
	\label{fn32}
\end{align}

\subsubsection{Counting pairs in PF model words with $k$ extra letters}

We can map the words $AaA\,\dots CcC\, X_1 x_1\dots X_k x_k$ of length $2n+k$ with $k$ extra letters $X_1,\dots,X_k$ to fully reducible PF model words of the form 
$AaA\,\dots CcC\, X_1 x_1 X_1\,\dots X_k x_k X_k$, with length $2(n+k)$, except the letters $X_1,\dots, X_k$ are fixed.
The number of subsequent identical letter pairs in the original words is the number of pairs within the $W_{n+k}$ words, minus the possible valleys between the extra letters ($X_i X_{i+1}$), plus the possible valleys on the interfaces of the type $x_i X_{i+1}$. 

Recall also that we can count these words as
\begin{align}
	W_{n,k}^{(d)} = \sum_{b=0}^{n} \frac{1}{d^k} W_{n+k}^{[k+b]},
\end{align}
where $b$ is the number of blocks before the first irreducible letter $X_1$.

We label $\tau'_0 = 1$ and $\tau'_{i\geq 1}=\tau_i$, while $\tau_0 = 0$, and observe that 
\begin{align}
			W_{n,k}^{[b]} = \frac{1}{d} W_{n+1,k-1}^{[b+1]} = \dots = \frac{1}{d^k} W_{n+k}^{[b+k]}.
\end{align}
which means
\begin{align}
			W_{n-(i+1),k}^{[b-1]} = \frac{1}{d} W_{n-i,k-1}^{[b]} = \dots = \frac{1}{d^k} W_{n-(i+1)+k}^{[b+k]}.
\end{align}
when taking $n\rightarrow n-(i+1)$ and $b\rightarrow b-1$.
Note that we can't use it as 
$W_{n-i,k-1}^{[0]} \neq sW_{n-(i+1),k}^{[-1]}$, because $W_{n,k}^{[0]}$ makes sense, while $W_{n-1,k}^{[-1]}$ does not.
Thus, counting things slowly and carefully, we get
\begin{align}
	\varphi_{n,k} &= \sum_{b=1}^{n} 
	\underbrace{
		b \sum_{i=0}^{n-b} d\tau'_i W_{n-(i+1),k}^{[b-1]}}_{\textrm{peaks within the first $b$ blocks}}
	+ 
	\sum_{b=0}^{n}
	\underbrace{
		k \sum_{i=0}^{n-b} \tau_i W_{n-i,k-1}^{[b]}}_{\textrm{peaks within the $k$ extra blocks}}
	\nonumber\\
	&+
	\sum_{b=1}^{n}
	\underbrace{\frac{b}{d}\,W_{n,k}^{[b]}}_{\textrm{v's between blocks}}
	+ 
	\sum_{b=0}^{n}
	\underbrace{\frac{k-1}{d-1} \left(W_{n,k}^{[b]}-W_{n,k-1}^{[b]}\right)}_{x_i X_{i+1}\textrm{ valleys, nonempty }x_i}
	.
\end{align}
Taking into account the values we should get for $i=0$, and $W_{n,k-1}^{[b]}=d W_{n-1,k}^{[b-1]}$ for $b\geq 1$, we get
\begin{align}
	\varphi_{n,k} &= \sum_{b=1}^{n} 
	\bigg(
			bd W_{n-1,k}^{[b-1]} - bdh W_{n-1,k}^{[b-1]}
		+
		b \sum_{i=0}^{n-b} 
				(gi+h)\, dC_i^{(d-1)} W_{n-(i+1),k}^{[b-1]}
	\bigg)\nonumber\\
	&+ \sum_{b=0}^{n} 
	\bigg(
			-kh W_{n,k-1}^{[b]}
			+	k \sum_{i=0}^{n-b} 
				(gi+h)\, C_i^{(h-1)} W_{n-i,k-1}^{[b]}
	\bigg)\nonumber\\
	&+ \frac{\bar{b}_k}{d}W_{n,k}^{(d)} 
			+ \frac{k-1}{d-1}\left(W_{n,k}^{(d)}-W_{n,k-1}^{(d)}\right).
\end{align}
The average length $i$ of an inside of a block in words with $k+b$ blocks with $n-b$ available pairs is $\frac{n-b}{b+k}$, so 
\begin{align}
	\varphi_{n,k} &= 
	 d(1-h) \sum_{b=1}^{n} (b-1+1) W_{n-1,k}^{[b-1]}
	 	+ \sum_{b=1}^{n} b \left(\frac{g(n-b)}{b+k}+h\right)
			  W_{n,k}^{[b]}\nonumber\\
		& - kh W_{n,k-1}^{(d)}
		+ k\sum_{b=0}^{n} \left(\frac{g(n-b)}{b+k}+h\right)
			  W_{n,k}^{[b]}
		+ \frac{\bar{b}_k}{d}W_{n,k}^{(d)} 
			+ \frac{k-1}{d-1}\left(W_{n,k}^{(d)}-W_{n,k-1}^{(d)}\right) \nonumber\\
&= 
	 \left(d(1-h)\left(\bar{b}_{k}+1\right)\right) W_{n-1,k}^{(d)}
	- kh W_{n,k-1}^{(d)}
	 	+ \sum_{b=0}^{n} \left(g(n-b)+(b+k)h\right)
			  W_{n,k}^{[b]} \nonumber\\
		&+ \frac{\bar{b}_k}{d}W_{n,k}^{(d)} 
			+ \frac{k-1}{d-1}\left(W_{n,k}^{(d)}-W_{n,k-1}^{(d)}\right) \nonumber\\
&= d(1-h)\left(\bar{b}_{k}+1\right) W_{n-1,k}^{(d)}
	 	+ \left(gn+hk + (h-g)\bar{b}_k + \frac{\bar{b}_k}{d} + \frac{k-1}{d-1}\right) W_{n,k}^{(d)} 
		\\
		&- \left(\frac{k-1}{d-1}+kh\right)W_{n,k-1}^{(d)}\nonumber
\end{align}
We could stop here and write the exact formula
\begin{align}
	f_{n,k}^{(d)} = \frac{\varphi_{n,k}^{(d)}}{W_{n,k}^{(d)}}
	&= gn + hk + \left(h-g + \frac{1}{d}\right)\bar{b}^{(d)}_k  + \frac{k-1}{d-1} \nonumber\\
	&+ d(1-h)\left(\bar{b}^{(d)}_{k}+1\right) \frac{W_{n-1,k}^{(d)}}{W_{n,k}^{(d)}} 
	- \left(\frac{k-1}{d-1}+kh\right)\frac{W_{n,k-1}^{(d)}}{W_{n,k}^{(d)}},
\end{align}
recalling the definition of $\bar{b}_k^{(d)}$ from \eqref{averageBLOCKSwnk}.

On the other hand, we can continue expressing $W_{n,k-1}^{(d)}$ in terms of $W_{n,k}^{(d)}$ as follows:
\begin{align}
	W_{n,k-1}^{(d)} = \sum_{b=0}^{n} W_{n,k-1}^{[b]} 
	= W_{n,k-1}^{[0]} + \sum_{b=1}^{n} W_{n,k-1}^{[b]}
	= W_{n,k-1}^{[0]} + \sum_{b=1}^{n} d W_{n-1,k}^{[b-1]}
	= C^{(d-1)}_{n,k-1} + d W_{n-1,k}^{(d)}
\end{align}
Because we know $W_{n,1}^{(d)} = d^{-1} W_{n+1}^{(d)}$ and
$W_{n,2}^{(d)} = d^{-2} \left(W_{n+2}^{(d)}- d(d-1)^{n+1}C_{n+1}\right)$,
this implies 
\begin{align}
	W_{n,1}^{(d)}= \frac{d^{-1} W_{n+1}^{(d)}}{d^{-2} \left(W_{n+2}^{(d)}- d(d-1)^{n+1}C_{n+1}\right)}
							\,W_{n,2}^{(d)}
				=  \frac{d}{4(d-1) - \frac{(d-2)^2}{d-1}}\,W_{n,2}^{(d)} \sim \frac{2}{5} W_{n,2}^{(d)}.
\end{align}
where made the approximation $W_{n+1}^{(d)} \sim \frac{d(d-1)}{(d-2)^2} C^{(d-1)}_{n+1}$ \eqref{Rnfirst} in the last step for $d=3$.

We can see that the leading term in $f_{n,2}$, scaling with $n$, has a prefactor exactly $g$, which agrees with the numerics.
How about the constant term? We need the ratio of $W_{n-1,k}^{(d)}$ to $W_{n,k}^{(d)}$, which is $4(d-1)=8$ for $d=3$ with error of order $\oo{n^{-1}}$. Thus, we have for $d=3$ and $k=2$:
\begin{align}
	f_{n,2} = \frac{\varphi_{n,k}}{W_{n,k}^{(3)}}
	&\sim gn + hk + \left(h-g + \frac{1}{d}\right)\bar{b}_k + \frac{k-1}{d-1}
	+ \frac{d(1-h)\left(\bar{b}_{k}+1\right)}{4(d-1)} 
	 - \frac{2}{5}\left(\frac{k-1}{d-1}+ kh\right)
	 \nonumber\\
	&\sim
		\frac{3n}{4} + \frac{1}{2} - \frac{\bar{b}_k}{6} + \frac{1}{2}
		+ \frac{3\times \frac{3}{4} (\bar{b}_k+1)}{8} - \frac{2}{5}
 \sim
		\frac{3n}{4} + 1 - \frac{9}{10}  + \frac{9}{5} - \frac{2}{5}
		= \frac{3n}{4} + \frac{3}{2}. \label{fn32with2extra}
\end{align}
using $\bar{b}_2\approx 5.4 = \frac{27}{5}$.
Recall that this is for a chain with $n$ pairs and 2 irreducible letters.
Thus, for a chain with length $N=2n+2$, we get $\frac{3N}{8}+\frac{3}{4}$, just what we got in \eqref{fn32} for the average number of pairs in irreducible words with length $N = 2n$
This means the leading two orders are the same. 

Thus, we calculated the leading two orders for the average number of letter pairs for fully reducible words and words with 2 irreducible letters, on a chain of the same length. However, unlike the $d=2$ case in Section~\ref{sec:countPF2}, we did not get a proof of an $\Theta\left(N^{-1}\right)$ gap between them. We only have a numerical indication of this. For an actual proof, one would need to expand this analysis to another order, which seems a tough ordeal.

\section{Counting peaks, valleys and pairs via generation function asymptotics}
\label{sec:GFasymptopairs}

Finally, there is one more analytic combinatorics approach that we can use, allowing us to perform asymptotical counting of peaks, valleys in Dyck paths, and subsequent letter pairs in PF model words. This analytical approach works for subspaces with a constant $k$ extra brackets or irreducible letters. At the end of Section~\ref{subsec:GFAverageNoPairs}, we will also discuss the case of general $k$.

The procedure starts with building bivariate ($u,z$) ordinary generating functions (BOGF) 
\begin{align}
	\sum_{t\in\textrm{words}} u^{\textrm{cost}(t)}z^{|t|}.
\end{align}
for counting these objects.
In each summand, the power of the parameter $z$ denotes the length $|t|$ of a word,
while the power of the parameter $u$ is the ``cost'' of the word $t$ -- the number of times the object (peak, valley, pair) appears in it. 
To find the BOGFs, we rely on standard tools form analytic combinatorics -- products and sequences of simpler generating functions. 
We then differentiate the BOGF with respect to $u$. For each term $u^k$ in the series expansion of the BOGF, this gives us $k u^{k-1}$. When we evaluate this at $u=1$, we obtain an ordinary generating function for counting the objects (peaks, valleys, pairs) marked by the $u$'s in the original BOGF. For this, we need to work out the asymptotic scaling of the $z^n$ term. Finally, to get the average object count, we need to divide the result by the asymptotic result for the total word count, and work out the relevant terms up to our desired precision.

\subsection{Counting peaks via generating functions}

Let us derive the BOGF (bivariate ordinary generating function) for peaks and valleys in Dyck words. The power of $z$ in the expansion of the BOGF will indicate the the length of a word (the number of bracket pairs) and there will be an extra factor of $u$ for each peak. 
	
	\vskip6pt
	\noindent\textbf{The uncolored (monochromatic) case.} A Dyck word can be an empty string or $L\alpha R$ followed by a sequence (possibly empty) of similar blocks. We want to mark each peak $LR$, adding a factor of $u$. 
	In the symbolic notation of \cite{FlajoletSedgewick}, labeling the empty word by $\varnothing$, the recursion reads:
	\begin{align}
		\mathcal{P'}=u\varnothing+\mathcal{Z}\times \mathcal{P'}\times\SEQ(\mathcal{Z}\times \mathcal{P'}),
	\end{align}
	where $\mathcal{Z}$ is called an {\em atom} representing a matched pair ($L\ldots R$), and adds a factor of $z$, while $u$ is a {\em marker} of a peak.
	This gives us the desired behavior, up to the empty word, which does not have a peak. Thus, when we later solve the recursion and find the generating function, we will require a $1-u$ correction. 
	Let us translate this to a generating function equation:
	\begin{align}
		P'(u,z)=u+\frac{z\,P'(u,z)}{1-z\,P'(u,z)}
	\end{align}
	We solve this quadratic equation and receive the allowed solution:
	\begin{align}
		P'(u,z)=\frac{1-z+u\,z-\sqrt{(z-uz-1)^2-4 u z }}{2\,z},
	\end{align}
	which means after the correction for the empty word, we get 
	\begin{align}
		P(u,z) = P'(u,z) - u + 1.
	\end{align}
	With its help, we can count the peaks in all Dyck paths with length $2n$ (with $n$ bracket pairs) as
	\begin{align}
		[z^n]  \left. \frac{\partial P(u,z)} {\partial u} \right|_{u=1}, \label{BOGFpeaks}
	\end{align}
	the coefficient of the $z^n$ term in this formula's power series expansion.
	
	\vskip6pt
	\noindent\textbf{Dyck paths with $d$ colors.} To take into account the possibility of coloring the bracket (up/down step) pairs with $d$ colors, we only need to adding a factor of $d$ to each $z$.
	This gives us the generating function:
	\begin{align}
		P'^{(d)}(u,z)=\frac{1-dz+duz-\sqrt{(dz-duz-1)^2-4duz}}{2\,dz}.
	\end{align}
	We again have to correct the behavior of the generating function for an empty word, adding $1-u$, which translates to 
	the final BOGF function counting peaks in $d\geq1$ colored Dyck words:
	\begin{align}
		P^{(d)}(u,z)=\frac{1+dz-duz-\sqrt{(dz-duz-1)^2-4duz}}{2\,dz}.
	\end{align}
	To get the number of peaks, we proceed as in \eqref{BOGFpeaks}.
	

	\vskip6pt
	\noindent\textbf{Dyck paths with $k$ extra steps and their peaks.}
	Counting the number of peaks in Dyck words with $k$ extra steps is quite easy. Each such word can be decomposed as $v_1 X_1v_2\ldots X_kv_{k+1}$, with possibly empty Dyck paths $v_i$. The extra steps can't be involved in peaks, so the BOGF for counting peaks in Dyck words with $k$ extra steps is just a Cartesian product of $k+1$ BOGFs counting the number of peaks in Dyck words $P^{(d)}(u,z)$.
	\begin{align}
		P^{(d),k}(u,z) = \left(P^{(d)}(u,z)\right)^{k+1}. \label{BOGFpeaksK}
	\end{align}
	There are $n-k/2$ pairs in words with length $2n$ and $k$ extra steps.
	Thus, to get the number of peaks in such words, we need to read off the coefficient of $z^{n-k/2}$ 
	in the series expansion of $\partial P^{(d),k}(u,z)/\partial u|_{u=1}$.

\subsection{Counting valleys via generating functions}

	\textbf{The uncolored (monochromatic) case.}\footnote{One could also simply count the number of valleys in uncolored Dyck paths simply as $\#$peaks$-1$, which agrees with what we derive here: $V(u,z) = P(u,z)/u$.}
	A Dyck path can be empty, or composed in block form as $L\alpha R$, possibly concatenated with other such blocks. We add a factor of $u$ for each such block, concatenating them necessarily forms a valley ($L\ldots \underline{RL}\ldots R$). The string $\alpha$ is a Dyck path, satisfying the same relations. In symbolic notation, the recurrence equation reads
	\begin{align}
		V=\varnothing+\mathcal{Z}\times V\times\SEQ(u\mathcal{Z}\times V)
	\end{align}
	and translates to the generating function equation
	\begin{align}
		V(u,z)=1+\frac{z\,V(u,z)}{1-z\,u\,V(u,z)}.
	\end{align}
	We can solve this quadratic equation and obtain the BOGF
	\begin{align}
		V(u,z)=\frac{1-z+uz-\sqrt{(z-uz-1)^2-4uz}}{2uz}.
		\label{BOGFvalley}
	\end{align}
	
	\vskip6pt
	\noindent
	\textbf{The $d$-color case.} Now, the steps (open/closed brackets) have one of $d$ colors. A valley between blocks is formed only when the two enclosing pairs have the same color $L_i\ldots \underline{R_i L_i}\ldots R_i$. This happens in only one of the $d$ possible cases for the starting color of the next block. This allows us to write the symbolic relation
	\begin{align}
		V^{(d)}=\varnothing+(\mathcal{Z}_1+\ldots+\mathcal{Z}_d)\times V^{(d)}
			\times
				\SEQ\left((\mathcal{Z}_1+\ldots+u\mathcal{Z}_i+\ldots+\mathcal{Z}_d)\times V^{(d)}\right),\label{DyckValleysSymbolic}
	\end{align}
	where we labeled w.l.o.g. $\mathcal{Z}_i$ with $u$, as one of the atoms has to form a valley.
	Therefore, we get
	\begin{align}
		V^{(d)}(u,z) &= 1 + \frac{d z\,V^{(d)}(u,z)}{1-z(u+d-1)\,V^{(d)}(u,z)}.\label{RecValleysD}
	\end{align}
	The solution for the generating function is thus
	\begin{align}
		V^{(d)}(u,z)=\frac{1-z+uz-\sqrt{(z-uz-1)^2-4(uz+(d-1)z)}}{2(uz+(d-1)z)}.
	\end{align}
	
	\vskip6pt
	\noindent
	\textbf{Counting valleys in Dyck paths with $k$-extra steps}
	Each such path can be written as $v_1 X_1v_2 X_2\ldots X_kv_{k+1}$.
	In the {\em uncolored} case, whenever there is a nonempty block $v_i$ preceding an extra letter $X_i$, an additional valley is formed. This allows us to write the equation:
	\begin{align}
		V^{k}(u,z)= \left( u(V(u,z)-1)+1 \right)^k V(u,z),
	\end{align}
where the $V(u,z)-1$ makes sure that empty blocks (the term 1 in the expansion of $V(u,z)$) don't contribute a valley (a factor of $u$).
		
	In the {\em $d$-color} case, only a fraction of nonempty colored Dyck walks form additional valleys before the extra steps. The colors match in only $1/d$ of the words and don't match in $(d-1)/d$ of them. Thus,
	\begin{align}
		V^{(d),k}(u,z)
		= \left(\frac{u+(d-1)}{d}\left(V^{(d)}(u,z)-1\right)+1\right)^kV^{(d)}(u,z). 
		\label{BOGFvalleyKD}
	\end{align}
	To get the number of valleys in such words, we need to read off the coefficient of $z^{n-k/2}$ 
	in the series expansion of $\partial V^{(d),k}(u,z)/\partial u|_{u=1}$.

\subsection{Counting pairs in the PF model via generating functions}
Thanks to the properties of the pair-flip (PF) model words and their relationship to $d-1$ colored Dyck paths, 
the BOGFs we derived above will be helpful in finding those for counting neighboring letter pairs in PF words.

Each PF model word has a mountain profile as in Figure~\ref{fig:peakcount}c. 
The neighboring letter pairs correspond to peaks in this profile, 
as well as same-color valleys. We will thus create a BOGF counting these peaks as well as valleys.

First, we will do this for colored Dyck paths, starting with the symbolic relation for valleys  \eqref{DyckValleysSymbolic}.
Adding peak-counting is easy, adding a $u$ factor for each empty word in the recursive relation:
\begin{align}
	PV'=u \varnothing+(\mathcal{Z}_1+\ldots+\mathcal{Z}_d)\times PV'^{(d)}
	\times		\SEQ\left((\mathcal{Z}_1+\ldots+u\mathcal{Z}_i+\ldots+\mathcal{Z}_d)\times PV'^{(d)}\right).
\end{align}
This works except for the length-0 empty word $\varnothing$, which does not contain a pair, and we'll correct that below in \eqref{PVcorr}.
Let us rewrite the symbolic representation into a generating function equation:
\begin{align}
	PV'^{(d)}(u,z) &= u + \frac{d z\,PV'^{(d)}(u,z)}{1-z(u+d-1)\,PV'^{(d)}(u,z)}.\label{RecPeaksValleysD}
\end{align}
We can solve this equation and correct the peak-and-valley counting for the empty word as well with
\begin{align}
	PV'^{(d)}(u,z)&=\frac{1+dz(u-1)-uz+u^2z-\sqrt{(1+z(u-1)(d+u))^2-4uz(d+u-1)}}{2z(d+u-1)},\\
	PV^{(d)}(u,z)&=PV'^{(d)}(u,z)-u+1. \label{PVcorr}
\end{align}

With this in hand, we can address the PF model and its neighboring letter pair counting. Any PF word can be written as a sequence of blocks $A_i v_i A_i$, with letters $A_i\in\{1,\ldots ,d\}$ and $v_i$'s corresponding to $d-1$ colored Dyck paths \eqref{Dcorrespond}. 
Thanks to this mapping we can write the symbolic relation for counting pairs in PF words
\begin{align}
	P\!F\!P^{(d)} = \varnothing + (\mathcal{Z}_1+\ldots+\mathcal{Z}_d)\times PV'^{(d-1)}
	\times \SEQ\left( (\mathcal{Z}_1+\ldots+u\mathcal{Z}_i+\ldots+\mathcal{Z}_d)\times PV'^{(d-1)} \right),
\end{align}
Here, the term $\varnothing$ stands for the empty word.
The factor $(\mathcal{Z}_1+\ldots+\mathcal{Z}_d)\times PV'^{(d-1)}$ counts the peaks and valleys within the block $v_1$, as well as a possible pair $A_1A_1$, thanks to the $PV'^{(d-1)}$ giving one $u$ for the empty word. Finally, the $u$ in front of only one of the $d$ terms in the final factor accounts for a possible pair between subsequent blocks: $A_i v_i \underline{A_i A_{i+1}} v_{i+1} A_{i+1}$, when the letters $A_i A_{i+1}$ match.
This symbolic relation translates to the generating function
\begin{align}
	P\!F\!P^{(d)}(u,z)=1+\frac{dzPV'^{(d-1)}(u,z)}{1-z(u+d-1)PV'^{(d-1)}(u,z)}.
	\label{PFpaircountBOGF}
\end{align}

\vskip6pt
\noindent
\textbf{Pairs in PF words with $k$-extra letters.}
As above, our way to count neighboring letter pairs in PF words will be to count the peaks and valleys in their mountain profiles. 

First, we write down the generating function counting both peaks and valleys in Dyck paths with $k$-extra steps. Each such path can be written as $v_1 X_1 v_2 X_2\ldots X_k v_{k+1}$. An extra letter $X_i$ can possibly contribute a valley if the preceding substring $v_i$ is nonempty and its last letter matches $X_i$. Similarly as for the valleys in colored Dyck paths \eqref{BOGFvalleyKD} we can write the BOGF for peaks and valleys together
\begin{align}
PV^{(d),k}(u,z)
=\left(\frac{u+(d-1)}{d}\left(PV^{(d)}(u,z)-1\right)+1\right)^kPV^{(d)}(u,z).
\label{PFpairsKimplicit}
\end{align}

With this in hand, we turn to the PF words with $k$ irreducible letters $X_1,\ldots X_k$.
We can rewrite them as $wX_1v$ with a fully reducible PF word $w$ and 
a word $v$ corresponding to a Dyck walk with $k-1$ extra steps (corresponding to the irreducible letters $X_2,\ldots,X_k$) and $d-1$ colors. Similarly to peaks and valleys in Dyck paths with extra terms, there can be an additional pair between a nonempty $w$ and the irreducible letter $X_1$ -- this happens in $1/d$ of such words. This lets us use the generating function for counting peaks and valleys in Dyck walks with extra terms to build the BOGF
\begin{align}
P\!F\!P^{(d),k}(u,z) = 
\left(\frac{u+(d-1)}{d}\left(P\!F\!P^{(d)}(u,z)-1\right)+1\right)PV^{(d-1),k-1}.\label{GFPFP}
\end{align}
Note that with some effort, one can verify that this agrees with \eqref{PFpaircountBOGF} for $k=0$ as well, when we formally extend the definition of $V^{(d),k}$ \eqref{BOGFvalleyKD} also to $k=-1$.

	\subsection{The average number of pairs in PF model words}
	\label{subsec:GFAverageNoPairs}

	We can now utilize the BOGF to derive the average number of pairs in PF words of total length $2n$ including the $k$ irreducible letters.
	First, we shift the BOGF by $z^{k/2}$ so the extra terms are added to the total length. Next, we form the cumulative cost function counting neighboring letter pairs $\partial z^{k/2}P\!F\!P^{(d),k}(u,z)) / \partial u |_{u=1}$
	to get the total count of the pairs, and read off the coefficient of the term 
	$z^{n}$. We then obtain the average cost by dividing it by the total number of words obtained from the BOGF (setting $u=1$) or using our previously derived generating function \eqref{WKgenerate}.
	
\begin{figure}
\begin{center}
		\begin{tabular}{| c | r | r r r r r r r r r r |}
			\hline
			\multicolumn{2}{|r|}{$n$} & 0 & 1 & 2 & 3 & 4 & 5 & 6 & 7 & 8 & 9\\
			\hline
			\multirow{2}{*}{$k=0$}&\#words & 1 & 3 & 15 & 87 & 543 & 3543 & 23823 & 163719 & 1143999 & 8099511\\
			& \#pairs & 0 & 3 & 27 & 225 & 1827& 14661& 116919 & 929097 & 7367355 & 58343949 \\
			\hline
			\multirow{2}{*}{$k=2$}&\#words & 0 & 1 & 7 & 47 & 319 & 2199 & 15375 & 108807 & 777919 & 5610167 \\
			& \#pairs & 0 & 0 & 9 & 105 & 987 & 8613 & 72567 & 599625 & 4896315 & 39673869\\
			\hline
			\multirow{2}{*}{$k=4$}&\#words & 0 & 0 & 1 & 11 & 95 & 759 & 5871 & 44743 & 338623 & 2555063\\
			& \#pairs & 0 & 0 & 0 & 15 & 231 & 2565 & 25047 & 228969 & 2013435 & 17269773\\
			\hline
			\end{tabular}

	\caption{The first 10 numerical values of the total number of $d=3$ fully reducible PF model words with fixed word length $N=2n$, and the number of subsequent letter pairs these words have.
	We also list these values for words with $k=2$ and $k=4$ irreducible letters.}
	\label{fig:PFpairnumerics}
	\end{center}
\end{figure}
	
	We will now focus on words of length $2n$ in the $d=3$ color PF model with a constant number $k$ of irreducible letters (including fully reducible words with $k=0$),
	and prove that the average number of neighboring letter pairs decreases with growing $k$. 
	In particular, we will show that for the uniform superposition in the subspace with $k$ irreducible letters, the average number of pairs is $\Theta\left(k^2/n\right)$ smaller than in the fully reducible subspace.
	Thanks to our numerical investigation (see Figure~\ref{fig:pairnumerics}), we believe this behavior persists for all $k$.

\vskip6pt
	\noindent
	\textbf{Counting the words with $k\in O(1)$ irreducible letters.} For constant $k$ (including $k=0$), we can obtain an asymptotic approximation by singularity analysis \cite{FlajoletOdlyzkoSingularityAnalysis, FlajoletSedgewick}.
	The essential singularity of the OGF sits at $z=\frac{1}{8}$.
	We are interested in low half-integer powers of $1-8z$ that dominate the asymptotic scaling. For $d=3$ and labeling $a=\left(1-\sqrt{1-8z}\right)^k$, the GF counting PF words with $k$ irreducible letters \eqref{WKgenerate} reads 
	\begin{align}
		\frac{a-3a\sqrt{1-8z}}{4^{k-1}z^{k/2}8(9z-1)},
	\end{align}
	when we include a $\sqrt{z}$ factor for each extra letter, so that the coefficient of the $z^n$ term counts the number of words of length $2n$, including the $k$ irreducible letters.
	Note that for chains with even length, only even $k$ are possible.
	
	For a constant $k$, this function is amenable to singularity analysis via the closure properties of $\Delta$-analytic functions \cite{FlajoletSedgewick}.
	We use the binomial theorem to find the low power terms for both odd and even powers of $\sqrt{1-8z}$ in $a$:
	\begin{align}
		a = (1-\sqrt{1-8z})^k
		=&\sum_{i=0}^{k}\binom{k}{i}\left(-\sqrt{1-8z}\right)^{i}\nonumber\\
		=& 1-k\sqrt{1-8z}+\binom{k}{2}(1-8z)-\binom{k}{3}(1-8z)^{3/2} \\
		&+\binom{k}{4}(1-8z)^{2}-\binom{k}{5}(1-8z)^{5/2}+\ldots \nonumber
	\end{align}
	We plug the expansion of $a$ to the generating function and consider only the half integer powers of $1-8z$.
	\begin{align}
		\frac{1}{4^{k-1}z^{k/2}8(9z-1)}\sum_{i=0}^{k/2}-\left(\binom{k}{2i+1}-3\binom{k}{2i}\right)
			(1-8z)^{(2i+1)/2}.
	\end{align}

	Moreover, we consider them only up to error $O\left((1-8z)^{7/2}\right)$.
	We expand this function around $z=1/8$ and using computer algebraic manipulation\footnote{A tedious calculation of the Taylor series, easily handled by e.g. Mathematica.}, receive
	the expression 
	\begin{align}
		 &-2(k+3)\,(1-8z)^{1/2}
		-\frac{k^3+9k^2+56k+162}{3}\,(1-8z)^{3/2} \\
		&-\frac{k^5+15k^4+200 k^3+1740k^2+10104k+29160}{60}\,(1-8z)^{5/2}
		+ O\left((1-8z)^{7/2}\right). \nonumber
	\end{align}
	with a $2^{-k/2+1}$ prefactor.

	We then use the standard function scale \cite{FlajoletSedgewick,FlajoletOdlyzkoSingularityAnalysis} and find the asymptotic approximation for counting the $d=3$ color PF words of length $2n$ with $k\in O(1)$ irreducible letters:
	\begin{align}
		\frac{8^n(k+3)}{2^{k/2-1}\sqrt{\pi n^3}}
			\bigg( 1 &-\frac{2k^3+18k^2+109k+315}{8(k+3)n}
			\label{WNKasymptotic}\\
							& + \frac{4k^5+60k^4+740k^3+6420k^2+37081k+106995 }{128(k+3)n^2}
							+O((k^6+1) n^{-3}) \bigg).
						\nonumber
	\end{align}
Note that we put in $(k^6+1)$ instead of $k^6$ in the final scaling of the error term so that it works also for $k=0$.
	

\vskip6pt
	\noindent
	\textbf{The total number of neighboring letter pairs for words with $k\in O(1)$ irreducible letters.} 
	We build the cumulative cost function for counting neighboring letter pairs in such words using \eqref{GFPFP}:
\begin{align}
	\left.\frac{\partial}{\partial u}z^{k/2}P\!F\!P^{(d),k}(u,z)\right|_{u=1}.
\end{align}
Recalling the generating functions for PF \eqref{W0generate} and Dyck words \eqref{ColorCGF},
this results in
\begin{align}
	z^{k/2} C^{(d-1),k-2}(z)\bigg(&W^{(d)}(z)\left(k\frac{\partial}{\partial u}PV^{(d-1)}(u,z)+\frac{k-1}{d-1}\left(C^{(d-1)}(z)-1\right) \right) \label{PairsPFcount}\\
	&+ z^{k/2}C^{(d-1)}(z)\left(\frac{\partial}{\partial u}P\!F\!P^{(d)}(u,z)+\frac{1}{d}\left(W^{(d)}(z)-1\right)\right)
	\bigg)\bigg|_{u=1},\nonumber 
\end{align}
For the $d=3$ color PF model, we label $b=(1-\sqrt{1-8z})^{k-1}$ and continue as
\begin{align}
	\frac{3}{2^{2k-3}}\,\frac{\left(1-k+4z(3k+1)\right)b+(k-1)b\sqrt{1-8z}}{z^{k/2-1}(1+3\sqrt{1-8z})^2\sqrt{1-8z}}.\label{PairCumulativeCostGF}
\end{align}
For a constant $k$, this function is $\Delta$-analytic with an algebraic singularity and therefore amenable to singularity analysis, thanks to the closure properties of $\Delta$-analytic functions.
Singularity analysis tells us the dominating contributions come from the essential singularity at $z=\frac{1}{8}$. 
We now need to expand the GF in the vicinity of this point, focusing on half-integer, low powers of $1-8z$. To find the expansion up to $(1-8z)^{5/2}$, it is enough to consider the nominator up to $(1-8z)^{3}$:
\begin{align}
	\left(1-k+4z(3k+1)\right)+\sum_{i=1}^6\left(\left(1-k+4z(3k+1)\right)\binom{k-1}{i}+(k-1)\binom{k-1}{i-1}\right)(-\sqrt{1-8z})^i.\label{PairCumulativeCostGFNominatorCut}
\end{align}
Thus, expanding \eqref{PairCumulativeCostGF} with the trimmed nominator \eqref{PairCumulativeCostGFNominatorCut} around the essential singularity, we receive
\begin{align}
	&(k+3)(1-8z)^{-1/2}
	+
	\half\,(k^3+9k^2+50k+144)(1-8z)^{1/2}\nonumber\\
	-&\frac{1}{24}\,(k^5+15k^4+172k^3+1488k^2+8584k+24768)(1-8z)^{3/2}\\
	+&\frac{1}{720}\,(k^7+21k^6+382k^5+5520k^4+63664k^3+551064k^2+3180048k+9175680)(1-8z)^{5/2},
	\nonumber\\
	+&O((1-8z)^{7/2})\nonumber
\end{align}
with a $3/2^{k/2+1}$ prefactor.
Finally, we use the standard function scale to receive the asymptotic expression counting the neighboring letter pairs in the $d=3$ PF model words of length $2n$ with $k$ irreducible letters:
\begin{align}
	3\,\frac{8^n(k+3)}{2^{k/2+1}\sqrt{\pi n}}&\bigg(1
	-\frac{2k^3+18k^2+101k+291}{8(k+3)n} \label{PFpairsKfinal}\\
	&+\frac{4k^5+60k^4+676k^3+5844k^2+33737k+97347}{128(k+3)n^2}\nonumber\\
	&-\frac{8k^7+168k^6+2876k^5+41460k^4+478502k^3+4142022k^2+23902749k+68968755}{3072(k+3)n^3}\nonumber\\
	&+O((k^8+1)n^{-4})\bigg).\nonumber
\end{align}

\vskip6pt
	\noindent
	\textbf{The average number of pairs for words with $k\in O(1)$ irreducible letters.} 
	Finally, we can divide \eqref{PFpairsKfinal} and \eqref{WNKasymptotic}  
	to get the average number of subsequent letter pairs in $d=3$ PF model words with length $N=2n$. For large $n$, after some calculation we receive
\begin{align}
	\frac{3n}{4}+\frac{3}{4}-\frac{3(k^3+9k^2+50k+144)}{16(k+3)n}+O((k^4+1)n^{-2}).
	\label{AVGfinalConst}
\end{align}

A few observations are in place.
First, this result agrees with the first two terms calculated in \eqref{fn32} and \eqref{fn32with2extra} by a direct method not involving analytic combinatorics. 
Second, this function is monotonously decreasing with $k$ from its maximum at $k=0$, for $k\in \mathbb{N}_0$. Therefore, the uniform superposition of fully reducible PF words has the maximum expected number of pairs out of all the uniform superpositions of PF words with $k=O(1)$ extra letters. Moreover, the average number of pairs is approximately $\frac{3k^2}{16n}$ smaller in the subspaces with $k$ irreducible letters. In Figure~\ref{fig:pairnumerics}, we showcase our numerics indicating that this behavior persists beyond the region of constant $k$, where expression \eqref{AVGfinalConst} is valid.
Third, this means when used as a perturbation, pair-counting can be used as a degeneracy-breaking term for the PF model, selecting a nearly-uniform superposition of states from the fully-reducible subspace as the unique ground state.

\vskip6pt
	\noindent
	\textbf{Discussion of asymptotics for general $k$.} 
	It was relatively straightforward to find the asymptotic scaling of these expression for a fixed $k$ and $n\rightarrow\infty$.
	However, we would like to know the uniform asymptotics up to $n^{-2}$ of the full spectrum for pair-counting calculations, keeping careful track of error terms.
	
	An important component of our generating functions are powers of another function \eqref{GFPFP} with an algebratic singularity. There are various approaches to finding their bivariate asymptotics
	\begin{align}
	f_{k,n}=[z^n] f(z)^k.
	\end{align}
	One must be very careful about the error terms, understanding the regimes of $k/n$ that they apply to. Drmota studied this \cite{DRMOTA1994139} and provided a uniform asymptotic scaling with full asymptotic expansion for $k/n\leq 1-\epsilon$; this is a complicated calculation even for the first few error terms. In our case, we know how to apply singularity analysis for bounded $k$. Semi-large powers $k=\Theta(\sqrt{n})$ were studied in \cite{BanderierAiry}, with limit law scaling as a Rayleigh probability density function (for a square root type singularity). They were also studied in a ``Lagrangian framework'', using the saddle point method $k\in [\epsilon n,\, (1-\epsilon)n ]$ (coalescence of saddles, double saddle point). Finally, for the large powers $k=\Theta(n)$, where both $k,n$ goes to infinity, we refer the reader to \cite{GARDY1995189} and also applications of saddle point method in \cite{FlajoletSedgewick}. We leave the analysis of these generating functions for our further studies.

\subsection{The variance of pairs distribution in PF model words}
The BOGF $P\!F\!P^{(d),k}$ holds the complete information about the number of words with a particular pair of neighboring pairs. We can thus use it to calculate the variance of this distribution. Working out what the second derivative of the BOGF does, we can write
\begin{align}
	\langle \#^2 \rangle_k^{(d)} - \left(\langle \#\rangle_k^{(d)}\right)^2
	=	
	\frac{[z^n]\left.\partial^2/\partial u^2 z^{k/2}P\!F\!P^{(d),k}(u,z)\right|_{u=1}}
	{[z^n]P\!F\!P^{(d),k}(1,z)}+\langle \# \rangle_k^{(d)}-\left(\langle \#\rangle_k^{(d)}\right)^2.
\end{align}
A tedious calculation results in 
\begin{align}
	\langle \#^2 \rangle_k^{(d)} - \left(\langle \#\rangle_k^{(d)}\right)^2 = \frac{13n}{32}+\frac{1}{64}+O((k^2+1)n^{-1}). \label{pairvariance}
\end{align}
The distribution (pair-count histogram) is thus $\bigO\left(\sqrt{n}\right)$ wide around the average
value \eqref{AVGfinalConst}.


\section{Variance of the Schmidt coefficients for the qubit PF model}
\label{sec:d2entropyapp}
In this Section, we calculate the variance for the distribution coming from the Schmidt coefficients for the $d=2$ PF model, 
\begin{align}
	p_k = 
	\frac{\binom{2n}{n-k}^2}{\binom{4n}{2n}},
\end{align}
where the $k=0$ term appears once, and all $k\geq 1$ terms appear in the probability distribution twice.
It is required in the proof of 
\eqref{eq:QubitEntropyVar}.

\begin{align}
   \Var(\{p_k\}) 
	&= \sum_{k=-n}^n k^2 p_k=\sum_{k=-n}^n k^2 \frac{{2n\choose n-k}^2}{{4n \choose 2n}} 
	=\frac{1}{\binom{4n}{2n}}\sum_{k=0}^{2n}(k-n)^2\binom{2n}{k}^2 \nonumber\\
&=\frac{1}{\binom{4n}{2n}}\left(\sum_{k=0}^{2n}k^2\binom{2n}{k}^2-2n\sum_{k=0}^{2n}k\binom{2n}{k}^2+n^2\sum_{k=0}^{2n}\binom{2n}{k}^2\right).
\label{d2var1}
\end{align}
We claim that
\begin{align}
    \Var(\{p_k\})=\frac{n^2}{4n-1}. \label{d2var}
\end{align}
To show this, let us express each term in \eqref{d2var1} separately. 
Vandermonde's identity $\sum_{k=0}^n\binom{r}{k}\binom{s}{n-k}=\binom{r+s}{n}$ implies $\sum_{k=0}^{m}\binom{m}{k}=\binom{2m}{m}$ and
\begin{align} 
\sum_{k=0}^{2n}\binom{2n}{k}^2=\binom{4n}{2n}.
\end{align}
Since $\binom{n}{k}=\frac{n}{k}\binom{n-1}{k-1}$ we can remove the $k$ prefactor in the next term in \eqref{d2var1} as follows:
\begin{align}
 \sum_{k=0}^{2n}k\binom{2n}{k}^2
    &= 2n\sum_{k=0}^{2n}\binom{2n-1}{k-1}\binom{2n}{k}
		= 2n\sum_{k=0}^{2n-1}\binom{2n-1}{k}\binom{2n}{k+1} \nonumber\\
    &= 2n\sum_{k=0}^{2n-1}\binom{2n-1}{k}\binom{2n}{2n-1-k}
    = 2n\binom{4n-1}{2n-1}=n\binom{4n}{2n}.
\end{align}
using Vandermonde's identity in the second-to-last step.
Finally, we can remove the $k^2$ as 
\begin{align}
\sum_{k=0}^{2n}k^2\binom{2n}{k}^2&=4n^2 \sum_{k=0}^{2n} \binom{2n-1}{k-1}^2=4n^2 \sum_{k=0}^{2n-1} \binom{2n-1}{k}^2=4n^2\binom{4n-2}{2n-1}=\frac{4n^3}{4n-1}\binom{4n}{2n}.
\end{align}
Plugging these expressions into \eqref{d2var1} gives us \eqref{d2var}.

\end{document}